\documentclass[sigconf,balance=false]{acmart}
\usepackage{popets}
\usepackage[toc,page,header]{appendix}
\usepackage{graphicx}
\usepackage{mathdots}
\usepackage{amsmath,amsfonts}
\usepackage{algorithmic}
\usepackage{textcomp}
\usepackage{xcolor}

\usepackage{hyperref} 
\usepackage{hyperxmp} 
\usepackage{balance} 
\usepackage{csquotes}
\usepackage{mdframed}
\usepackage{lipsum}
\usepackage[normalsize]{subfigure}
\usepackage{wrapfig}
\usepackage{algorithmic}
\usepackage[ruled, vlined, commentsnumbered, linesnumbered]{algorithm2e}

\usepackage{lipsum}
\usepackage{amsmath}
\usepackage{tcolorbox}

\usepackage{xcolor,colortbl}

\definecolor{Gray}{gray}{0.85}
\definecolor{LightCyan}{rgb}{0.88,1,1}

\newcolumntype{a}{>{\columncolor{Gray}}c}
\newcolumntype{b}{>{\columncolor{white}}c}

\newtheorem{theorem}{\bf{Theorem}}[section]
\newtheorem{definition}{\bf{Definition}}[section]
\newtheorem{lemma}{\bf{Lemma}}[section]

\newtheorem{proposition}[theorem]{\bf{Proposition}}

\newenvironment{remark}[1][Remark]{\begin{trivlist}
\item[\hskip \labelsep {\bfseries #1}]}{\end{trivlist}}

\newcommand{\DEL}[1]{\iffalse #1 \fi}
\newcommand{\rd}{\color{black}}

\newcommand{\revisiondone}{\color{black}}
\newcommand{\bl}{\color{black}}

\usepackage{popets}
\setcopyright{popets}
\copyrightyear{YYYY}
\acmYear{YYYY}
\acmVolume{YYYY}
\acmNumber{X}
\acmDOI{XXXXXXX.XXXXXXX}
\acmISBN{}
\acmConference{Proceedings on Privacy Enhancing Technologies}
\begin{document}
\title{Time-Efficient Locally Relevant Geo-Location Privacy Protection
\vspace{-0.00in}
}

\author{Chenxi Qiu}
\orcid{}
\affiliation{%
  \institution{University of North Texas}
  \city{} 
  \state{} 
  \country{} 
}
\email{chenxi.qiu@unt.edu}

\author{Ruiyao Liu}
\orcid{}
\affiliation{%
  \institution{University of North Texas}
  \city{} 
  \state{} 
  \country{} 
}
\email{ruiyaoliu@my.unt.edu}

\author{Primal Pappachan}
\orcid{}
\affiliation{%
  \institution{Portland State University}
  \city{} 
  \state{} 
  \country{} 
}
\email{primal@pdx.edu}

\author{Anna Squicciarini}
\orcid{}
\affiliation{%
  \institution{The Pennsylvania State University}
  \city{} 
  \state{} 
  \country{} 
}
\email{acs20@psu.edu}

\author{Xinpeng Xie}
\orcid{}
\affiliation{%
  \institution{University of North Texas}
  \city{} 
  \state{} 
  \country{} 
}
\email{xinpengxie@my.unt.edu}

\begin{abstract}
Geo-obfuscation serves as a \emph{location privacy protection mechanism (LPPM)}, enabling mobile users to share obfuscated locations with servers, rather than their exact locations. This method can protect users' location privacy when data breaches occur on the server side since the obfuscation process is irreversible. To reduce the utility loss caused by data obfuscation, \emph{linear programming (LP)} is widely employed, which, however, might suffer from a polynomial explosion of decision variables, rendering it impractical in large-scale geo-obfuscation applications. 

In this paper, we propose a new LPPM, called \emph{\underline{L}ocally \underline{R}elevant \underline{Geo}-obfuscation (LR-Geo)}, to optimize geo-obfuscation using LP in a time-efficient manner. This is achieved by confining the geo-obfuscation calculation for each user exclusively to the \emph{locally relevant (LR)} locations to the user's actual location. Given the potential risk of LR locations disclosing a user's actual whereabouts, we enable users to compute the LP coefficients locally and upload them only to the server, rather than the LR locations. The server then solves the LP problem based on the received coefficients. Furthermore, we refine the LP framework by incorporating an exponential obfuscation mechanism to guarantee the indistinguishability of obfuscation distribution across multiple users. Based on the constraint structure of the LP formulation, we apply Benders' decomposition to further enhance computational efficiency. Our theoretical analysis confirms that, despite the geo-obfuscation being calculated independently for each user, it still meets geo-indistinguishability constraints across multiple users with high probability. Finally, the experimental results based on a real-world dataset demonstrate that LR-Geo outperforms existing geo-obfuscation methods in computational time, data utility, and privacy preservation. 
\end{abstract}

\keywords{Geo-obfuscation, location privacy, linear programming}

\maketitle

\vspace{-0.00in}
\section{Introduction}
\label{sec:introduction}
\vspace{-0.00in}

Among a variety of \emph{location privacy protection mechanisms (LPPMs)}, geo-obfuscation has become the preferred paradigm for protecting individual location privacy against server-side data breaches \cite{Yu-NDSS2017}. Geo-obfuscation allows mobile users to report obfuscated locations instead of their exact locations to servers in location-based services (LBS). As the obfuscation process is irreversible \cite{Bakken-IEEESP2004}, users' exact locations are well-protected even if the obfuscated locations are disclosed to attackers. This is achieved by satisfying certain privacy criteria, such as \emph{geo-indistinguishability (Geo-Ind)} \cite{Andres-CCS2013}, which requires that, for any two locations geographically close, the probability distribution of their obfuscated locations should be sufficiently close so that it is difficult for an attacker to distinguish the two locations based on their obfuscated representations. \looseness = -1

Although geo-obfuscation provides a strong privacy guarantee for users' locations, the location errors introduced by obfuscation can negatively impact the quality of LBS. Many recent efforts \cite{Mendes-PETS2020, Simon-EuroSP2019, Wang-WWW2017, Shokri-TOPS2017, Yu-NDSS2017, Xiao-CCS2015, Bordenabe-CCS2014, Qiu-CIKM2020, Qiu-TMC2020, Qiu-SIGSPATIAL2022, Al-Dhubhani-PETS2017, Wang-CIDM2016} aim to address the quality issue caused by geo-obfuscation using \emph{linear programming (LP)} \cite{Linear&Nonlinear}, of which the objective is to minimize the utility loss with the privacy criterion like Geo-Ind guaranteed. For the sake of computational tractability, the LP-based methods typically discretize the location field into a finite set of discrete locations. Its decision variables, represented as an \emph{obfuscation matrix}, determine the probability distributions of obfuscated locations given each possible real location. 

Due to the intricate complexity of LP, generating the obfuscation matrix directly on users' mobile devices is not feasible. Instead, the matrix is calculated by a server, which optimizes the matrix before it is downloaded by the mobile devices \cite{Wang-WWW2017}. Given that the server lacks knowledge of users' precise locations, the server typically considers every location within the target area when calculating the matrix, regardless of whether it is currently occupied by a user. After downloading the matrix, each user selects the specific row of the matrix that matches their actual location to determine the probability distribution of the obfuscated locations. Consequently, the LP formulation of the obfuscation matrix involves $K^2$ decision variables, where $K$ denotes the number of discrete locations within the target region. This results in a significant challenge in accommodating a large array of locations. For instance, the inclusion of thousands of distinct locations within a modestly sized town escalates the number of decision variables into the millions \cite{Qiu-TMC2020}. As compared in Table \ref{Tb:reference} in Section \ref{sec:related} (Related Work), most current LP-based works limit the number of discrete locations $K$ to up to 100. \looseness = -1


\vspace{0.03in}
\noindent \textbf{Motivations}. The traditional LP-based geo-obfuscation methods (e.g., \cite{Qiu-TMC2020, Wang-WWW2017}) have a high computation overhead since the LP is formulated completely by the server side, which requires accounting for all locations within the target region. However, from an individual user's perspective, they engage only with the specific row that matches their actual location.  Although this single row cannot be generated in isolation as it is linked to some other rows by ``Geo-Ind'', such constraints are only enforced between the nearby locations \cite{Andres-CCS2013}. This indicates that, if the LP can be formulated locally by each user, they only need to consider ``locally relevant'' locations so that the computational cost can be significantly reduced. In practical terms, when a user chooses an obfuscated location, the relevance of how another user 100 kilometers away selects their obfuscated location due to the Geo-Ind constraints is minimal.

Motivated by the above observation, this paper introduces a new geo-obfuscation paradigm,  termed \emph{\underline{L}ocally \underline{R}elevant \underline{Geo}-obfuscation (LR-Geo) locations}. The core idea of LR-Geo is to allow each user to formulate the LP by themselves by focusing exclusively on their \emph{Locally Relevant (LR)}, thereby streamlining the process of generating obfuscation matrices. Nevertheless, the development of LR-Geo presents several distinct challenges:

\vspace{-0.08in}

\subsubsection*{\textbf{Challenge 1: How to determine the LR location set?}} First, it is important to note that even a location far from a user's location can have an indirect impact on the user's obfuscation distribution since the distant location can have higher relevance to other locations closer to the user by the Geo-Ind constraints. Considering such a ``multi-hop'' influence of Geo-Ind is hard to circumvent while pursuing the globally optimal solution, our approach focuses on striking a balance between optimizing the obfuscation matrix and enhancing computational efficiency, achieved by selecting an appropriate LR location set.

Specifically, we introduce a \emph{Geo-Ind graph} to describe the Geo-Ind constraints between each nearby location pair, which also enables us to quantify the ``multi-hop'' impact of Geo-Ind constraints through the path distance between nodes in the graph (see \textbf{Theorem \ref{thm:SPD}}). Using the Geo-Ind graph, we determine the LR location set for each user as the collection of locations whose path distance from the user's actual location does not surpass a predefined threshold. Following this, we formulate the LP of LR-Geo for each user to focus exclusively on their selected LR location set.

\vspace{-0.08in}
\subsubsection*{\textbf{Challenge 2: How to calculate LR-Geo?}} Despite having a relatively smaller LP size, the calculation of LR-Geo still needs to be migrated to the server since (i) the computational demands of LR-Geo remain relatively high for mobile devices, and (ii) LR-Geo's LP formulation involves assessing data utility for downstream decision-making, a task typically handled by the server rather than individual users \cite{Wang-WWW2017}. However, each user needs to keep the LR location set hidden from the server, as these locations could potentially disclose the user's actual location. As a workaround, we enable each user to locally compute the coefficients of the LP formulation with server assistance and then upload these coefficients to the server. We demonstrate that the uploaded coefficients can be used by the server to solve the LP problems but cannot be reversed to unveil the LR location of the user (by examples in Section \ref{subsec:discussionthreats} and experimental results in Fig. \ref{fig:exp:threat} in Section \ref{sec:performance}).

\vspace{-0.08in}
\subsubsection*{\textbf{Challenge 3: How to guarantee Geo-Ind across multiple users?}} 
Given that each user conceals their LR location set from the server, formulating Geo-Ind constraints across users in LP becomes another challenge for the server. To address this, we enable the server to apply \emph{exponential distribution constraints} on a selected subset of obfuscated locations for each user. We demonstrate that adhering to these constraints ensures that the chosen obfuscated locations meet Geo-Ind constraints across users even though their obfuscation is calculated in an independent manner (see \textbf{Theorem \ref{thm:privacyguarantee}}). Moreover, our experimental findings in Fig. \ref{fig:exp:GVR} indicate that while unselected obfuscated locations do not theoretically guarantee Geo-Ind, they still possess a high probability (e.g. 99.81\% on average) of meeting Geo-Ind constraints in practice. Additionally, by integrating the exponential mechanism with LP, the constraint matrix of LR-Geo follows a \emph{ladder block structure}, making the problem well-suited to \emph{Benders' decomposition}, which further improves the computation efficiency of solving LR-Geo.

\vspace{-0.08in} 
\subsubsection*{\textbf{Experimental results}} 
Lastly, in our experiment, we assessed LR-Geo's performance by simulating its application to road map data sourced from Rome, Italy \cite{roma-taxi-20140717}. The results revealed that LR-Geo efficiently generates obfuscation matrices within 100 seconds for cases involving up to 1500 locations in the target area. This marks a substantial enhancement over existing LP-based geo-obfuscation techniques (as listed in Table \ref{Tb:reference}), which can only handle up to 100 locations. Furthermore, our experimental results show that LR-Geo's obfuscation matrix not only adheres closely to the theoretical lower bound of expected cost, as established in \textbf{Theorem \ref{thm:lowerbounds} and Theorem \ref{thm:upperbounds}}, but also outperforms contemporary benchmarks \cite{Qiu-TMC2020, Andres-CCS2013, ImolaUAI2022} in terms of time efficiency and cost-effectiveness.

\vspace{-0.08in} 
\subsubsection*{\textbf{Contributions}} In summary, the contributions of this paper are summarized as follows: 
\newline 1. We introduce LR-Geo, a new geo-obfuscation approach that significantly reduces the computational overhead of geo-obfuscation while maintaining a high level of optimality.
\newline 2. We develop a remote computing framework that allows for the offloading of LR-Geo computations to a server while preserving the privacy of each user's LR location set.
\newline 3. To achieve Geo-Ind across multiple users' obfuscation matrices, we integrate exponential distribution constraints within the LP computational framework. Given LR-Geo's constraint structure, we apply Benders' decomposition to enhance computational time efficiency.
\newline 4. Our experimental results demonstrate that LR-Geo not only approximates optimal solutions with considerably lower computational costs but also outperforms existing state-of-the-art methods in time efficiency and cost-effectiveness.

The rest of the paper is organized as follows: The next section provides the preliminaries of geo-obfuscation. Section \ref{sec:motivation} describes the motivation and Section \ref{sec:algorithm} designs the algorithm. Section \ref{sec:performance} evaluates the algorithm's performance. Section \ref{sec:related} presents the related work and  Section \ref{sec:conclude} makes a conclusion.  

\vspace{-0.00in}

\vspace{-0.08in}
\section{Preliminary}
\label{sec:preliminary}
\vspace{-0.00in}

In this section, we introduce the preliminary knowledge of geo-obfuscation, including its framework in LBS in \textbf{Section \ref{subsec:LBS}}, its privacy criteria Geo-Ind in \textbf{Section \ref{subsec:geoind}}, and the LP formulation in \textbf{Section \ref{subsec:LP}}. The main notations used throughout this paper can be found in Table \ref{Tb:Notationmodel} in \textbf{Section \ref{sec:notations} in Appendix}. 

\vspace{-0.05in}
\subsection{Geo-Obfuscation in LBS}
\label{subsec:LBS}
\vspace{-0.00in}
We consider an LBS system composed of a \emph{server} and \emph{a set of users}, where users need to report their locations to the server to receive the desired services. Like \cite{Shokri-CCS2012, Andres-CCS2013, Bordenabe-CCS2014, Wang-WWW2017}, we assume that the server is not malicious, but it might suffer from a \emph{passive attack where attackers can eavesdrop on the users' reported locations breached by the server}. In this case, users can hide their exact locations from the server using geo-obfuscation mechanisms \cite{Wang-WWW2017}.

In general, a geo-obfuscation mechanism can be represented as a \emph{probabilistic function}, of which the input and the output are the user's real location and obfuscated location, respectively. For the sake of computational tractability, many existing works like \cite{Wang-WWW2017, Bordenabe-CCS2014, Qiu-TMC2020, Pappachan-EDBT2023, Qiu-EDBT2024} approximate the users' location field to a discrete location set $\mathcal{V} = \{v_1, ..., v_K\}$. In this case, the obfuscation function can be represented as a \emph{stochastic obfuscation matrix} $\mathbf{Z} = \{z_{i,k}\}_{K\times K}$, where each $z_{i,k}$ denotes the probability of taking $v_k$ as the obfuscated location given the actual location $v_i$.

Besides hiding the users' actual location, the obfuscation matrix $\mathbf{Z}$ is designed to minimize the \emph{utility loss} (or \emph{cost}) caused by geo-obfuscation. As an example, in this paper, we focus on a category of LBS where a mobile user needs to physically travel to a specified destination to receive service (e.g., hotel/restaurant recommendations \cite{yelp}) or implement a task (e.g., spatial crowdsourcing \cite{waze, openstreetmap}). Typically, these LBS types strive to minimize travel expenses for users. Accordingly, \emph{we define the cost resulting from geo-obfuscation as the distortion between the estimated travel distances (using obfuscated locations) and the actual travel distances incurred by users}. Note that our approach in this paper can be readily adapted in other LBS applications as long as the explicit relationship between cost and location obfuscation can be established.

To calculate the traveling costs, global LBS information such as traffic conditions and destination distribution is needed. Since global information is hard to maintain by individuals, many existing works like \cite{Wang-WWW2017, Shokri-TOPS2017, Yu-NDSS2017, Qiu-CIKM2020, Qiu-TMC2020} let the server manage the computation of the obfuscation matrix. Specifically, before reporting the location to the server, each privacy-aware user downloads the obfuscation matrix $\mathbf{Z}$ from the server. Given the current location $v_i$, the user finds the corresponding row $\mathbf{z}_i = [z_{i,1}, ..., z_{i,K}]$ in the obfuscation matrix, based on which the user then randomly selects an obfuscated location to report. In what follows, we call $\mathbf{z}_i$ the \emph{obfuscation vector} of the location $v_i$.


\DEL{
\begin{figure}[t]
\centering
\begin{minipage}{0.45\textwidth}
\centering
  \subfigure{
\includegraphics[width=1.00\textwidth]{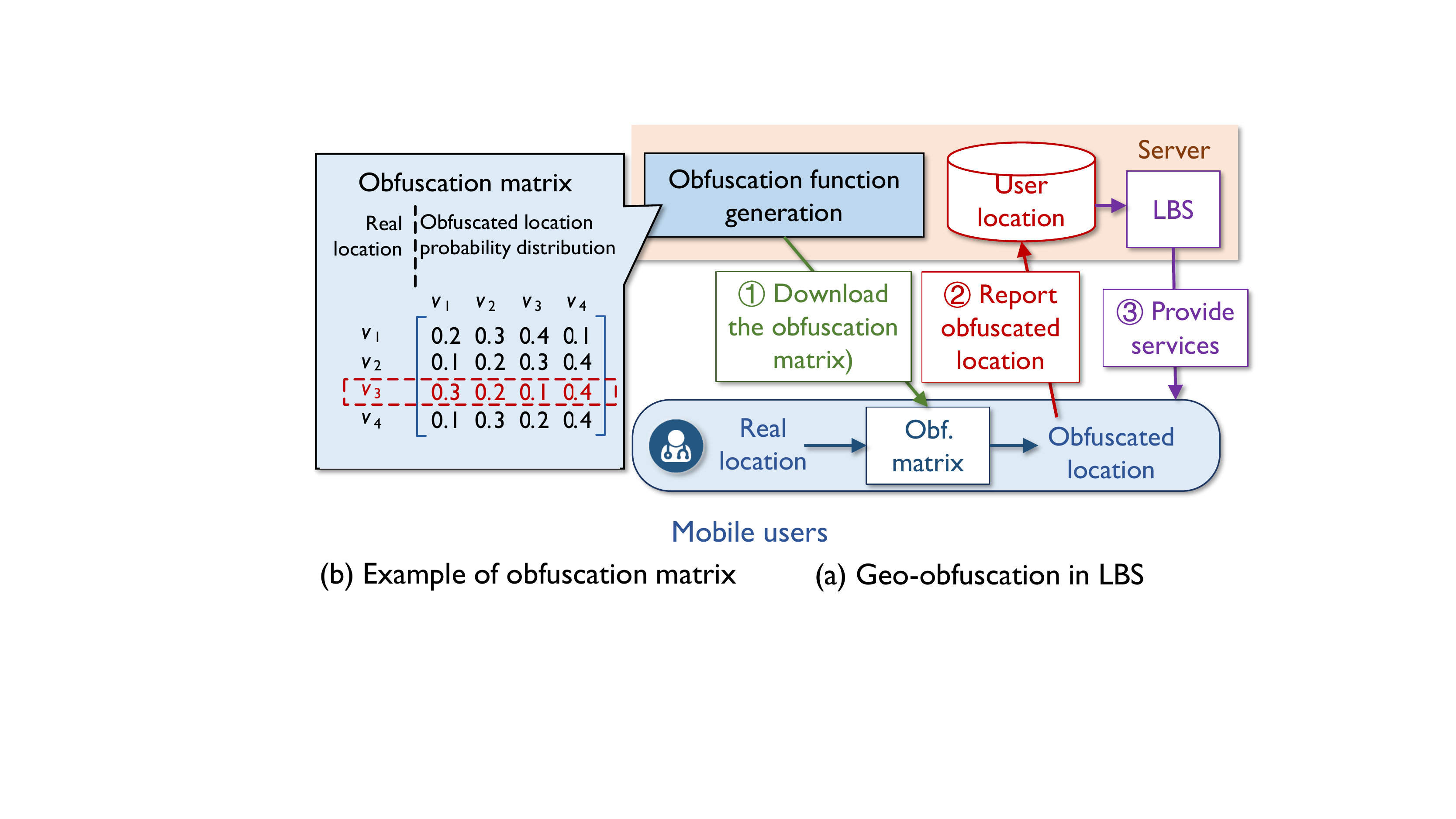}}
\label{}
\end{minipage}
\vspace{-0.00in}
\caption{\small The framework of CPR user assignment.}
\label{fig:framework}
\vspace{-0.00in}
\end{figure}}

\DEL{
\begin{figure}[t]
\centering
\begin{minipage}{0.48\textwidth}
\centering
  \subfigure{
\includegraphics[width=1.00\textwidth]{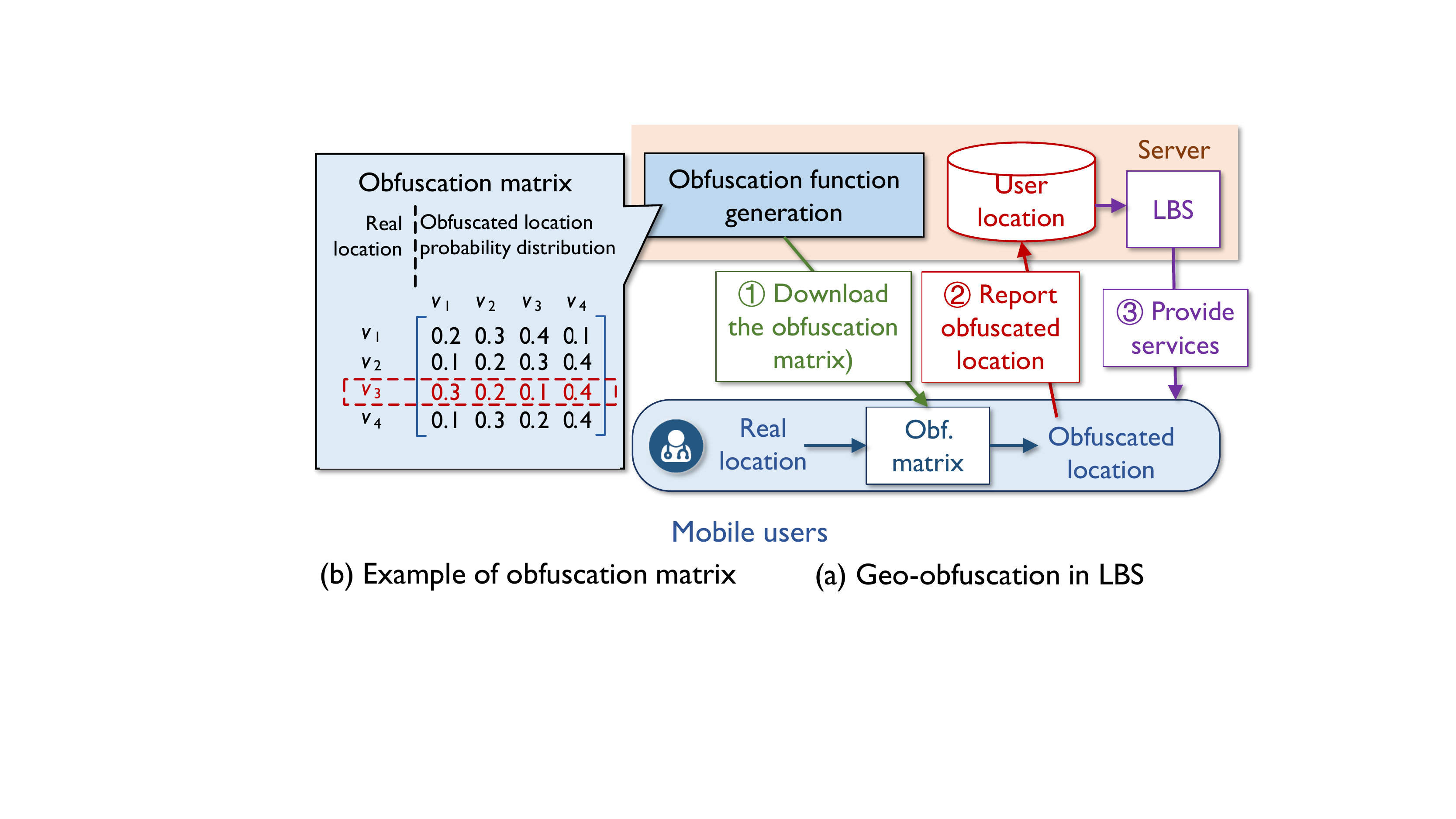}}
\label{}
\end{minipage}
\vspace{-0.20in}
\caption{The framework of geo-obfuscation in LBS. 
\newline \small Process in (a): \textcircled{1} Mobile users first download the obfuscation function (matrix) from the server. \textcircled{2} Users then report their locations to the server. \textcircled{3} Based on the users' reported locations, the server provides services to the participating users.
}
\label{fig:framework}
\vspace{-0.00in}
\end{figure}}


\vspace{-0.08in}
\subsection{Geo-Indistinguishability}
\label{subsec:geoind}

Although the server takes charge of generating the obfuscation matrix, the users' exact locations are still hidden from the server since the obfuscated locations are selected in a probabilistic manner \cite{Wang-WWW2017}. In particular, the obfuscation matrix  $\mathbf{Z}$ is designed to satisfy the privacy criterion \emph{Geo-Ind}, indicating that even if an attacker has obtained the users' reported (obfuscated) location and $\mathbf{Z}$ from the server, it is still hard for the attacker to distinguish the users' exact locations from the nearby locations. 

We use $d_{v_i, v_j}$ to denote the \emph{Haversine distance} (the angular distance on the surface of a sphere) between $v_i$ and $v_j$. Given a threshold $\gamma > 0$, we call two locations $v_i$ and $v_j$ ``neighboring locations'' if their distance $d_{v_i, v_j} \leq \gamma$. \emph{Geo-Ind} is formally defined in \emph{Def. \ref{def:GeoI}}  \cite{Andres-CCS2013}: \looseness = -1
\vspace{-0.00in}
\begin{definition} 
\label{def:GeoI}
(Geo-Ind) An obfuscation matrix $\mathbf{Z}$ satisfies $(\epsilon, \gamma)$-Geo-Ind if, for each pair of neighboring locations $v_i, v_j \in \mathcal{V}$ with $d_{v_i, v_j} \leq \gamma$, the following constraints are satisfied 
\vspace{-0.00in}
\begin{equation}
\label{eq:Geo-Ind-general}
z_{i,k}  - e^{{\epsilon d_{v_i, v_j}}}  z_{j,k} \leq 0, ~ \forall v_k \in \mathcal{V},
\vspace{-0.00in}
\end{equation}
i.e., the probability distributions of the obfuscated locations of $v_i$ and $v_j$ are sufficiently close. Here, $\epsilon$ is called the privacy budget. Higher $\epsilon$ implies a lower privacy level. 
\end{definition}
\vspace{-0.00in}
In what follows, we use $\mathcal{E} = \left\{(v_i, v_j)\in \mathcal{V}^2 | d_{v_i, v_j} \leq \gamma \right\}$ to denote the set of neighboring locations in $\mathcal{V}$. 

{\revisiondone Note that the existing works like \cite{Bordenabe-CCS2014, Wang-WWW2017, Qiu-TMC2020} do not require Geo-Ind to be satisfied only between locations that are within a distance smaller than $\gamma$. Here, we consider $\gamma$ as it forms a more general model, i.e., $\gamma = \infty$ when $\gamma$ is not included in \emph{Definition \ref{def:GeoI}}. In practical, the choice of $\gamma$ depends on the user's privacy requirements, defining the range within which the user's location should be indistinguishable. For example, if a student wants to obscure their location within a campus, selecting $\gamma$ to cover the entire campus would be sufficient. }



\DEL{In general, we call $v_j$ a $k$th-order neighbor of $v_i$ (denoted by $v_j \in \mathcal{N}^{(k)}_{v_i}$), if the shortest-path distance (i.e., the number of edges in the shortest path) between $v_j$ and $v_i$ in the graph $\mathcal{G}$ is $k$. We let $\mathcal{V}^{(k)}_{v_i}  = \bigcup_{l=1}^{(k)}\mathcal{N}^{(l)}_{v_i}$ include all $v_i$'s 1st to $k$th neighbors.}

\vspace{-0.00in}
\subsection{LP Problem Formulation}
\label{subsec:LP}
\vspace{-0.00in}
\noindent \textbf{Constraints}. 
In addition to satisfying Geo-Ind in Equ. (\ref{eq:Geo-Ind-general}), for every real location $v_i$, the total probability of its obfuscated locations should be equal to 1: 
\vspace{0.00in}
\begin{equation}
\label{eq:unitmeasure}
\textstyle    
\sum_{k=1}^K z_{i,k} = 1,~\forall v_i \in \mathcal{V} ~\mbox{(probability unit measure)}.
\vspace{-0.00in}
\end{equation}

\noindent \textbf{Objective function}. Given the target location $v_l$, the real location $v_i$, and the obfuscated location $v_k$, we define the \emph{cost} of LBS as the discrepancy between the estimated travel cost $\mathrm{tc}_{v_i, v_l}$ and the actual travel cost $\mathrm{tc}_{v_k, v_l}$ to reach $v_l$
\vspace{-0.00in}
\begin{equation}
\label{eq:deltad}
\delta_{v_i,v_k,v_l} = \left|\mathrm{tc}_{v_i, v_l} - \mathrm{tc}_{v_k, v_l}\right|. 
\vspace{-0.00in}
\end{equation}
We assume that the server has the prior distribution of the target locations $\mathbf{q} = [q_1, ..., q_K]$, where $q_l$ ($l =1, ..., K$) denotes the probability that a target's location is at $v_l$. The objective is to minimize the \emph{expected cost} caused by the obfuscation matrix $\mathbf{Z}$:  
\begin{eqnarray}
\label{eq:overallcost}
\textstyle   
\mathcal{L}\left(\mathbf{Z}\right) = \sum_{i=1}^K p_i\sum_{k=1}^K \sum_{l=1}^K q_l\delta_{v_i,v_k,v_l} z_{i,k} = \sum_{i=1}^K\mathbf{c}_i\mathbf{z}^{\top}_i, 
\end{eqnarray}
where $p_k$ ($k =1, ..., K$) denotes the prior probability that a user's real location is at $v_k$, $\mathbf{c}_i = \left[c_{v_i,v_1}, ..., c_{v_i,v_K}\right]$ denote the cost (cost) coefficients of $\mathbf{z}_i$ in the objective function, and each $c_{v_i, v_k}$ is given by
\vspace{-0.00in}
\begin{equation}
\label{eq:c}
\textstyle 
c_{v_i, v_k} = p_i \sum_{l=1}^K q_l \delta_{v_i,v_k,v_l} ~ (i = 1, ..., K).
\vspace{-0.00in}
\end{equation}
{\revisiondone In related works such as \cite{Shokri-CCS2012, Andres-CCS2013}, utility loss is typically defined as the distance between original and obfuscated locations. However, this metric does not fully capture the utility loss in many vehicle-based applications, as they fail to consider the constraints imposed by vehicle mobility \cite{Qiu-TMC2020}. Therefore, in this paper, we define utility loss based on downstream decision-making in data processing. Specifically, in our experiments, we focus on spatial crowdsourcing, where utility loss of $\mathbf{Z}$ is measured by the expected estimation error in travel cost caused by $\mathbf{Z}$.}

\noindent \textbf{Problem formulation}. To satisfy the constraints of Geo-Ind (Equ. (\ref{eq:Geo-Ind-general})) and the probability unit measure (Equ. (\ref{eq:unitmeasure})),  and minimize $\mathcal{L}\left(\mathbf{Z}\right)$ (Equ. (\ref{eq:overallcost})), the problem of \emph{LR obfuscation matrix generation (OMG)} can be formulated as the following LP problem:
\vspace{-0.0in}
\begin{eqnarray}
\label{eq:OMGobj}
\min && \textstyle \mathcal{L}\left(\mathbf{Z}\right) = \sum_{i=1}^K\mathbf{c}_i\mathbf{z}^{\top}_i \\ \label{eq:OMGconstr}
\mbox{s.t. } && \mbox{ Equ. (\ref{eq:Geo-Ind-general})(\ref{eq:unitmeasure})  are  satisfied}. 
\vspace{-0.00in}
\end{eqnarray}

\vspace{-0.00in}
\section{Motivations and Observations}
\label{sec:motivation}
Although the OMG problem outlined in Equ. (\ref{eq:OMGobj})(\ref{eq:OMGconstr}) can be solved using classical LP algorithms such as the simplex method \cite{Linear&Nonlinear}, it is hampered by high computational costs. The time complexity of an LP problem depends on the number of decision variables and the number of linear constraints \cite{Linear&Nonlinear}. In OMG, the decision matrix $\mathbf{Z}$ encompasses $O(|\mathcal{V}|^2)$ decision variables, where it must adhere to the Geo-Ind constraints for every pair of neighboring locations in $\mathcal{V}$, resulting in $O(|\mathcal{E}||\mathcal{V}|)$ linear constraints. This substantial computational demand renders the LP-based geo-obfuscation impractical for scenarios with a large number of locations. 
Therefore, \textbf{enhancing the computational efficiency of solving LP-based geo-obfuscation is the primary goal of this paper}.

\begin{figure}[t]
\centering
\hspace{-0.0in}
\begin{minipage}{0.46\textwidth}
\centering
  \subfigure{
\includegraphics[width=1.00\textwidth]{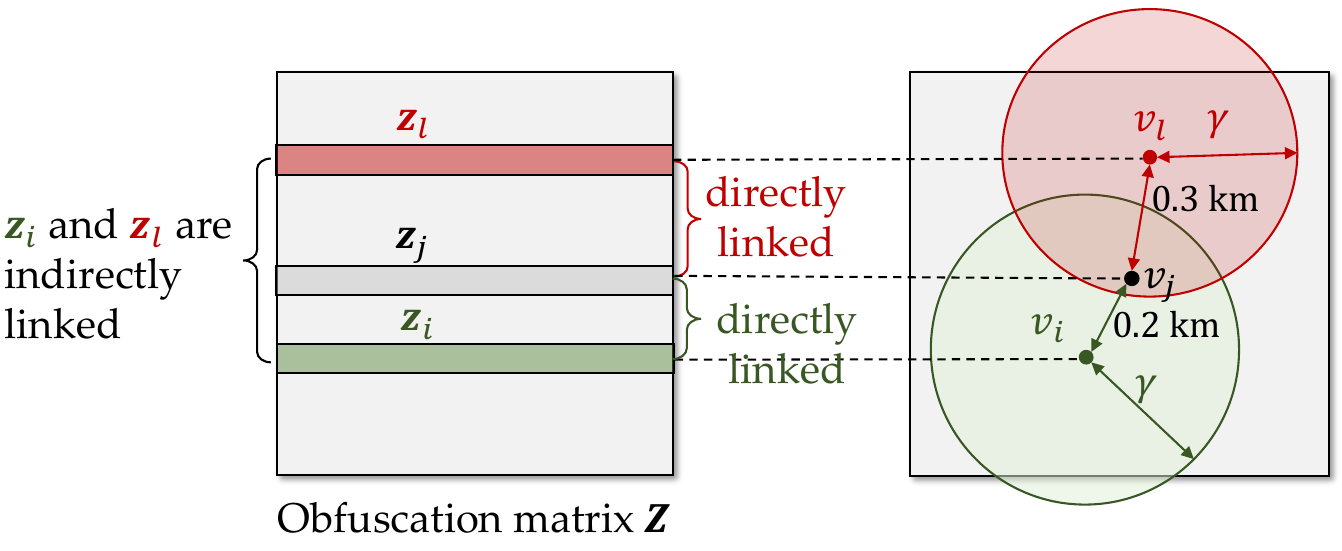}}
\end{minipage}
\vspace{-0.15in}
\caption{Directly/indirectly linked obfuscation vectors. }
\label{fig:linkedvariable}
\vspace{-0.15in}
\end{figure}

\vspace{0.03in}
\noindent \textbf{Observations}. When a user at location $v_i$ seeks to obfuscate their actual location, 
they use only the $i$th row $\mathbf{z}_i = [z_{i,1}, ..., z_{i,K}]$ instead of the entire matrix $\mathbf{Z}$. {\rd However, determining $\mathbf{z}_i$ in isolation is not feasible within OMG as $\mathbf{z}_i$ is linked to other rows (obfuscation vectors) of $\mathbf{Z}$ by the Geo-Ind constraints.} 

As depicted in Fig. \ref{fig:linkedvariable}, Geo-Ind requires that obfuscation vector $\mathbf{z}_i$ is directly connected to another vector $\mathbf{z}_j$ if their corresponding locations $v_i$ and $v_j$ are neighbors. Additionally, a location $v_l$, even if distant from $v_i$, indirectly influences $\mathbf{z}_i$ through its neighborly relation with $v_j$. Considering that it is hard to circumvent such ``multi-hop'' influence of Geo-Ind between locations while pursuing the globally optimal solution, we aim to balance the optimality of geo-obfuscation and computational efficiency by focusing on a selectively identified set of locations that exert a significant influence on $v_i$. This approach is based on the intuition that locations nearer to $v_i$ have obfuscation vectors $\mathbf{z}_l$ with a more pronounced effect on $\mathbf{z}_i$. The pivotal question then becomes how to quantify the extent of influence between $\mathbf{z}_i$ and other vectors, such as $\mathbf{z}_l$.

To this end, we introduce the concept of the \emph{Geo-Ind graph} in \textbf{Def. \ref{def:GeoIGraph}}. We then detail the application of this graph to measure the Geo-Ind connection between $\mathbf{z}_i$ and $\mathbf{z}_l$ in \textbf{Theorem \ref{thm:SPD}}.
\begin{definition}
\label{def:GeoIGraph}
(\textbf{Geo-Ind Graph}) \textit{Geo-Ind graph} is defined as an undirected graph $\mathcal{G} = (\mathcal{V}, \mathcal{E})$ to describe the Geo-Ind constraints between locations within a set $\mathcal{V}$. Here, $\mathcal{V}$ represents the set of nodes, each corresponding to a distinct location, and $\mathcal{E}$ denotes the set of edges. Each edge $(v_i, v_j) \in \mathcal{E}$ indicates that the locations $v_i$ and $v_j$ are neighbors (i.e. $d_{v_i, v_j}\leq \gamma$) with the edge weight equal to the distance $d_{v_i, v_j}$ between them.
\end{definition}

\begin{theorem}
\label{thm:SPD}
Consider two locations $v_i$ and $v_j$ are connected through at least one path in the Geo-Ind graph $\mathcal{G}$. Let the \textit{path distance} $D_{v_i,v_j}$ represent the sum of weights of the edges forming the shortest path between $v_i$ and $v_j$. Their probabilities of selecting location $v_k$ as the obfuscated location is constrained by: 
\begin{equation}
\label{eq:Geo-Ind-path}
z_{i,k} \leq e^{\epsilon D_{v_i,v_j}} z_{j,k}.
\end{equation}
\end{theorem}
\begin{proof}
Detailed proof can be found in Section \ref{sec:proofs} in Appendix. 
\end{proof}
\vspace{-0.00in}

\textbf{Theorem \ref{thm:SPD}} implies the extent to which a pair of obfuscation vectors, $\mathbf{z}_i$ and $\mathbf{z}_j$, are linked through Geo-Ind constraints depends on the path distance $D_{v_i,v_j}$ between their respective locations $v_i$ and $v_j$ in the Geo-Ind graph $\mathcal{G}$. A higher path distance between locations implies a weaker linear constraint between their obfuscation vectors.

In Fig. \ref{fig:linkeddec}, we follow the example of Fig. \ref{fig:linkedvariable}, and check specifically how the obfuscation vector $\mathbf{z}_i$ is impacted by the decision vectors $\mathbf{z}_j$ and $\mathbf{z}_l$ according to the conclusion of Theorem \ref{thm:SPD}. As Fig. \ref{fig:linkeddec} shows, given the path distances $D_{v_i, v_j} = 0.2$km and $D_{v_i, v_l} = 0.5$km, and the privacy budget $\epsilon = 10.0$km$^{-1}$, each entry of $\mathbf{z}_i$ follows the following linear constraints according to Theorem \ref{thm:SPD}: 
\vspace{-0.00in}
\begin{eqnarray}
z_{i,k} \leq e^{\epsilon D_{v_i, v_j}} z_{j,k} \Rightarrow z_{i,k} \leq e^{2} z_{j,k} \mbox{ (stronger constraints)} \\
z_{i,k} \leq e^{\epsilon D_{v_i, v_l}} z_{l,k} \Rightarrow z_{i,k} \leq e^{5}z_{l,k} \mbox{ (weaker constraints)} 
\end{eqnarray}
indicating that \textbf{$z_{l,k}$ enforces a weaker constraint on $z_{i,k}$ compared to $z_{j,k}$}. Given that $z_{i,k}$ represents a probability measure and therefore cannot exceed 1, the condition $z_{i,k} \leq e^{5} z_{l,k}$ becomes irrelevant for instances where $z_{l,k}$ is just marginally greater than 0 (when $z_{l,k} \geq 0.0067$). This is because the upper limit of $z_{i,k} \leq 1$ naturally satisfies the condition $z_{i,k} \leq e^{5} z_{l,k}$ under these circumstances. Conversely, the constraint $z_{i,k} \leq e^{2} z_{j,k}$ is more stringent and remains applicable unless $z_{j,k} \geq 0.1353$. As $z_{j,k}$ increases beyond 0.1353, the condition $z_{i,k} \leq 1$ is adequate to fulfill the constraint of $z_{i,k} \leq e^{2} z_{j,k}$.


Overall, \textbf{the insight obtained from Theorem \ref{thm:SPD} and the example in Fig. \ref{fig:linkeddec} lead us to focus on a set of ``locally relevant locations'' that are within a specified path distance threshold from $v_i$ in the Geo-Ind graph.} By focusing the LP problem on this narrowed-down LR location set, we can substantially decrease the computational demands associated with solving the LP problem.

\begin{figure}[t]
\centering
\hspace{-0.0in}
\begin{minipage}{0.49\textwidth}
\centering
  \subfigure{
\includegraphics[width=1.00\textwidth]{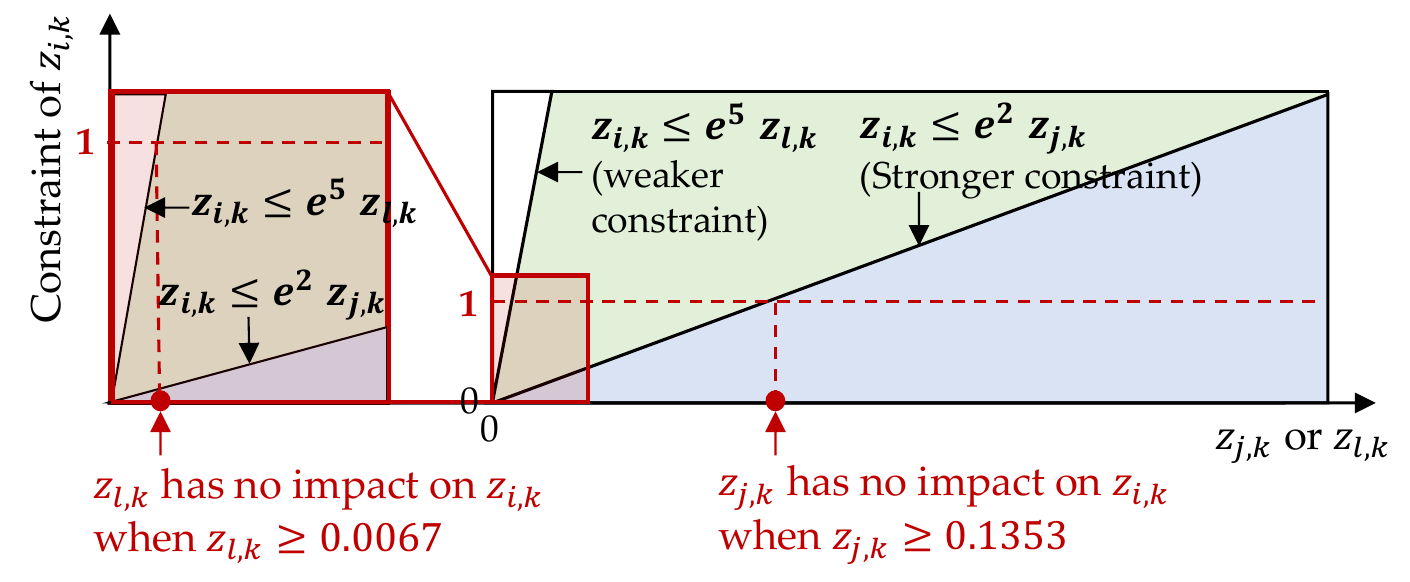}}
\end{minipage}
\vspace{-0.15in}
\caption{Strongly \& weakly linked decision vectors. }
\label{fig:linkeddec}
\vspace{-0.15in}
\end{figure}

In the next section, we introduce the details of our method. 

\vspace{-0.00in}
\section{Methodology}
\vspace{-0.00in}
\label{sec:algorithm}
In this section, we present \emph{\underline{L}ocally \underline{R}elevant \underline{Geo}-Obfuscation (LR-Geo)}, detailing its main concepts and problem formulation in  \textbf{Section \ref{subsec:idea}}, the computational framework in \textbf{Section \ref{subsec:computeframework}--\ref{subsec:costreference}}, the theoretical performance analysis in \textbf{Section \ref{subsec:perfanalysis}}, and a discussion of potential threats in \textbf{Section \ref{subsec:discussionthreats}}.

\vspace{-0.00in}
\subsection{Locally Relevant Geo-Obfuscation}
\vspace{-0.00in}
\label{subsec:idea}

We consider a scenario where $M$ users $\{1, ..., M\}$ need to obfuscate their locations. Without loss of generality, we assume each user $m$ is located at $v_m$ ($m = 1, ..., M$). Therefore, each user $m$ needs to use the $m$th row of the obfuscation matrix $\mathbf{Z}$, denoted by $\mathbf{z}_m = [z_{m,1}, ..., z_{m,K}]$, to determine the probability distribution of $v_m$'s obfuscated locations. 

Inspired by the insights discussed in Section \ref{sec:motivation}, the underlying concept of \emph{LR-Geo} is to generate an obfuscation matrix focusing solely on the ``locally relevant (LR) locations'' surrounding each user's actual location $v_m$ to reduce the computational overhead. According to \textbf{Theorem \ref{thm:SPD}}, within the Geo-Ind graph $\mathcal{G}$, locations with shorter path distances to $v_m$ exhibit stronger connections of their obfuscation vectors to $\mathbf{z}_m$ through Geo-Ind constraints. Therefore, we identify the LR location set based on their path distance to $v_m$ in the Geo-Ind graph, as described in \textbf{Definition \ref{def:LRlocations}}: 
\begin{definition}
\label{def:LRlocations}
(LR location set) The LR location set of $v_m$, denoted by $\mathcal{N}_{m}$, is defined as the set of locations with their path distances to $v_m$ no greater than a predetermined threshold $\Gamma$: 
\begin{equation}
\label{eq:Gamma}
\mathcal{N}_{m} = \left\{v_j \in \mathcal{V}\left|D_{m, j} \leq \Gamma\right.\right\}, 
\vspace{-0.00in}
\end{equation}
where the constant $\Gamma$ is called the LR distance threshold. 
\end{definition}





\DEL{
Clearly, a lower $\Gamma$ causes fewer decision vectors to derive in OMG, which mitigates the computation cost. While on the other hand, lower $\Gamma$ might deviate the derived $\mathbf{Z}$ from the optimal values even though the removed decision vectors are weakly linked to $\mathbf{z}_i$. Hence, the first question is 
\begin{itemize}
\item [\textbf{Q1}:] How to find $\mathcal{N}_{m}$ with an appropriate $\Gamma$ to reduce the obfuscation matrix computation cost with guaranteeing its optimality at an acceptable level?
\end{itemize}

On the other hand, since $v_i$ is close to the center of its LR location set $\mathcal{N}_{m}$, when migrating the calculation of $\mathbf{Z}$ to the server, $\mathcal{N}_{m}$ should be hidden from the server as it might disclose $v_i$. Therefore, the second research question is 
\begin{itemize}
\item [\textbf{Q2}:] How to migrate the obfuscation matrix computation to the server without disclosing the LR location sets?
\end{itemize}
In Section \ref{sec:algorithm}, we introduce how to address Q1 and Q2. }


\DEL{
{\bl 
\subsection{Local Geo-Indistinguishability (can be removed)}

Our goal is to only \emph{hide $v_i$ from its neighbors}, i.e., $v_i$ is geo-indistinguishable from its neighbors. It is less relevant whether other locations in the region are geo-indistinguishable or not. 

Consider that, when a user at location $v_i$ obfuscates his/her location, to satisfy Geo-Ind, the obfuscation distribution of $v_i$, $\mathbb{P}\mathrm{r}\left(Y=v_k|X=v_i\right)$, should be sufficiently close to the obfuscation distribution of any of its neighbor $v_j\in \mathcal{N}_{m}$, $\mathbb{P}\mathrm{r}\left(Y=v_k|X=v_j\right)$, i.e., 
\begin{equation}
e^{-\epsilon d_{v_i, v_j}} \leq \frac{\mathbb{P}\mathrm{r}\left(Y=v_k|X=v_i\right)}{\mathbb{P}\mathrm{r}\left(Y=v_k|X=v_j\right)} \leq e^{\epsilon d_{v_i, v_j}},
\end{equation}
so that it is hard for the attacker to distinguish $v_i$ and $v_j$ from their obfuscated (reported) locations. 

In the original definition of Geo-Ind, the Geo-Ind constraints are enforced for each pair of neighbors. However, from a single user's perspective, the user cares more about whether his/her location can be hidden in a certain range (i.e., from its neighbors); rather than whether other locations outside this range are geo-indistinguishable or not.

By referring to the original definition of Geo-Ind (Definition \ref{def:GeoI}), we formally define \emph{partial Geo-Ind} as follows: 
\begin{definition}
\label{def:partialGeoI}
An obfuscation matrix $\mathbf{Z}$ satisfies partial $(\epsilon, r, v_i)$-Geo-Ind, if $v_i$ is geo-indistinguishable from any of its neighbor $v_j \in \mathcal{N}_{m}$, i.e.,   $z_{i,k}  - e^{{\epsilon d_{v_i, v_j}}}  z_{j,k} \leq 0$  and $z_{j,k}  - e^{{\epsilon d_{v_i, v_j}}}  z_{i,k} \leq 0$.
\end{definition}

Our objective is to generate an obfuscation matrix for the user located at $v_i$, called \emph{partial obfuscation matrix}, that satisfies partial $(\epsilon, r, v_i)$-Geo-Ind and the expected estimation error of traveling cost caused by obfuscation is minimized. }}

\subsubsection{LR Location Set Searching} For each user $m$, the LR location set $\mathcal{N}_{m}$ can be found locally by the user. Specifically, given the coordinates of the locations in $\mathcal{V}$ and the neighbor threshold $\gamma$, the user first creates the \emph{Geo-Ind graph} $\mathcal{G} = (\mathcal{V}, \mathcal{E})$ by checking whether the \emph{Haversine distance} between each pair of locations is no higher than $\gamma$. The user then builds a \emph{shortest path tree} rooted at $v_i$ in $\mathcal{G}$ using Dijkstra's algorithm \cite{Algorithm}. The shortest path tree provides the path distance between $v_i$ and each $v_j \in \mathcal{V}$, $D_{i, j}$, based on which the user then determines whether $v_j$ should be included in the LR location set $\mathcal{N}_{m}$ based on Equ. (\ref{eq:Gamma}). 


\begin{figure}[t]
\centering
\vspace{-0.00in}
\begin{minipage}{0.45\textwidth}
\centering
  \subfigure{
\includegraphics[width=1.00\textwidth]{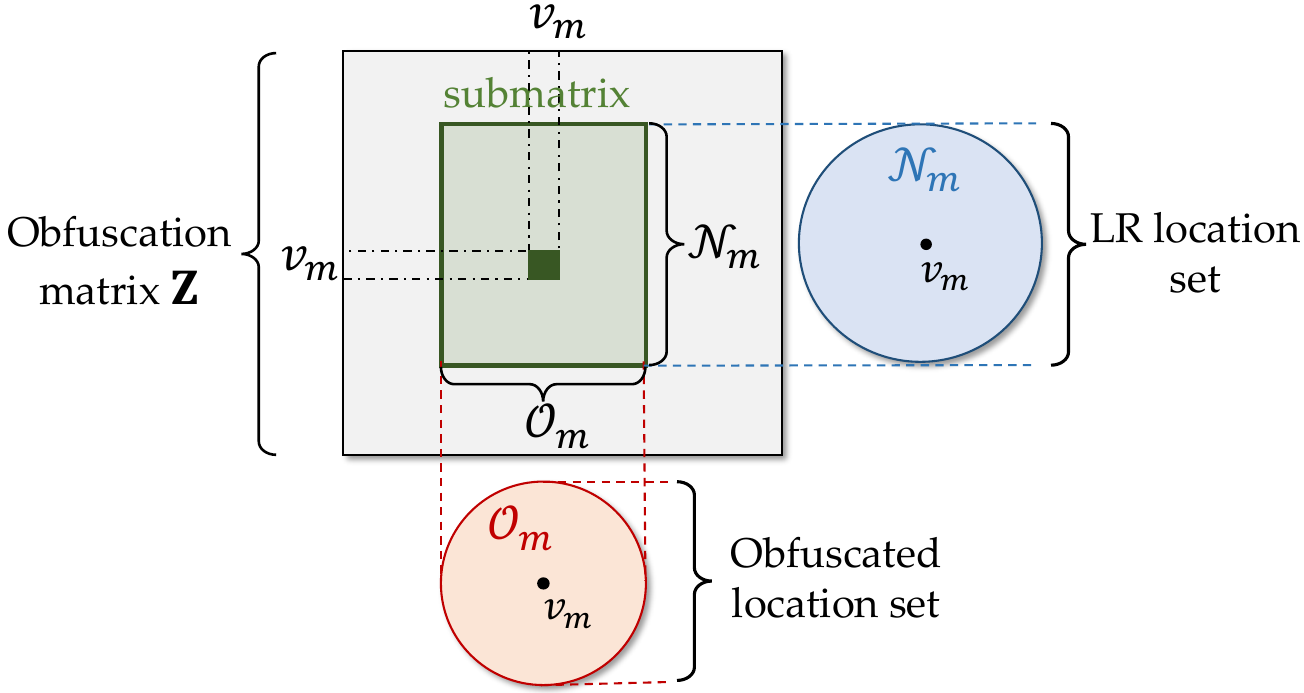}}
\end{minipage}
\vspace{-0.15in}
\caption{Generate a submatrix instead of the whole matrix. }
\label{fig:highlevelidea}
\vspace{-0.15in}
\end{figure}

\DEL{
\looseness = -1

\vspace{-0.00in}
\begin{algorithm}[h]
\SetKwFunction{push}{push}
\SetKwFunction{pop}{pop}
\SetKwFunction{top}{top}
\SetKwInOut{Input}{Input}
\SetKwInOut{Output}{Output}
\small 
\Input{The discrete location set $\mathcal{V}$}
\Output{The LR location set $\mathcal{N}_{m}$}
$\mathcal{N}_{m} \leftarrow \phi$\; 
\tcp{Build the Geo-Ind graph: $\mathcal{G}= (\mathcal{V}, \mathcal{E})$}
\For{each pair $v_j, v_k \in \mathcal{V}$}{
    \If{$d_{v_j, v_k} \leq \gamma$}{
        Add $e_{j, k}$ to $\mathcal{E}$\; 
    }
}
\tcp{Create the LR location set $\mathcal{N}_{m}$}
Build the SP tree rooted at $v_i$ using Dijkstra's algorithm\; 
\For{each $v_j \in \mathcal{V}$}{
    \If{$d_{v_i, v_j} \leq \Gamma$}{
     Add $v_j$ to $\mathcal{N}_{m}$\;
    }
}
\Return $\mathcal{N}_{m}$\; 
\normalsize
\caption{LR location set creation. }
\label{al:locationset}
\end{algorithm}
\vspace{-0.00in}}

\noindent \textbf{Time complexity}. To build $\mathcal{G}$, the user needs to compare the \emph{Haversine distance} between each pair of locations with $\gamma$. This process involves a total of $O(|\mathcal{V}|^2)$ comparisons. 
The time complexity of building a shortest path tree using Dijkstra's algorithm is $O(|\mathcal{V}|^2)$. Therefore, the time complexity of LR location set identification is $O(|\mathcal{V}|^2 + |\mathcal{V}|^2) = O(|\mathcal{V}|^2)$. 

\vspace{0.03in}

\DEL{
\begin{figure}[t]
\centering
\hspace{-0.0in}
\begin{minipage}{0.230\textwidth}
\centering
  \subfigure{
\includegraphics[width=1.00\textwidth]{./fig/exp/ObfPDFvsDistance}}
\end{minipage}
\begin{minipage}{0.230\textwidth}
\centering
  \subfigure{
\includegraphics[width=1.00\textwidth]{./fig/exp/ObfCDFvsDistance}}
\end{minipage}
\vspace{-0.00in}
\caption{Example of the reduced obfuscation range. }
\label{fig:motivation}
\vspace{-0.00in}
\end{figure}}

\vspace{0.00in}
\subsubsection{Obfuscation Range} In the original OMG formulation (Equ. (\ref{eq:OMGobj})(\ref{eq:OMGconstr})), the obfuscation range convers the entire location set $\mathcal{V}$, even though many of these obfuscated locations receive zero probability assignments from the LP algorithm due to their high cost. 
To further reduce the computation cost, we limit the selection of obfuscated locations to a smaller range. Given a real location $v_i$, we consider its obfuscated location range as a circle $\mathcal{C}\left(v_i, r_{\mathrm{obf}}\right)$ centered at $v_i$ with radius $r_{\mathrm{obf}}$. 
Then, the set of the obfuscated locations of $v_i$, denoted by $\mathcal{O}_{m}$ ($\mathcal{O}_{m}\subseteq \mathcal{V}$), can be calculated by
\vspace{-0.00in}
\begin{equation}
\label{eq:O}
\mathcal{O}_{m} = \left\{v_k\in\mathcal{V}\left|d_{v_i, v_k}\leq r_{\mathrm{obf}}\right.\right\}. 
\vspace{-0.00in}
\end{equation}
For each obfuscated location $v_k \notin \mathcal{O}_m$, we assign a small value $\xi$ to the probability $z_{i,k}$. Here, $\xi$ is a small value and we will specify how to determine $\xi$ in Equ. (\ref{eq:expo}) in Section \ref{subsec:computeframework}. 

As Fig. \ref{fig:highlevelidea} shows, after deriving both $\mathcal{N}_{m}$ and $\mathcal{O}_{m}$, the user only needs to download a submatrix of $\mathbf{Z}$, of which the rows and the columns cover $\mathcal{N}_{m}$ and $\mathcal{O}_{m}$, respectively.

\subsubsection{Problem Formulation}
Given each user $m$'s LR location set $\mathcal{N}_m$, we define the user's LR obfuscation matrix as $\mathbf{Z}_{\mathcal{N}_m} = \left\{z^{(m)}_{i,k}\right\}_{\mathcal{N}_m\times \mathcal{V}}$, which includes the obfuscation vectors of all the relevant locations $\mathcal{N}_m$. The cost caused by $\mathbf{Z}_{\mathcal{N}_m}$ is defined by 
\begin{equation}
\textstyle 
\mathcal{L}\left(\mathbf{Z}_{\mathcal{N}_m}\right) = \sum_{v_i \in \mathcal{N}_m}^K\mathbf{c}_i\mathbf{z}^{(m)\top}_i
\end{equation}
where $\mathbf{c}_i = \left[c_{v_i, v_1}, ..., c_{v_i, v_K}\right]$ are the cost coefficients (defined by Equ. (\ref{eq:c})) of the obfuscation vector $\mathbf{z}^{(m)}_i = \left[z^{(m)}_{i,1}, ..., z^{(m)}_{i,K}\right]$. 

The objective of each user $m$ is to minimize $\mathcal{L}\left(\mathbf{Z}_{\mathcal{N}_m}\right)$ while guaranteeing the Geo-Ind constraints among the obfuscation vectors of relevant locations $\mathcal{N}_m$ and the probability unit measure constraint for each obfuscation vector $\mathbf{z}_i$ in $\mathcal{L}\left(\mathbf{Z}_{\mathcal{N}_m}\right)$. Given the LR location set $\mathcal{N}_m$ and the obfuscated location set $\mathcal{O}_m$, we let each user $m$ formulate the \emph{Locally Relevant Obfuscation Matrix Generation (LR-OMG)} problem as the following LP problem:
\begin{eqnarray}
\label{eq:LR-OMGobj}
\min && \textstyle
\mathcal{L}\left(\mathbf{Z}_{\mathcal{N}_m}\right) = \sum_{v_i \in \mathcal{N}_m}\mathbf{c}_i\mathbf{z}^{(m)\top}_i\\  \label{eq:LR-OMGconstr1}
\mbox{s.t. } && \frac{z^{(m)}_{i,k}}{z^{(m)}_{j,k}} \leq e^{{\epsilon d_{v_i, v_j}}}, \forall v_i, v_j \in \mathcal{N}_m~s.t.~ d_{v_i, v_j} \leq \gamma, \forall v_k \\ \label{eq:LR-OMGconstr2}
&& \textstyle 
\sum_{k} z_{i,k} = 1, \forall v_i \in \mathcal{N}_m \\ \label{eq:LR-OMGconstr3}
&& z_{i,k} = \xi, \forall v_k \notin \mathcal{O}_m.
\vspace{-0.00in}
\end{eqnarray}

\subsection{Computation Framework}
\label{subsec:computeframework}

Considering the limited computation capability of users, like most related works \cite{Qiu-TMC2020, Bordenabe-CCS2014, Wang-WWW2017}, we migrate the computation load of LR-OMG to the servers. Note that, for each user $m$, directly uploading the LR location set $\mathcal{N}_m$ and the obfuscated location set $\mathcal{O}_m$ to the server might cause additional privacy leakage, as both $\mathcal{N}_m$ and $\mathcal{O}_m$ can be leveraged to infer the user's real location $v_m$. As a solution shown in Fig. \ref{fig:coefficientcalculation}, we let each user $m$ ($m = 1, ..., M$) upload the coefficient of the formulated LP problem (Equ. (\ref{eq:LR-OMGobj})--(\ref{eq:LR-OMGconstr3}))) to the server, including 
the \textbf{distance matrix} $\mathbf{D}_{\mathcal{N}^2_m}$ and the \textbf{cost matrix} $\mathbf{C}_{\mathcal{N}_m, \mathcal{O}_m}$ to the server, instead of $\mathcal{N}_m$ and $\mathcal{O}_m$. 
\vspace{0.03in}

\noindent \textbf{(1) Distance matrix} $\mathbf{D}_{\mathcal{N}^2_m}$: The user $m$ calculates the Haversine distance $d_{v_i, v_j}$ between each pair of locations $v_i, v_j \in \mathcal{N}_{m}$. Then, the user uploads the distance matrix $\mathbf{D}_{\mathcal{N}^2_m} = \left\{d_{v_i, v_j}\right\}_{(v_i, v_j) \in \mathcal{N}^2_m}$ to the server, which uses each distance value $d_{v_i, v_j}$ to establish the Geo-Ind constraints for each pair of decision variables $z_{i,k}$ and $z_{j,k}$ in Equ. (\ref{eq:Geo-Ind-general}). Note that $\mathbf{D}_{\mathcal{N}^2_m}$ solely provides information about the relative positions of the locations within $\mathcal{N}_{m}$ without disclosing their specific coordinates. \looseness = -1

{\revisiondone It is important to note that if locations are unevenly distributed, an attacker could potentially narrow down the search range of a target user's location based on the user's distance matrix. For example, if users' locations are represented by the "connections" of the road network (retrievable via OpenStreetMap \cite{openstreetmap}), users located in downtown areas would tend to generate a distance matrix with lower values due to the higher density of connections in those areas. {\rd In this paper, we focus primarily on locations that are evenly distributed.} Specifically, similar to existing works \cite{Yu-NDSS2017, Bordenabe-CCS2014, Wang-WWW2017}, we represent the location field using a grid map composed of equally sized grid cells, where each cell is a discrete location. In this case, the distance matrix generated by each user have the same distance values.} 

\begin{figure}[t]
\centering
\hspace{-0.0in}
\begin{minipage}{0.48\textwidth}
\subfigure
{
\includegraphics[width=0.96\textwidth]{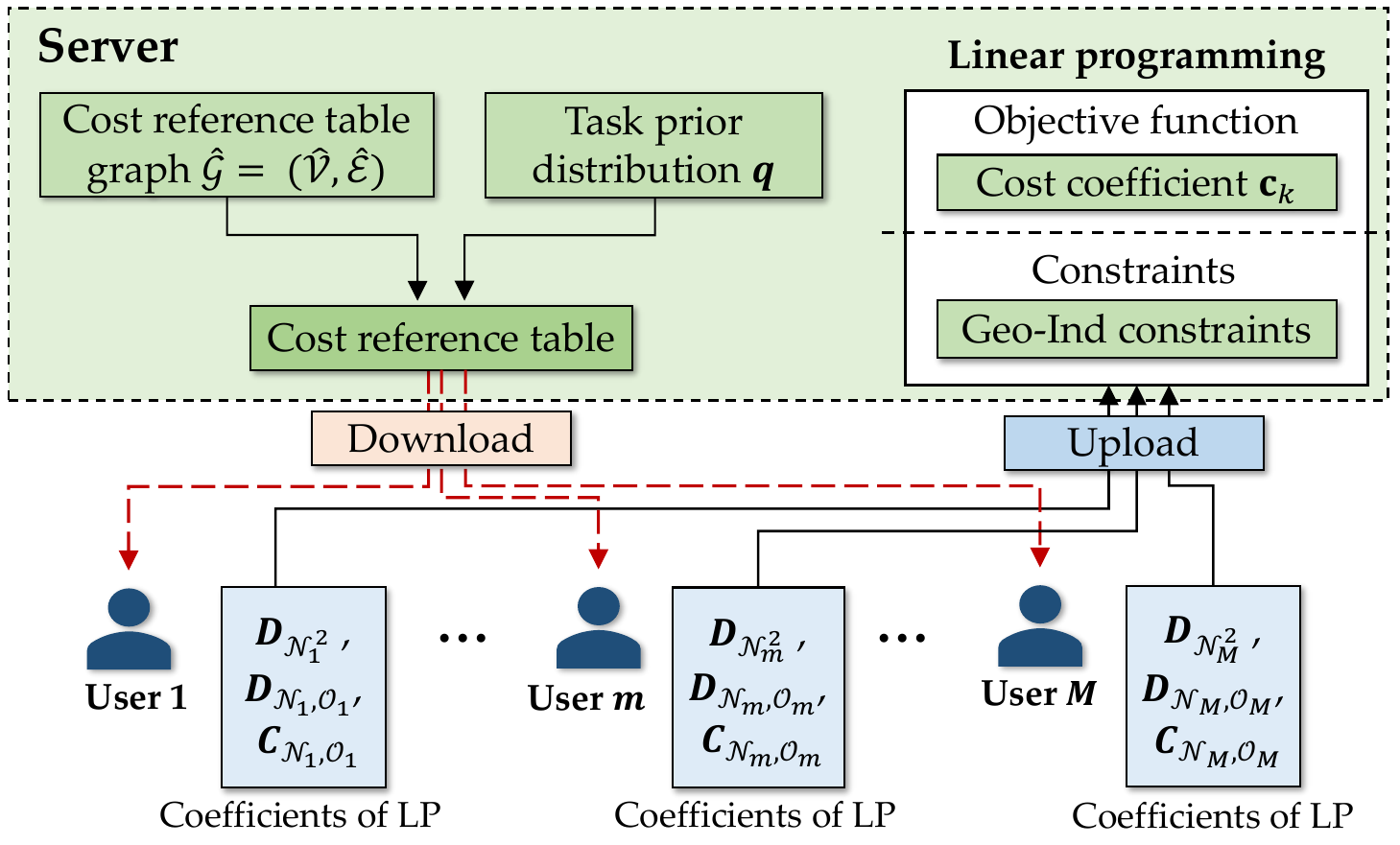}}
\end{minipage}
\vspace{-0.15in}
\caption{Obfuscation matrix calculation. }
\label{fig:coefficientcalculation}
\vspace{-0.15in}
\end{figure}

\vspace{0.03in}
\noindent \textbf{(2) Cost matrix} $\mathbf{C}_{\mathcal{N}_m, \mathcal{O}_m} = \left\{c_{v_i, v_k}\right\}_{(v_i, v_k) \in \mathcal{N}_m\times \mathcal{O}_m}$ includes the cost coefficients $c_{v_i, v_k}$ caused by each obfuscated location $v_k \in \mathcal{O}_m$ given each location $v_i \in \mathcal{N}_i$, from which the server can specify the objective function in Equ. (\ref{eq:LR-OMGobj}). Note that, to compute the cost coefficient $c_{v_i, v_k}$ in Equ. (\ref{eq:c}), it requires knowledge of the target locations $\mathbf{q} = [q_1, ..., q_K]$. However, retaining this information at the user's end presents challenges, primarily due to the dynamically changing target distribution that introduces additional communication costs. Moreover, privacy concerns, particularly regarding the confidentiality of targets (e.g., passengers in Uber-like platform \cite{Qiu-TMC2020}), further complicate the maintenance of such information.

To facilitate the estimation of $c_{v_i, v_k}$, as Fig. \ref{fig:coefficientcalculation} shows, we let each user download a \emph{cost reference} table maintained by the server. This table associates each pair $(v_i, v_k)$ with a value approximating the estimated cost caused by $v_k$ when the real location is $v_i$, all while keeping the actual target locations confidential. The process of constructing the cost reference table, ensuring the privacy of both users and targets, will be elaborated in \textbf{Section \ref{subsec:costreference}}. In this section, we assume that {\bf users can accurately obtain each $c_{v_i, v_k}$ using the cost reference table}. Further analysis of the performance guarantee utilizing the cost reference table will be presented in \textbf{Section \ref{subsec:perfanalysis}}. We will also illustrate that the cost matrix cannot be reversed to unveil the LR location of the user in \textbf{Section \ref{subsec:discussionthreats}} and experiment in Fig. \ref{fig:exp:threat}.

\vspace{-0.10in}
\subsection{Exponential Mechanism Integration} 
Given the coefficient matrices $\mathbf{D}_{\mathcal{N}^2_m}$ and $\mathbf{C}_{\mathcal{N}_m, \mathcal{O}_m}$ ($m = 1, .., M$), the server can compute $\mathbf{Z}_{\mathcal{N}_m}$ for each user. However, as each LR matrix is generated independently, the obfuscation vectors, such as $\mathbf{z}_i^{(m)}$ and $\mathbf{z}_j^{(n)}$ from different LR matrices $\mathbf{Z}_{\mathcal{N}_m}$ and $\mathbf{Z}_{\mathcal{N}n}$, may not satisfy the Geo-Ind constraints. Conversely, jointly deriving $\mathbf{Z}_{\mathcal{N}_1}, ..., \mathbf{Z}_{\mathcal{N}_M}$ using only LP not only incurs high computational overhead but also necessitates the disclosure of $\mathcal{N}_m$ and $\mathcal{O}_m$, which should be hidden from the server. 

{\revisiondone As a solution, we integrate the exponential geo-obfuscation mechanism into LR-Geo, allowing obfuscated locations that follow an exponential distribution to satisfy Geo-Ind across users without imposing the corresponding linear constraints in the LP formulation. Note that it is hard for the server to detect which obfuscation probabilities of the LR locations violate the Geo-Ind constraints, as checking such constraints require users' real locations, which are hidden from the server.} Instead, like \cite{ImolaUAI2022}, we let each user locally determine which obfuscated locations should follow the exponential mechanism. We define an indicator matrix $\mathbf{Q} = \left\{q_{i,k}\right\}_{(v_i, v_k)\in \mathcal{V}^2}$ to indicate whether the probability distribution of obfuscated location $v_k$ needs to follow the exponential distribution when the real location is $v_i$. If $q_{i,k} = 1$, we enforce the following constraint (exponential distribution) for each obfuscated location $v_k \in \mathcal{V}$ 
\begin{equation}
\label{eq:expo}
z_{i,k} = \left\{\begin{array}{ll} y_k e^{-\frac{\epsilon d_{v_i, v_k}}{2}} & \mbox{if $v_k \in \mathcal{O}_m$} \\ \xi = 
y_k e^{-\frac{\epsilon r_{\mathrm{obf}}}{2}} & \mbox{if $v_k \notin \mathcal{O}_m$} 
\end{array}\right.. 
\end{equation}
where $y_k \geq 0$ is a decision variable. In what follows, we let $\mathbf{y} = [y_1, ..., y_K]$. Note that in Equ. (\ref{eq:LR-OMGconstr3}) we have set $z_{i,k} = \xi$ when $v_k \notin \mathcal{O}_m$, and here in Equ. (\ref{eq:expo}), we specify $\xi = y_k e^{-\frac{\epsilon r_{\mathrm{obf}}}{2}}$. 

{\revisiondone Like \cite{ImolaUAI2022}, we adopt a heuristic strategy to determine the indicator matrix $\mathbf{Q}$. In particular, we assign $q_{i,k} = 1$ when the distance $d_{v_i, v_k}$ exceeds $r_{\mathrm{exp}}$, where $r_{\mathrm{exp}}$, a predefined threshold, is less than or equal to the obfuscation range $r_{\mathrm{obf}}$. This approach is based on the rationale of applying the exponential mechanism more extensively to obfuscated locations that are significantly distant from the actual location. Such locations are often associated with lower probability values, thereby minimizing their influence on the expected cost. Given that LP-based methods tend to yield lower costs, incorporating an exponential mechanism for these distant locations can reduce cost impacts more effectively. However, our framework can accommodate alternative methods for determining $\mathbf{Q}$. 

In our experimental setup in Section \ref{sec:performance}, we set $r_{\mathrm{obf}} = 4\text{km}$, resulting in a Geo-Ind violation ratio of up to 0.13\% (as shown in Fig. 10), i.e., the Geo-Ind constraint is likely to be satisfied with a sufficiently high ratio. \looseness = -1 }

\begin{proposition} 
\cite{ImolaUAI2022} Given any $y_k \in \mathbb{R}^+$, if the constraints in Equ. (\ref{eq:expo}) are satisfied, then for each pair of $z_{i,k}$ and $z_{j,k}$ with $q_{i,k} = q_{j,k} = 1$, the Geo-Ind constraint $z_{i,k}  - e^{{\epsilon d_{v_i, v_j}}}  z_{j,k} \leq 0$ is satisfied (The detailed proof can be found in the proof of \textbf{Proposition 1} in \cite{ImolaUAI2022}). 
\end{proposition}

To enable the server to establish the constraints of the exponential distribution in Equ. (\ref{eq:expo}), each user $m$ computes the distance matrix $\mathbf{D}_{\mathcal{N}m, \mathcal{O}m} = \left\{d_{v_i, v_k}\right\}_{(v_i, v_k) \in \mathcal{N}_m\times \mathcal{O}_m}$ and uploads the matrix to the server. It's important to note that $\mathbf{D}_{\mathcal{N}_m, \mathcal{O}_m}$ contains only the relative distances between locations within $\mathcal{N}_m$ and $\mathcal{O}_m$, and does not provide enough information to deduce the specific locations in either $\mathcal{N}_m$ or $\mathcal{O}_m$.

\vspace{0.05in}
\noindent\textbf{Problem formulation}. After collecting 
the coefficient matrices $\mathbf{D}_{\mathcal{N}^2_m}$, $\mathbf{D}_{\mathcal{N}_m, \mathcal{O}_m}$ and $\mathbf{C}_{\mathcal{N}_m, \mathcal{O}_m}$ ($m = 1, .., M$) and adding the constraints of exponential mechanism in Equ. (\ref{eq:expo}), we can formulate the following \emph{Central LR-Geo (CLR-Geo)} problem at the server side: 
\begin{eqnarray}
\label{eq:OMGLRobj}
\min && \textstyle 
\sum_{m=1}^M\mathcal{L}\left(\mathbf{Z}_{\mathcal{N}_m}\right) \\ \label{eq:OMGLRconstr}
\mbox{s.t. } && \mbox{Equ. (\ref{eq:LR-OMGconstr1})(\ref{eq:LR-OMGconstr2}) are satisfied for each $\mathcal{N}_m$}\\ \label{eq:LR-OLMGexpoconstr} 
&& \mbox{Equ. (\ref{eq:expo}) is satisfied $\forall v_k \in \mathcal{V}$ with $q_{i,k} = 1$.}
\vspace{-0.00in}
\end{eqnarray}

\subsection{Benders' Decomposition to Enhance Computation Efficiency}
\label{subsec:Benders}

\subsubsection{Problem reformulation of CLR-Geo} We rewrite the objective function in Equ. (\ref{eq:OMGLRobj}) as 
\begin{eqnarray}
\nonumber 
\textstyle \sum_{m=1}^M\mathcal{L}\left(\mathbf{Z}_{\mathcal{N}_m}\right) &=& \textstyle 
 \sum_{m=1}^M\sum_{v_i\in \mathcal{N}_m} \sum_{v_k \in \mathcal{V}}\underbrace{c_{v_i, v_k}z_{i,k}q_{i,k}}_{\small \mbox{following Equ. (\ref{eq:expo})}} \\ \nonumber 
&+& \textstyle  \sum_{m=1}^M\sum_{v_i\in \mathcal{N}_m}\sum_{v_k \in \mathcal{V}}c_{v_i, v_k}z_{i,k}(1-q_{i,k}) \\ 
&=& \textstyle  \sum_{k=1}^K \alpha_k y_k + \sum_{m=1}^M \mathbf{c}'_{\mathcal{N}_m}\mathbf{z}'_{\mathcal{N}_m}
\end{eqnarray}
where each $\alpha_k = \sum_{m=1}^M\sum_{v_i\in \mathcal{N}_m} q_{i,k} c_{v_i, v_k} e^{-\frac{\epsilon d_{v_i, v_k}}{2}}$ is a constant, and in $\mathbf{c}'_{\mathcal{N}_m}$, each $c'_{i,k} = c_{v_i, v_k}z_{i,k}(1-q_{i,k})$. 

We rewrite the constraints of Equ. (\ref{eq:LR-OMGconstr1})(\ref{eq:LR-OMGconstr2}) and Equ. (\ref{eq:expo}) as 
\begin{equation}
\label{eq:LR-OMGconstrrewrite}
\mathbf{A}_{\mathcal{N}_m}\mathbf{z}'_{\mathcal{N}_m} + \mathbf{B}_{\mathcal{N}_m}\mathbf{z}''_{\mathcal{N}_m}\left(\mathbf{y}\right) \geq \mathbf{b}_{\mathcal{N}_m}
\end{equation}
where 
\vspace{-0.1in}
\begin{equation}
\mathbf{A}_{\mathcal{N}_m} = \left[\begin{array}{c}\mathbf{A}^{\mathrm{GeoI}}_{\mathcal{N}_m}\\ \mathbf{A}^{\mathrm{unit}}_{\mathcal{N}_m}\\
-\mathbf{A}^{\mathrm{unit}}_{\mathcal{N}_m}\end{array}\right],  \mathbf{B}_{\mathcal{N}_m} = \left[\begin{array}{c}\mathbf{B}^{\mathrm{GeoI}}_{\mathcal{N}_m}\\ \mathbf{B}^{\mathrm{unit}}_{\mathcal{N}_m} \\ 
-\mathbf{B}^{\mathrm{unit}}_{\mathcal{N}_m}\end{array}\right], \mathbf{b}_{\mathcal{N}_m} = \left[\begin{array}{c}\mathbf{b}^{\mathrm{GeoI}}_{\mathcal{N}_m}\\ \mathbf{b}^{\mathrm{unit}}_{\mathcal{N}_m} \\ 
-\mathbf{b}^{\mathrm{unit}}_{\mathcal{N}_m}\end{array}\right]
\end{equation}
and (1) $\mathbf{z}'_{\mathcal{N}_m}$ (resp. $\mathbf{z}''_{\mathcal{N}_m}\left(\mathbf{y}\right)$) includes the obfuscation probabilities $z_{i,k}$ \emph{not adhering to} (resp. \emph{adhering to}) the exponential mechanism, where $v_i \in \mathcal{N}_m$ (note that each entry in $\mathbf{z}''_{\mathcal{N}_m}\left(\mathbf{y}\right)$ follows Equ. (\ref{eq:expo}), therefore the vector is written as a function of $\mathbf{y}$).  
\newline (2) $\mathbf{A}^{\mathrm{GeoI}}_{\mathcal{N}_m}$ (resp. $\mathbf{B}^{\mathrm{GeoI}}_{\mathcal{N}_m}$) denotes the coefficient matrix of the \emph{Geo-Ind constraints} for $\mathbf{z}'_{\mathcal{N}_m}$ (resp. $\mathbf{z}''_{\mathcal{N}_m}\left(\mathbf{y}\right)$); 
\newline (3) $\mathbf{A}^{\mathrm{unit}}_{\mathcal{N}_m}$ (resp. $\mathbf{B}^{\mathrm{unit}}_{\mathcal{N}_m}$) denotes the coefficient matrix of the \emph{unit measure constraints} for $\mathbf{z}'_{\mathcal{N}_m}$ (resp. $\mathbf{z}''_{\mathcal{N}_m}\left(\mathbf{y}\right)$);
\newline (4) $\mathbf{b}^{\mathrm{GeoI}}_{\mathcal{N}_m}$ and $\mathbf{b}^{\mathrm{unit}}_{\mathcal{N}_m}$ are the right hand side coefficient vectors of the \emph{Geo-Ind constraints} and the \emph{unit measure constraints}, respectively.

As Fig. \ref{fig:blockstructure} shows, the constraint matrix of the reformulated problem has a \emph{block ladder structure}, lending the problem well to \emph{Benders decomposition (BD)} \cite{Rahmaniani-EJOR2017}. Due to the limit of space, we list the detailed formulations of $\mathbf{A}^{\mathrm{GeoI}}_{\mathcal{N}_m}$, $\mathbf{B}^{\mathrm{GeoI}}_{\mathcal{N}_m}$, $\mathbf{A}^{\mathrm{unit}}_{\mathcal{N}_m}$, $\mathbf{B}^{\mathrm{unit}}_{\mathcal{N}_m}$, and coefficient vectors $\mathbf{b}^{\mathrm{GeoI}}_{\mathcal{N}_m}$, and $\mathbf{b}^{\mathrm{unit}}_{\mathcal{N}_m}$ in \textbf{Section \ref{sec:Bendersnotations}} in Appendix. 

\subsubsection{High level idea of the algorithm} Benders' decomposition is an optimization method that breaks a complex problem into two simpler stages: a \emph{master program} and a set of \emph{subproblems}.
\newline \textbf{Stage 1: Master Program (MP)} - The MP focuses on deriving the set of decision variables $\{y_1, ..., y_K\}$. It simplifies the problem by replacing the cost terms $\mathbf{c}'_{\mathcal{N}_m}\mathbf{z}'_{\mathcal{N}_m}$ associated with data perturbation in each user $m$ with a single decision variable $w_m$. The MP is formulated as a linear programming problem that minimizes the sum cost of the primary decision variables $\sum_{k=1}^K \alpha_k y_k$ and the substituted cost terms $\sum_{m=1}^M w_m$. The MP iteratively guesses the values of $w_m$ and uses them in the subsequent stage to guide optimization.
\newline \textbf{Stage 2: Subproblems} - Each subproblem $\mathsf{sub}_m$ validates the guessed values of $w_m$ generated by the MP considering its own constraint $\mathbf{A}_{\mathcal{N}_m}\mathbf{z}'_{\mathcal{N}_m} \geq \mathbf{b}_{\mathcal{N}_m} - \mathbf{B}_{\mathcal{N}_m}\mathbf{z}''_{\mathcal{N}_m}\left(\mathbf{y}\right)$. If the guessed value is not optimal, the subproblem suggests a new constraint (or cut) to refine the MP's future guesses. As shown in Fig. \ref{fig:blockstructure}, the constraint matrix $\mathbf{A}_{\mathcal{N}_m}$ of different subproblems are disjoint, allowing them to validate the guessed values of $w_m$ in parallel. This iterative process continues, with each subproblem helping to progressively narrow the solution space until the optimal solution is found.

\begin{figure}[t]
\centering
\hspace{-0.0in}
\begin{minipage}{0.50\textwidth}
\subfigure
{
\includegraphics[width=0.92\textwidth]{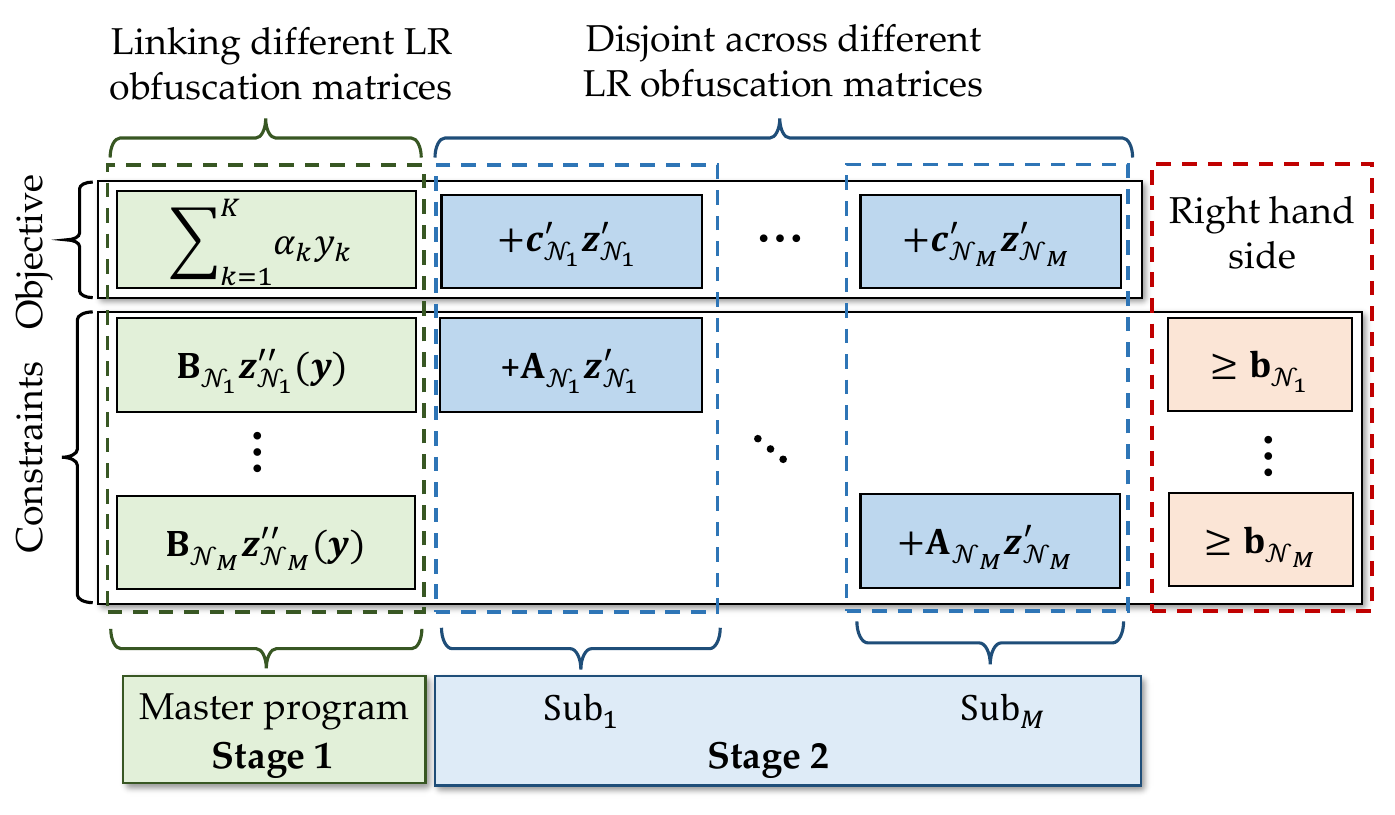}}
\end{minipage}
\vspace{-0.05in}
\caption{Block ladder structure of the CLR-Geo problem. }
\label{fig:blockstructure}
\vspace{-0.05in}
\end{figure}

The method iteratively refines the solutions by collecting these cuts from subproblems, guiding the MP towards the optimal solution over time. This approach is particularly useful for problems where the main challenge is efficiently managing computational complexity by breaking it down into more manageable parts. Due to the limit of space, the detailed description of the algorithm is given in Section \ref{sec:BDdetaileddescription} in Appendix.


\vspace{-0.03in}
{\revisiondone \subsubsection{Time complexity of BD} 
The MP and the subproblems are both formulated as LP problems, with time complexity depending on the number of decision variables \cite{Cohen-STOC2019}. To facilitate analysis, we use a monotonically increasing function $O(T(n))$ to represent the time complexity of LP given the number of decision variables $n$. 

In each iteration, the MP in Stage 1 has $O(T(|\mathcal{V}| +M))$ time complexity, including $|\mathcal{V}|$ decision variables to determine $y_1, \ldots, y_K$ and $M$ decision variables $w_1, \ldots, w_M$ to ``guess'' the utility loss of all the users. Each subproblem $m$ in Stage 2 (run in parallel) has $O(T(|\mathcal{V}||\mathcal{N}_m|^2))$ time complexity to validate the guessed value of $w_m$ generated by the MP.  Since Stage 2 is terminated only after all subproblems are completed, its time complexity is $O(\max_{l} T(|\mathcal{V}||\mathcal{N}_m|^2))$ $= O(T( |\mathcal{V}|\max_{m}|\mathcal{N}_m|^2))$, assuming all the subproblems are run in parallel. Hence, the total computation time of each iteration is given by $O(T(|\mathcal{V}| +M) + T(|\mathcal{V}|\max_{m} |\mathcal{N}_m|^2))$. Considering that both $O(|\mathcal{V}|+M)$ and $O(|\mathcal{V}|\max_{m} |\mathcal{N}_m|^2)$ are much smaller than the number of decision variables in the original OMG, $O(|\mathcal{V}|^3)$, the time complexity of each iteration of BD is significantly lower than that of the original PMO. {\rd The remaining question} is \emph{how many iterations are needed to converge the solution to the optimal. }} 

\vspace{0.03in}
\noindent \textbf{Convergence of Benders decomposition}.
\vspace{-0.05in}
\begin{proposition}
\label{prop:BDbound}
(\emph{Upper and lower bounds of CLR-Geo's optimal}) \cite{Rahmaniani-EJOR2017}
(1) Because the MP relaxes the constraints of the original CLR-Geo (Equ. (\ref{eq:OMGLRobj})--(\ref{eq:LR-OLMGexpoconstr})), the optimal solution of the MP (Equ. (\ref{eq:MPObj}) -- (\ref{eq:MPzy0})) offers a \textbf{lower bound} of the optimal solution of CLR-Geo. 
\newline (2) If the solution of the subproblems (Equ. (\ref{eq:dualobj})-(\ref{eq:dualconstr1})) exists, then by combining the solutions of the subproblems and the solution of the MP, we obtain a feasible solution of CLR-Geo, which is an \textbf{upper bound} of the optimal solution.
\end{proposition}
\vspace{-0.05in}

\begin{wrapfigure}{r}{0.230\textwidth}
\vspace{-0.00in}
\begin{minipage}{0.230\textwidth}
\centering
    \subfigure{
\includegraphics[width=1.00\textwidth]{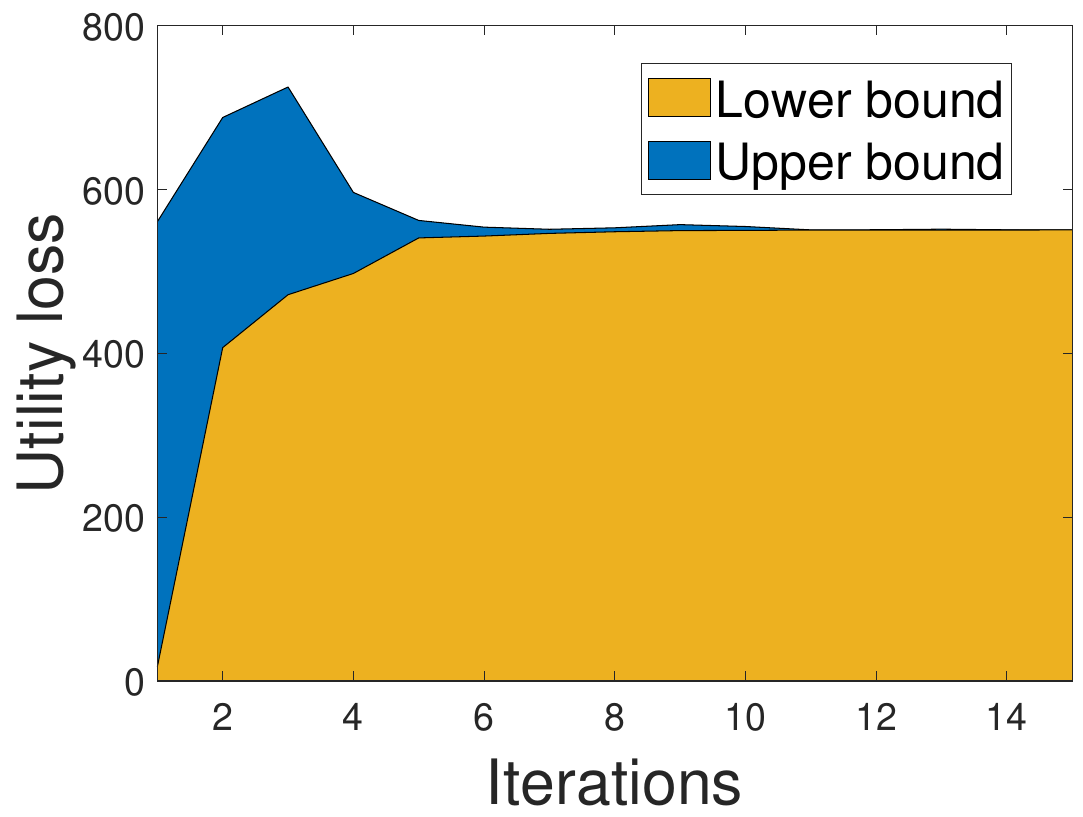}}
\vspace{-0.25in}
\caption{An example of Benders' convergence.}
\vspace{-0.05in}
\label{fig:BDconvergenceEg}
\end{minipage}
\end{wrapfigure}
The optimal solution for the CLR-Geo necessarily resides within the interval demarcated by the upper and lower bounds as delineated in \textbf{Proposition \ref{prop:BDbound}}. The narrower this interval, the nearer the solution derived from the Benders' Decomposition is to the optimal solution. Fig. \ref{fig:BDconvergenceEg} illustrates the evolution of these bounds throughout the iterative process. 
Considering a prolonged convergence tail,  we opt to conclude the algorithm once the discrepancy between the optimal upper and lower bounds diminishes to less than a specified margin, $\xi$ (e.g., we set $\xi = 0.01$km in our experiment in Section \ref{sec:performance}).

\subsection{Cost Matrix Estimation}
\label{subsec:costreference}
In Sections \ref{subsec:computeframework} and \ref{subsec:Benders}, we assumed that each user $m$ knows the cost matrix $\mathbf{C}_{\mathcal{N}_m, \mathcal{O}m}$. We now relax this assumption and elucidate the methodology by which users, with the assistance of the server, can estimate $\mathbf{C}_{\mathcal{N}_m, \mathcal{O}_m}$.

According to Equ. (\ref{eq:c}), the calculation of each $c_{v_i, v_k}$ requires
\begin{itemize}
\item the coordinates of the locations in $\mathcal{N}_{m}$ and $\mathcal{O}_{m}$  (to derive cost error $\delta_{v_i,v_k,v_l}$), which are \textbf{known by the user but unknown by the server}; 
\item the coordinates of the target locations (to derive cost error $d_{v_i,v_k,v_l}$) and the targets' prior distribution $\mathbf{q}$, which are \textbf{known by the server but unknown by the user}. 
\end{itemize}
Since both the server and the user possess only partial information required to compute $\mathbf{c}_k$, a ``cooperative'' approach is employed to calculate $\mathbf{c}_k$ through the exchange of intermediate values between the two parties. \emph{Throughout the calculation of $\mathbf{c}_k$, the server must remain unaware of $\mathcal{N}_{m}$ and $\mathcal{O}_{m}$ to protect privacy}. Since the server has the global information including the traveling cost between any pair of locations and the target distribution, we let the server generate a \emph{cost reference table} to assist the user estimate $\mathbf{c}_k$.

\subsubsection{Cost reference table} The server constructs a discrete set of locations $\hat{\mathcal{V}}$, which is sufficiently dense to ensure that for any given location pair $(v_j, v_k)$ from the sets $\mathcal{N}_{m}$ and $\mathcal{O}_{m}$, respectively, users can identify a corresponding pair $(\hat{v}_j, \hat{v}_k)$ within $\hat{\mathcal{V}}$ that is closely approximated to $(v_j, v_k)$. This approximation is then used to estimate cost coefficients. To establish $\hat{\mathcal{V}}$, the server employs a grid map to discretize the location field, such that locations within the same grid cell are indistinguishable. For specific applications involving constrained user mobility, such as vehicles' mobility, the server might alternatively divide the road network into segments, treating each as a distinct location \cite{Qiu-TMC2020}. Given the superior computational resources available to servers compared to those of users, $\hat{\mathcal{V}}$ can afford to feature a finer granularity of location discretization than that of $\mathcal{N}_{m}$ and $\mathcal{O}_{m}$.

\vspace{0.03in}
\noindent \textbf{Table format}. As Fig. \ref{fig:costreferenceTable} shows, each row of the cost reference table includes 
the coordinates of an ``approximated'' real location $\hat{v}_i$, an ``approximated'' obfuscated location $\hat{v}_k$, and the expected cost $\beta_{i,k}$ given the real and the obfuscated locations $\hat{v}_i$ and $\hat{v}_k$, respectively. The expectation of the cost $\beta_{i,k}$ is to take over all possible target locations $
\beta_{i,k} = \sum_{j=1}^Q q_j\delta_{\hat{v}_i, \hat{v}_k, \hat{v}_j}$, where $q_j$ is the probability that the target's nearest location in $\hat{\mathcal{V}}$ is $\hat{v}_j$. 

In what follows, we use $\hat{\mathbf{C}}_{\mathcal{N}_m, \mathcal{O}_m} = \left\{\hat{c}_{v_i,v_k}\right\}_{(v_i, v_k) \in \mathcal{N}_m\times \mathcal{O}_m}$ to denote the cost matrix estimated by the cost reference table. 
When a user estimates each $\hat{c}_{v_i,v_k}$ in $\hat{\mathbf{C}}_{\mathcal{N}_m, \mathcal{O}_m}$, the user first finds the nearest locations of $v_i$ and $v_k$ in  $\hat{\mathcal{V}}$, denoted by $\hat{v}_i$ and $\hat{v}_k$ respectively, and then calculates the cost estimation error $\hat{c}_{v_i,v_k}$ by 
\vspace{-0.00in}
\begin{equation}
\label{eq:costestimate}
\hat{c}_{v_i,v_k} = p_i \left(\beta_{i,k} + d_{v_i, \hat{v}_i} + d_{v_k, \hat{v}_k}\right), 
\vspace{-0.00in}
\end{equation}
which gives an upper bound of the real $c_{v_i, v_k}$ (see \textbf{Lemma \ref{lem:upperbounds}} in Appendix). 
\vspace{-0.00in}

\vspace{-0.00in}

\begin{figure}[t]
\centering
\hspace{-0.0in}
\begin{minipage}{0.48\textwidth}
\centering
  \subfigure{
\includegraphics[width=1.00\textwidth]{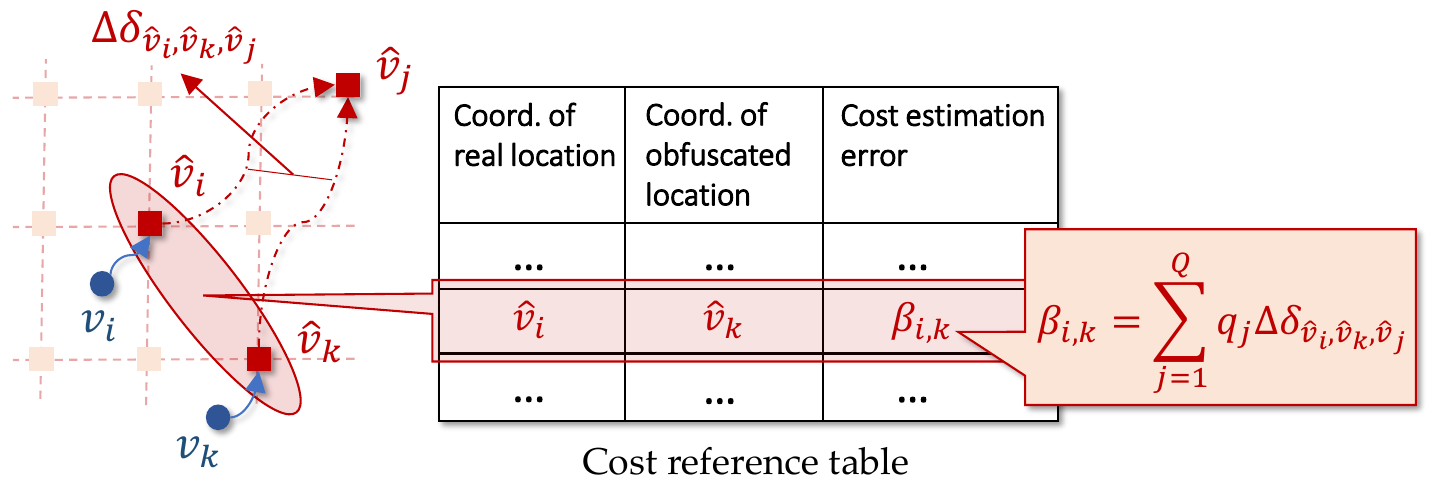}}
\end{minipage}
\vspace{-0.05in}
\caption{Cost reference table. }
\label{fig:costreferenceTable}
\vspace{-0.05in}
\end{figure}


\DEL{
\begin{algorithm}[h]
\SetKwFunction{push}{push}
\SetKwFunction{pop}{pop}
\SetKwFunction{top}{top}
\SetKwInOut{Input}{Input}
\SetKwInOut{Output}{Output}
\Input{Cost reference table that covers the discrete location set $\hat{\mathcal{V}}$}
\Output{cost coefficient matrix $\hat{\mathbf{C}}_{\mathcal{N}_m, \mathcal{O}_m}$}
\For{each $(v_i, v_k) \in \mathcal{N}_{m}\times \mathcal{O}_{m}$}{
    $\hat{v}_i \leftarrow \arg\min_{\hat{v}_l \in \hat{\mathcal{V}}}d_{\hat{v}_i, \hat{v}_l}$\; 
    $\hat{v}_k \leftarrow \arg\min_{\hat{v}_l \in \hat{\mathcal{V}}}d_{\hat{v}_k, \hat{v}_l}$\;
    $\hat{c}_{v_i,v_k} \leftarrow p_i \left(\beta_{i,k} - \tilde{d}_{v_i, \hat{v}_i} - \tilde{d}_{v_k, \hat{v}_k}\right)$\;
}
\Return $\hat{\mathbf{C}}_{\mathcal{N}_m, \mathcal{O}_m}$\; 
\normalsize
\caption{Cost coefficient estimation using the cost reference table. }
\label{al:costestimate}
\end{algorithm}
\vspace{-0.00in}}

{\revisiondone \noindent \textbf{Time complexity of $\hat{\mathbf{C}}_{\mathcal{N}_m, \mathcal{O}_m}$'s estimation}. The cost estimation traverses the $|\mathcal{N}_{m}||\mathcal{O}_{m}|$ pairs of locations in $\mathcal{N}_{m}\times \mathcal{O}_{m}$, and for each location pair, it needs to find their closest locations in $\hat{\mathcal{V}}$, taking $O(|\hat{\mathcal{V}}|)$ comparisons. Therefore, the time complexity of cost estimation is $O(|\mathcal{N}_{m}||\mathcal{O}_{m}||\hat{\mathcal{V}}|)$.}

\DEL{
\vspace{0.00in}
\begin{wrapfigure}{r}{0.27\textwidth}
\vspace{-0.00in}
\begin{minipage}{0.265\textwidth}
\centering
    \subfigure{
\includegraphics[width=1.00\textwidth]{./fig/exp/cost reference tableattack}}
\vspace{-0.00in}
\caption{Potential attack of cost reference table. }
\label{fig:costreferenceTableattack}
\end{minipage}
\vspace{-0.00in}
\end{wrapfigure}}

\subsubsection{Cost reference table generation}
To generate a cost reference table, the server first creates a \emph{weighted directed graph} $\tilde{\mathcal{G}} = \left(\hat{\mathcal{V}}, \hat{\mathcal{E}}\right)$ to describe the traveling cost between locations in the discrete location set $\hat{\mathcal{V}}$, where each pair of \emph{adjacent} locations $\hat{v}_i$ and $\hat{v}_j$ are connected by an edge $\hat{e}_{i,j} \in \hat{\mathcal{E}}$. Each $\hat{e}_{i,j} \in \hat{\mathcal{E}}$ is assigned a weight $d_{\hat{v}_i,\hat{v}_j}$, reflecting the traveling cost from $\hat{v}_i$ to $\hat{v}_j$. The server then builds the SP tree rooted at each target location $\hat{v}_j \in \hat{\mathcal{V}}$ in $\tilde{\mathcal{G}}$, based on which it calculates the shortest path distance $D_{\hat{v}_i, \hat{v}_j}$ ($D_{\hat{v}_i, \hat{v}_k}$ resp.) from each location $\hat{v}_i$ ($\hat{v}_k$ resp.) to $\hat{v}_j$,  and derives $\delta_{\hat{v}_i,\hat{v}_k,\hat{v}_j}$ using Equ. (\ref{eq:deltad}). Finally, the server calculates $\beta_{i,k}$ for each $(\hat{v}_i,\hat{v}_k)$ using their cost estimation errors  $\delta_{\hat{v}_i,\hat{v}_k,\hat{v}_j}$ 

\vspace{0.03in}
{\revisiondone \noindent \textbf{Time complexity of cost reference table generation}. The construction of each SP tree can be achieved by Dijkstra's algorithm, which has a time complexity of $O(|\hat{\mathcal{V}}|^2)$ \cite{Algorithm}. For each designated target location, the server is required to generate an SP tree, culminating in a collective computational effort of $O(|\hat{\mathcal{V}}|^3)$ operations. Furthermore, the computation of $\beta_{i,k}$ is called $|\hat{\mathcal{V}}|^2$ times, with each instance necessitating $O(|\hat{\mathcal{V}}|)$ operations, thereby rendering its time complexity to be $O(|\hat{\mathcal{V}}|^3)$. Consequently, the overall time complexity associated with the creation of the requisite table is determined to be $O(|\hat{\mathcal{V}}|^3+|\hat{\mathcal{V}}|^3) = O(|\hat{\mathcal{V}}|^3)$, indicating significant computational demand for these operations.}


\subsubsection{Cost reference table size reduction} We can to further reduce the location set $\hat{\mathcal{V}}$ while guaranteeing the accuracy of the cost coefficient $\hat{\mathbf{c}}_k$ estimation. To achieve this, we define $\hat{\mathcal{V}}$ in a circular region, referred to as $\mathcal{C}_{\mathrm{cr}}$, so that $\hat{\mathcal{V}}$ consists of locations within this circle. Consequently, the accuracy of $\hat{\mathbf{c}}_k$ estimation can be guaranteed if $\mathcal{C}_{\mathrm{cr}}$ encompasses both $\mathcal{N}_{m}$ and $\mathcal{O}_{m}$.

Note that if   $\mathcal{C}_{\mathrm{cr}}$ covers $\mathcal{N}_{m}$ and $\mathcal{O}_{m}$ only at a minimum level, the range of $\mathcal{N}_{m}$ and $\mathcal{O}_{m}$ can be possibly disclosed to the server. 
Therefore, instead of only covering $\mathcal{N}_{m}$ and $\mathcal{O}_{m}$, we allow the user to request a larger $\mathcal{C}_{\mathrm{cr}}$. Initially, the user randomly selects a location $v_a$ from the LR location set $\mathcal{N}_{m}$ by following a uniform distribution. Subsequently, the user reports a circle $\mathcal{C}(v_a, \max{2\Gamma, \Gamma+r_{\mathrm{obf}}})$ to the server as the \emph{requested range of cost reference table}, with $v_a$ serving as the center. 

\vspace{-0.07in}
\begin{proposition}
\label{prop:cost reference tablerange}
The circle $\mathcal{C}(v_a, \max\{2\Gamma, \Gamma+r_{\mathrm{obf}}\})$ covers all the locations in both $\mathcal{N}_{m}$ and $\mathcal{O}_{m}$. 
\end{proposition}
\vspace{-0.07in}

\DEL{
\begin{proof}
First, since $d_{v_i, v_j} \leq d_{v_i, v_j}$, and $d_{v_i, v_j} \leq \Gamma$, $\forall v_j\in \mathcal{N}_{m}$, we obtain that $d_{v_i, v_j} \leq \Gamma, ~ \forall v_j \in \mathcal{N}_{m}$. 

Also, according to Equ. (\ref{eq:O}), we have $d_{v_i, v_j} \leq r_{\mathrm{obf}}, ~ \forall v_j \in \mathcal{O}_{m}$, which implies that 
$d_{v_i, v_j} \leq \max\{\Gamma, r_{\mathrm{obf}}\}, ~ \forall v_j \in \mathcal{N}_{m} \cup \mathcal{O}_{m}$. $d_{v_i, v_a} \leq \Gamma$ because $v_a$ is selected within the circle $\mathcal{C}(v_a, \Gamma)$. Then, according to the triangle inequality, 
\begin{eqnarray}
d_{v_a, v_j} \leq d_{v_i, v_a} + d_{v_i, v_j}leq  \Gamma + \max\{\Gamma, r_{\mathrm{obf}}\} = \max\{2\Gamma, \Gamma + r_{\mathrm{obf}}\}, 
\end{eqnarray}
for each $v_j \in \mathcal{N}_{m} \cup \mathcal{O}_{m}$, 
indicating that $\mathcal{C}(v_a, \max\{2\Gamma, \Gamma+r_{\mathrm{obf}}\})$ covers both $\mathcal{N}_{m}$ and $\mathcal{O}_{m}$. 
\end{proof}}
\vspace{-0.00in}
Moreover, according to the requested range $\mathcal{C}(v_a, \max\{2\Gamma, \Gamma+r_{\mathrm{obf}}\})$ and how the user selects the location $v_a$, the server can only infer that the user's real location is in the circle $\mathcal{C}(v_a, \Gamma)$,  where $\Gamma > \gamma$, indicating that the user's location is well hidden from the server. \looseness = -1

\vspace{-0.05in}
\subsection{Performance Analysis}
\label{subsec:perfanalysis}
In this section, we provide the theoretical analysis of the performance for the LR-Geo solution (Equ. (\ref{eq:OMGLRobj})--(\ref{eq:LR-OLMGexpoconstr})), including the \emph{privacy guarantee} in \textbf{Theorem \ref{thm:privacyguarantee}}, the lower bound and the upper bound of \emph{expected cost} in \textbf{Theorem  \ref{thm:lowerbounds}} and \textbf{Theorem \ref{thm:upperbounds}}, respectively. The detailed proofs of these theorems can be found in \textbf{Section \ref{sec:proofs} in Appendix}.

\vspace{-0.05in}
\subsubsection{Privacy guarantee} We first prove that the chosen obfuscated locations adhering to the exponential distribution constraints (Equ. (\ref{eq:expo})) meet Geo-Ind constraints across users even though their geo-obfuscation is calculated in a relatively independent manner.   
\begin{theorem}
\label{thm:privacyguarantee}
(Privacy guarantee) Given two locations, $v_i$ from user $n$'s LR location set $\mathcal{N}_n$ and $v_j$ from user $m$'s LR location set $\mathcal{N}_m$, if both locations satisfy the condition $q_{i,k} = q_{j,k} = 1$ (indicating that $z^{(n)}_{i,k}$ and $z^{(m)}_{j,k}$ satisfy the exponential distribution constraints), then their obfuscation distributions still satisfy the $(\epsilon, \gamma)$-Geo-Ind constraints, 
\vspace{-0.00in}
\begin{equation}
z^{(n)}_{i,k}  - e^{{\epsilon d_{v_i, v_j}}}  z^{(m)}_{j,k} \leq 0,~\forall v_k \in \mathcal{V}. 
\vspace{-0.10in}
\end{equation}
\end{theorem}

{\revisiondone 
\begin{proposition}
\label{prop:privacyguarantee}
For any LR location $v_i$ of a User $n$, i.e., $v_i \in \mathcal{N}_n$, let $\mathcal{A}_i$ denote the set of obfuscated locations of $v_i$ that adhere to the exponential distribution constraints. Then:
\begin{itemize}
    \item[(1)] For each pair of Users $n$ and $m$, the Geo-Ind constraint violation ratio is upper bounded by: 
    \begin{equation}
    \small 1- \frac{2\sum_{(v_i, v_j) \in \mathcal{N}_n \times \mathcal{N}_m} |\mathcal{A}_i \cap \mathcal{A}_j|+(|\mathcal{N}_n|^2+|\mathcal{N}_m|^2-|\mathcal{N}_n|-|\mathcal{N}_m|)|\mathcal{V}|}{(|\mathcal{N}_n|+|\mathcal{N}_m|)(|\mathcal{N}_n|+|\mathcal{N}_m|-1)|\mathcal{V}|}.
    \end{equation}
    \item[(2)] For all users, the Geo-Ind constraint violation ratio is upper bounded by:
    \begin{equation}
    \small 1-\frac{2\sum_{n=1}^M \sum_{m=n+1}^M \sum_{(v_i, v_j) \in \mathcal{N}_n \times \mathcal{N}_m} |\mathcal{A}_i \cap \mathcal{A}_j|+ \sum_{n=1}^M\left(|\mathcal{N}_n|^2 - |\mathcal{N}_n|\right)|\mathcal{V}|}{\sum_{n=1}^M|\mathcal{N}_n|\left(\sum_{n=1}^M|\mathcal{N}_n|-1\right)|\mathcal{V}|}.
    \end{equation}
\end{itemize}
\end{proposition}}

\begin{remark}
According to Proposition \ref{prop:privacyguarantee}(1), the Geo-Ind violation ratio for each user $i$ also depends on the proportion of obfuscated locations from other users that follow the exponential mechanism. For example, there is no Geo-Ind guarantee for a user when no other user applies exponential mechanism. Hence, our system encourages users to cooperate to achieve low Geo-Ind violation ratio for both sides. While $(\epsilon, \gamma)$-Geo-Ind is not guaranteed between $z^{(n)}_{i,k}$ and $z^{(m)}_{j,k}$ if one of them doesn't follow the exponential mechanism, our experimental findings in Fig. \ref{fig:exp:GVR} demonstrate that unselected obfuscated locations still possess a high probability {\revisiondone (99.81\% on average)} of meeting Geo-Ind constraints.
\end{remark}

\subsubsection{Lower bound and upper bound of the expected cost} Given the LR location set $\mathcal{N}_m$, we formulate the following relaxed LR-Geo problem: 
\vspace{-0.1in}
\begin{eqnarray}
\label{eq:OMGLRobjguarantee}
\min && \textstyle 
\sum_{m=1}^M\mathcal{L}\left(\mathbf{Z}_{\mathcal{N}_m}\right) \\ \label{eq:OMGLRconstr}
\mbox{s.t. } && \mbox{Equ. (\ref{eq:LR-OMGconstr1})(\ref{eq:LR-OMGconstr2}) are satisfied for each $\mathcal{N}_m$}
\vspace{-0.00in}
\end{eqnarray}


\begin{theorem}
\label{thm:upperbounds}
(Upper bound of the minimum expected cost) 
Using the estimated cost coefficient  $\hat{c}_{v_i,v_k}$ in Equ. (\ref{eq:costestimate}), the solution of the CLR-Geo problem in Equ. (\ref{eq:OMGLRobj})--(\ref{eq:LR-OLMGexpoconstr}) offers an upper bound of the minimum expected cost.
\end{theorem}

Next, we define another cost estimation $\tilde{c}_{v_i,v_k}$ by 
\vspace{-0.05in}
\begin{equation}
\label{eq:costestimatelower}
\tilde{c}_{v_i,v_k} = p_i \left(\beta_{i,k} - d_{v_i, \hat{v}_i} - d_{v_k, \hat{v}_k}\right), 
\vspace{-0.05in}
\end{equation}
which gives a lower bound of the real $c_{v_i, v_k}$ (see \textbf{Lemma \ref{lem:lowerbounds} in Appendix}). 

\begin{theorem}
\label{thm:lowerbounds}
Using the estimated cost in Equ. (\ref{eq:costestimatelower}), the solution of the relaxed LR-Geo problem in Equ. (\ref{eq:OMGLRobjguarantee})--(\ref{eq:OMGLRconstr}) offers a lower bound of the minimum expected cost. \looseness = -1
\end{theorem}

\subsection{Discussion of Potential Inference Using Estimated Cost Matrix}
\label{subsec:discussionthreats}

In this part, we illustrate that it is hard to infer the locations in $\mathcal{N}_{m}$ and $\mathcal{O}_{m}$ using the estimated cost matrix $\hat{\mathbf{C}}_{\mathcal{N}_m, \mathcal{O}_m}$. 

\begin{figure}[t]
\centering
\hspace{-0.0in}
\begin{minipage}{0.29\textwidth}
\centering
  \subfigure[A potential inference attack carried out by an attacker using $\hat{\mathbf{C}}_{\mathcal{N}_m, \mathcal{O}_m}$]{
\includegraphics[width=1.00\textwidth]{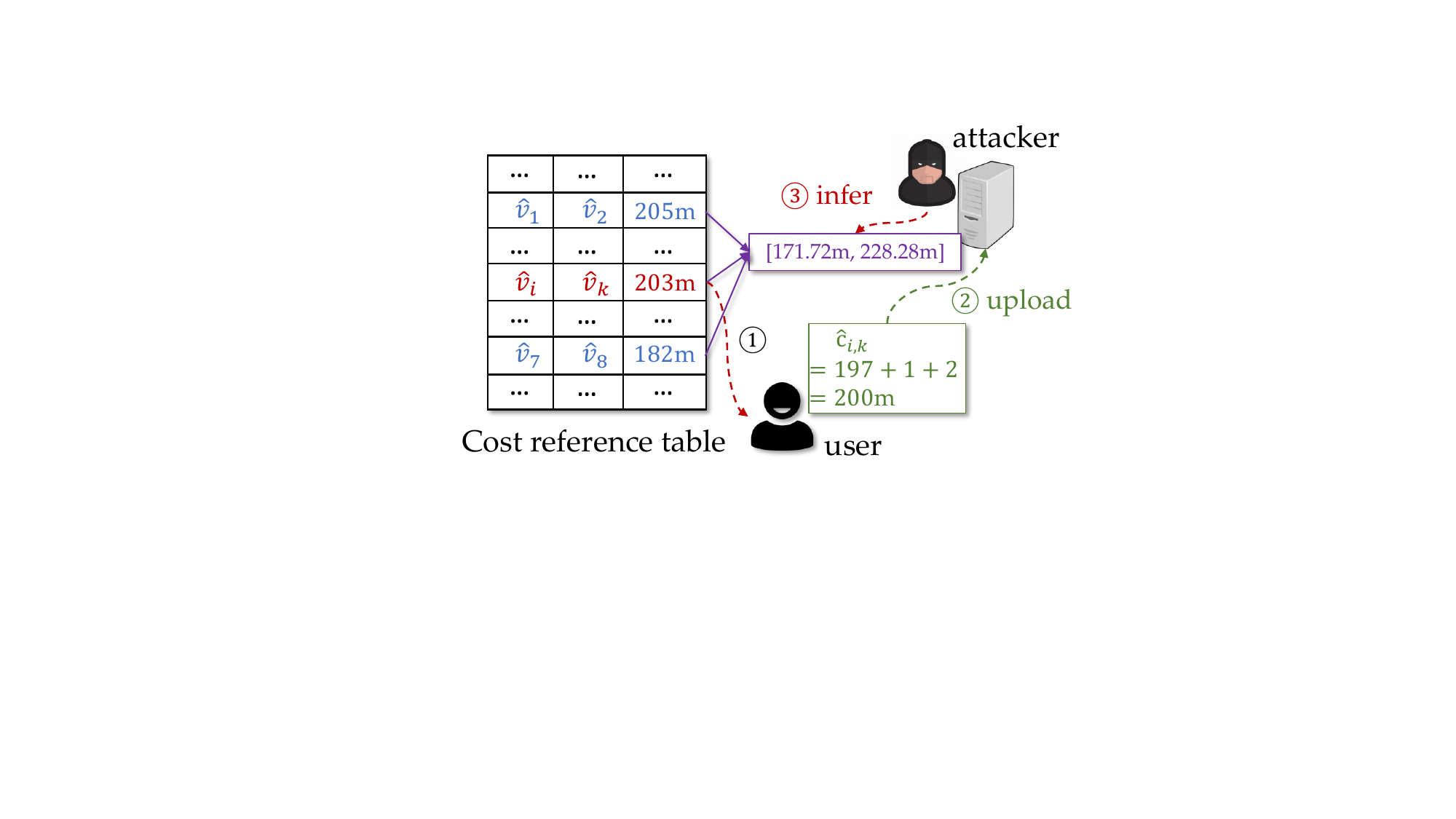}}
\end{minipage}
\hspace{0.05in}
\begin{minipage}{0.147\textwidth}
\centering
  \subfigure[Example in experiment.]{
\includegraphics[width=1.00\textwidth]{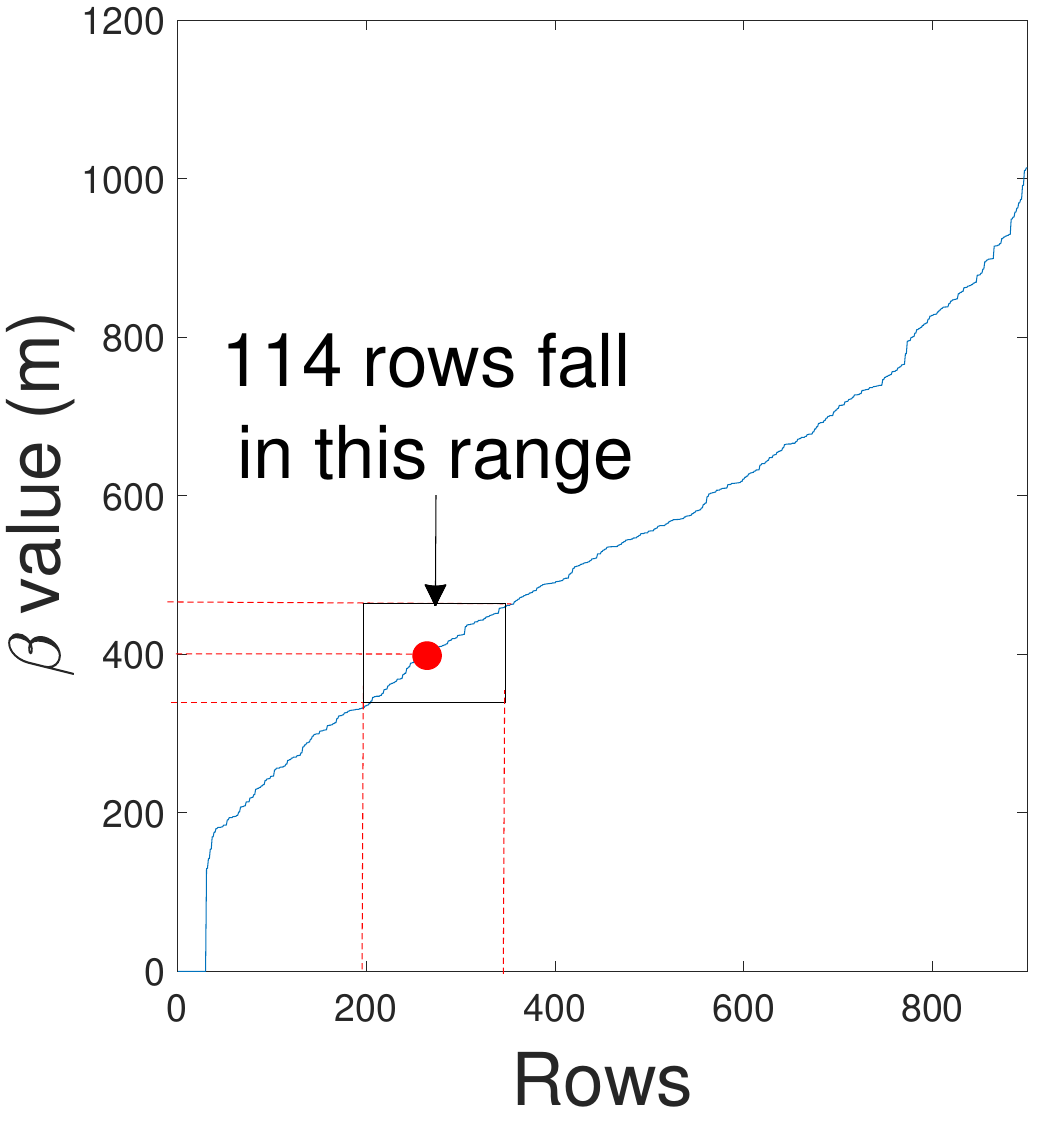}}
\end{minipage}
\vspace{-0.15in}
\caption{Potential inference attack using estimated cost matrix $\hat{\mathbf{C}}_{\mathcal{N}_m, \mathcal{O}_m}$ and the cost reference table. }
\label{fig:costreferenceTableattack}
\vspace{-0.05in}
\end{figure}

Fig. \ref{fig:costreferenceTableattack}(a) gives an example, where the user calculates the cost coefficient $\hat{c}_{v_i,v_k} = \beta_{i,k} + d_{v_i, \hat{v}_i} + d_{v_k, \hat{v}_k} = 197 + 1 + 2 = 200$m (\textcircled{1}), and uploads $\hat{c}_{v_i,v_k}$ to the server (\textcircled{2}). After receiving $\hat{c}_{v_i,v_k}$, a potential attack by the attacker is to find the corresponding $\beta$ value in the cost reference table, of which the estimated cost coefficient by a user can be possibly 200m according to Equ. (\ref{eq:costestimate}). 

Note that both $d_{v_i, \hat{v}_i}$ and $d_{v_k, \hat{v}_k}$ in Equ. (\ref{eq:costestimate}) are unknown by the server, while the server can drive the maximum possible value of $d_{v_i, \hat{v}_i}$ and $d_{v_k, \hat{v}_k}$, denoted by $\delta_{\max}$ (e.g. $\delta_{\max} = 28.28$ in Fig. \ref{fig:costreferenceTableattack}(a)), based on the distribution of $\hat{\mathcal{V}}$. In this case, the server can derive that the matched $\beta$ value is in the interval $[\hat{c}_{v_i,v_k}-\delta_{\max}, \hat{c}_{v_i,v_k}+\delta_{\max}] = [171.72, 228.28]$ (\textcircled{3}), which might cover other $\beta$ values in the cost reference table, like 205m and 182m in Fig. \ref{fig:costreferenceTableattack}(a). In this case, the attacker cannot identify which $\beta$ is the true $\beta_{i,k}$ in the interval, and the more $\beta$ values fall in the interval, the more difficult for the attacker to find the true $\beta_{i,k}$. 

{\revisiondone Fig. \ref{fig:costreferenceTableattack}(b) gives another example of how many $\beta$ values can possibly match an estimated cost coefficient using the real world map information (more details can be found in our experiment in Section \ref{sec:performance}). In this example, the server creates a cost reference table covering 900 locations in $\hat{\mathcal{V}}$ by a grid map with each cell size equal to 100m. The maximum distance from a location in $\mathcal{N}_{m}$ and $\mathcal{O}_{m}$ to its nearest location in $\hat{\mathcal{V}}$ is 70.7m. Given an estimated cost coefficient $\hat{c}_A =$ 400m, its corresponding $\beta_A$ is in the interval [400m$-$70.7m, 400m$+$70.7m], where 114 $\beta$ values fall in this interval. 
On average, each cost coefficient is matched by 102.98 rows of the cost reference table. The more comprehensive experimental results can be found in Fig. \ref{fig:exp:threat} in \textbf{Section \ref{sec:performance}}.}

\vspace{-0.05in}
\section{Performance Evaluation}
\label{sec:performance}
\vspace{-0.00in}
In this section, we conduct a simulation using real-world map information to evaluate the performance of LR-Geo in terms of computation efficiency, privacy, and cost, with the comparison of several benchmarks \cite{Qiu-TMC2020, Andres-CCS2013, ImolaUAI2022}. Specifically, we focus on the application of vehicular spatial crowdsourcing \cite{Qiu-TMC2020}, such as Uber like platform \cite{Uber}, where vehicles need to physically travel to a disignated location to complete the task\footnote{The MATLAB source code of LR-Geo is available at: \url{https://github.com/chenxiunt/LocalRelevant_Geo-Obfuscation}}. 

We first introduce the settings of the experiment in \textbf{Section \ref{subsec:settings}}, and then evaluate the performance of different geo-obfuscation methods in \textbf{Section \ref{subsec:exp:computation} and Section \ref{subsec:exp:cost}}. 

\vspace{-0.15in}
\subsection{Settings}
\label{subsec:settings}
\vspace{-0.03in}
\subsubsection{Dataset} 

{\revisiondone 
We selected the city ``\emph{Rome, Italy}'' as the target region (the bounding area with coordinate $(lat=41.66, lon=12.24)$ as the south-west corner, and coordinate $(lat=42.10, lon=12.81)$ as the north-east corner). Similar to existing works \cite{Yu-NDSS2017, Bordenabe-CCS2014, Wang-WWW2017}, we approximate the location field by partitioning the entire target region into a $40 \times 40$ grid. Each grid cell represents a discrete location within the location set, and the distances between cells are calculated based on the travel distance between the centers of the cells. To calculate the travel distances, we retrieve the road map information of Rome, including both the node set and edge set, using OpenStreetMap \cite{openstreetmap}. We compute the shortest path distances between cell centers on the road map using Dijkstra's algorithm \cite{Algorithm}. Additionally, we assume a uniform distribution of targets. 
\looseness = -1 
}

\vspace{-0.03in}
\subsubsection{Benchmarks} We compare LR-Geo with the following benchmarks, which are all based on Geo-Ind: 
\newline (1) \emph{LP-based geo-obfuscation (labeled as ``LP'')} \cite{Qiu-TMC2020}: LP considers the network-constrained mobility features of the vehicles and employs LP formulated in Equ. (\ref{eq:OMGobj})(\ref{eq:OMGconstr}) to minimize the expected cost. 
\newline (2) \emph{Laplacian noise (labeled as ``Laplace'')} \cite{Andres-CCS2013}: Laplace adds a polar Laplacian noise $\phi$ to the real location, i.e., $v_i + \phi$ and approximate it by the closest location $v_k = \arg\min_{v\in \mathcal{V}}d_{v, v_i+\phi}$.   
\newline (3) \emph{Exponential mechanism (labeled as ``\emph{ExpMech}'')} \cite{Andres-CCS2013}: In ExpMech, the probability distribution of the obfuscated location of each real location $v_i$ follows a polar Laplace distribution $z_{i,k} \propto e^{-\epsilon c_{v_i, v_k}/2}$. 
\vspace{0.00in}
\newline (4) \emph{``\emph{ConstOPTMech}'' or ``\emph{ConstOPT}''} \cite{ImolaUAI2022}: Like our approach, ConstOPT applies the exponential distribution constraint for a subset of the obfuscation probabilities and uses LP for the optimization of the remaining obfuscation probabilities, to balance the utility and scalability of the data perturbation method.
\vspace{-0.00in}
\newline (5) {\revisiondone \emph{``\emph{LR-Geo-F}''}: In addition to the four benchmarks mentioned above, we also consider LR-Geo with the neighbor threshold $\gamma$ set to an infinite value. In this case, the Geo-Ind Graph is fully connected, which aligns with existing works such as \cite{Bordenabe-CCS2014, Wang-WWW2017, Qiu-TMC2020}, which do not require Geo-Ind to be satisfied only between locations that are within a distance smaller than $\gamma$. We use "LR-Geo-F" to label LR-Geo with $\gamma = \infty$, where "F" stands for "fully connected Geo-Ind Graph".} \looseness = -1

\vspace{-0.03in}
\subsubsection{Metrics} We measure the following metrics to evaluate the performance of our method and the benchmarks:
\begin{itemize}
\item [(i)] \emph{Computation time}, which is defined as the amount of time to calculate an obfuscation matrix. The experiments are performed by a desktop with 13th Gen Intel Core i7 processor, 16 cores. We used the Matlab LP toolbox \texttt{linprog}, with the algorithm ``\texttt{dual-simplex}'' \cite{matlab} to solve LP. 
\item [(ii)] \emph{Expected cost} $\mathcal{L}(\mathbf{Z})$:  $\mathcal{L}(\mathbf{Z})$ is defined in Equ. (\ref{eq:overallcost}), meaning the expected estimation error of traveling cost caused by $\mathbf{Z}$. 
\item [(iii)] \emph{Geo-Ind violation (GV) ratio}, which is defined as the ratio: 
\begin{equation}
\label{eq:GVratio}
\frac{\mbox{\# of $(z_{i,k}, z_{j,k})$ violating Geo-Ind in Equ. (\ref{eq:Geo-Ind-general})}}{\mbox{\# of $(z_{i,k}, z_{j,k})$  that should satisfy Geo-Ind in Equ. (\ref{eq:Geo-Ind-general})}}. 
\end{equation}
The GV ratio reflects how the derived obfuscation matrix can achieve Geo-Ind. In the following experiment, by default, we set $\epsilon$ by 10.0km$^{-1}$, the cell size of the cost reference table by 0.1km, the LR distance threshold $\Gamma$ by 20km. 
\end{itemize}

\vspace{-0.15in}
\subsection{Computation Efficiency}
\label{subsec:exp:computation}


\subsubsection{Comparison with the benchmarks} Table \ref{Tb:exp:computation} compares the computational times for LR-Geo against four benchmark methods, where the number of locations $K$ equals 100, 200, 300, and 400, respectively. The table reveals that \textbf{while LR-Geo has higher computational time compared to Laplace and ExpMech, it significantly outperforms both LP and ConstOPT in terms of efficiency}.

{\revisiondone Specifically, at $K = 200$, LR-Geo demonstrates a remarkable reduction in computation time, showing a decrease of 99.51\% and 98.12\% compared to LP and ConstOPT, respectively. For both LP and ConstOPT, computation times exceed the 1800-second threshold when $K \geq 300$. This enhanced efficiency of LR-Geo is due to its strategic approach of confining the set of locations under consideration to LR locations only. Conversely, the alternative LP-based methods evaluate every location within the targeted area, resulting in substantial computational overhead. 

Additionally, the computation time of LR-Geo is 87.71\% lower than that of LR-Geo-F, as LR-Geo-F imposes Geo-Ind constraints for all pairs of locations, not just those with a distance smaller than $\gamma$. This leads to a higher number of linear constraints in the LP formulation and, consequently, greater computational overhead. Nevertheless, the computation time of LR-Geo-F is still 99.26\% and 97.18\% lower than that of LP and ConstOPT at $K = 200$, respectively.}

In addition, both Laplace and ExpMech can attain lower computation times compared to LR-Geo. This efficiency stems from their methodology of selecting obfuscated locations based on predefined probability distributions - the Laplacian and exponential distributions, respectively - bypassing the need for LP, which in turn reduces the computation overhead. However, a notable drawback of these two methods is their inability to accurately estimate the cost caused by geo-obfuscation. This oversight results in an increased cost associated with geo-obfuscation, as the chosen obfuscated locations may lead to high traveling distances to the designated locations.


\begin{table}[t]
\vspace{-0.00in}
\caption{Computation time (seconds) of different methods. Mean$\pm$1.96$\times$ standard deviation. }
\vspace{-0.1in}
\label{Tb:exp:computation}
\centering
\small 
\begin{tabular}{ c|c|c|c|c}
\hline
\hline
\rowcolor{Gray}
\multicolumn{1}{ c  }{}
&\multicolumn{4}{ c }{Problem size}
\\
\cline{2-5}
\multicolumn{1}{ c|  }{Methods}
&\multicolumn{1}{ |c| }{$K = 100$}
 & \multicolumn{1}{ |c| }{$K = 200$}&\multicolumn{1}{ |c }{$K = 300$}&\multicolumn{1}{ |c }{$K = 400$}
 \\ 
\hline
\hline
\multicolumn{1}{ c|  }{LR-Geo} & 1.28$\pm$1.35 & 1.42$\pm$ 0.63 & 1.34$\pm$0.64
 & 1.20$\pm$0.85 \\ 
\multicolumn{1}{ c|  }{LR-Geo-F} & 1.63$\pm$1.00 & 2.14$\pm$ 0.71 & 2.42$\pm$0.85
 & 3.56$\pm$1.07 \\ 
\multicolumn{1}{ c|  }{LP} & 27.31$\pm$10.31 & 287.50$\pm$48.28 & $\geq 1800$  & $\geq 1800$ \\ 
\multicolumn{1}{ c|  }{ConstOPT} & 3.69$\pm$1.92 & 75.88$\pm$10.81 & $\geq 1800$ & $\geq 1800$\\ 
\multicolumn{1}{ c|  }{Laplace} & $\leq$0.005 & $\leq$0.005 & $\leq$0.005 & $\leq$0.005\\ 
\multicolumn{1}{ c|  }{ExpMech} & $\leq$0.005 & $\leq$0.005 & $\leq$0.005 & $\leq$0.005\\ 
\hline
\end{tabular}
\vspace{-0.00in}
\end{table}

\vspace{-0.10in}
\subsubsection{Scalability} Table \ref{Tb:exp:computation} illustrates that the computation time for all algorithms escalates as the size of the location set $K$ increases. Notably, \textbf{even when $K$ reaches 400, the average computation time for LR-Leo remains at a comparatively low figure, approximately 0.8--1.8 seconds}.

We expanded our examination of $K$ across a wider spectrum, from 100 to 1,600, and charted the computation times of LR-Geo in Fig. \ref{fig:exp:scalability}(a). This figure reveals that the computation time for LR-Geo is maintained at the same level with an increase in $K$, reaching up to 1.8 seconds. Moreover, Fig. \ref{fig:exp:scalability}(b) presents the computation times for LR-Geo as the number of users varies from 2 to 10. As expected, there is a noticeable rise in computation time corresponding to an increase in the number of users. This trend is attributed to the framework of Benders' decomposition (introduced in Section \ref{subsec:Benders}), where the server is tasked with generating a subproblem for each user. The increase in the number of subproblems heightens the probability of encountering at least one subproblem that fails to achieve optimal convergence swiftly, thereby prolonging the convergence time.

\begin{figure}[t]
\centering
\begin{minipage}{0.23\textwidth}
\centering
  \subfigure[Different location set sizes]{
\includegraphics[width=1.00\textwidth, height = 0.12\textheight]{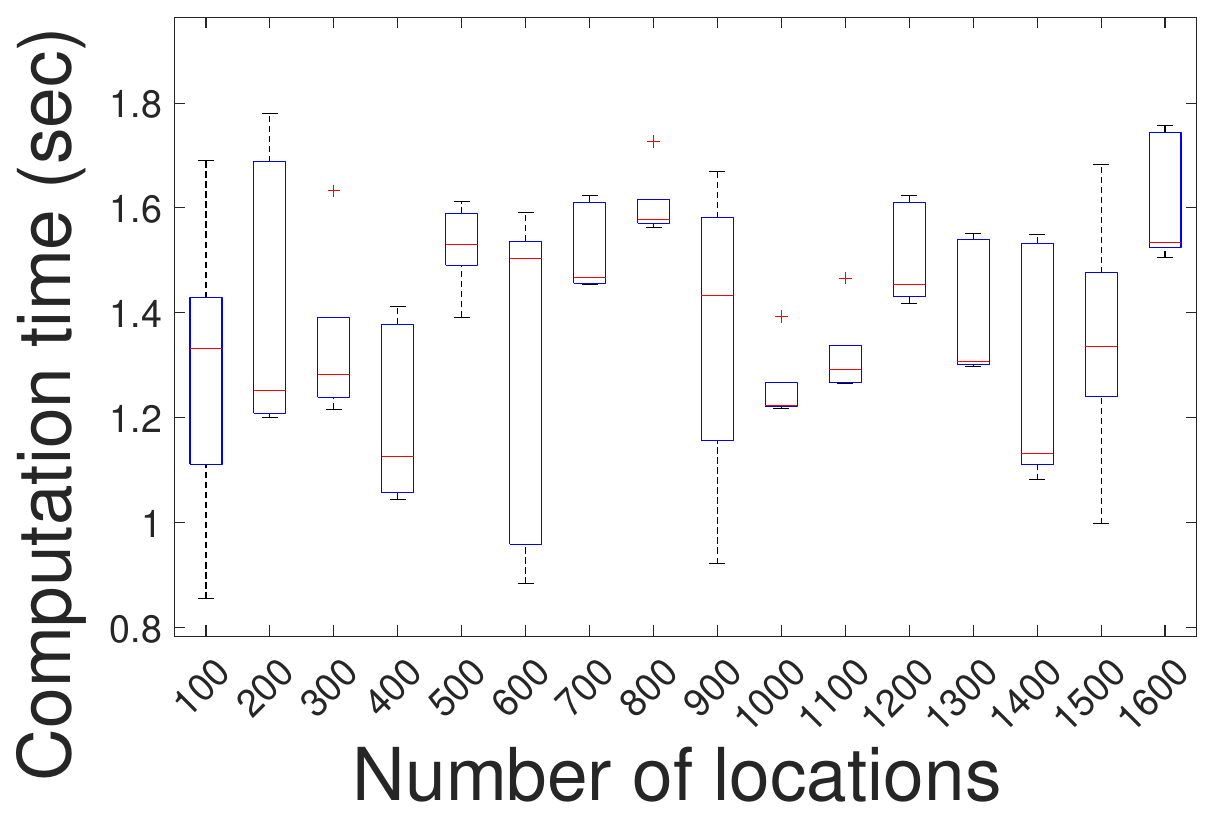}}
\label{}
\end{minipage}
\hspace{0.05in}
\begin{minipage}{0.23\textwidth}
\centering
  \subfigure[Different numbers of users]{
\includegraphics[width=1.00\textwidth, height = 0.12\textheight]{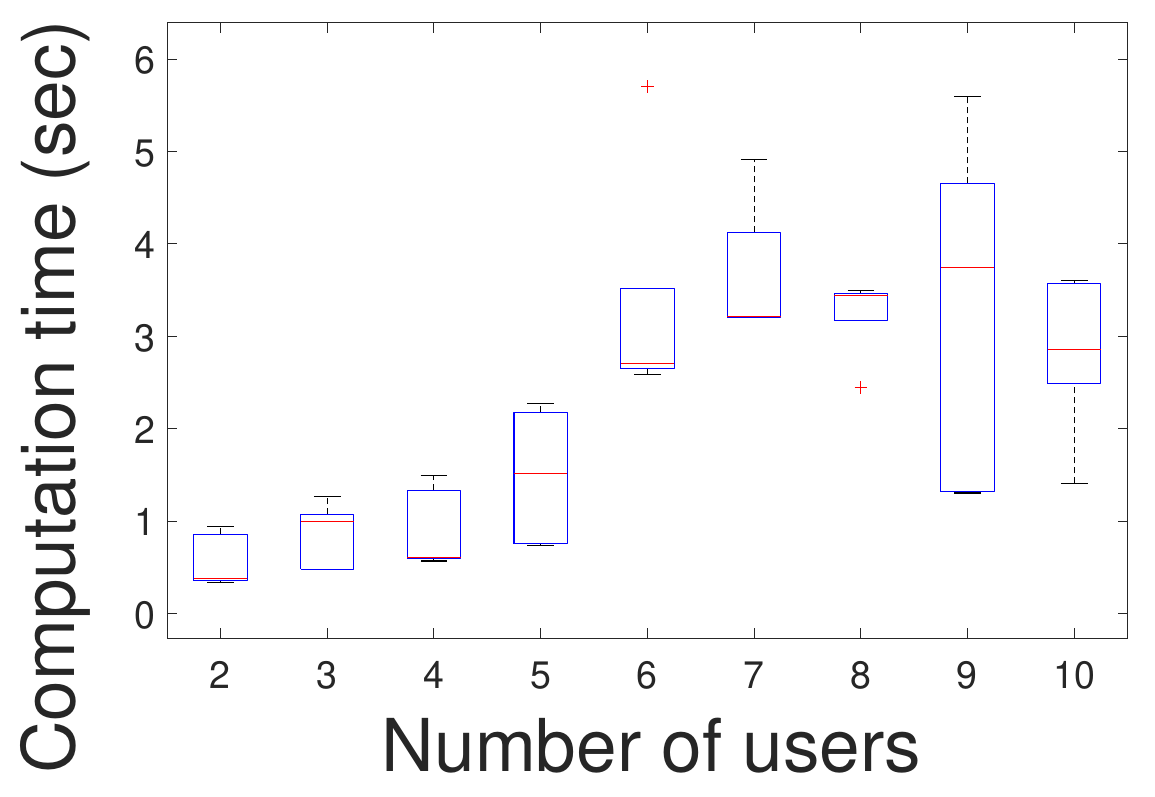}}
\label{}
\end{minipage}
\vspace{-0.15in}
\caption{Scalability.}
\label{fig:exp:scalability}
\vspace{-0.15in}
\end{figure}

\DEL{
\begin{figure}[t]
\centering
\begin{minipage}{0.23\textwidth}
\centering
  \subfigure[Optimality loss]{
\includegraphics[width=1.00\textwidth, height = 0.12\textheight]{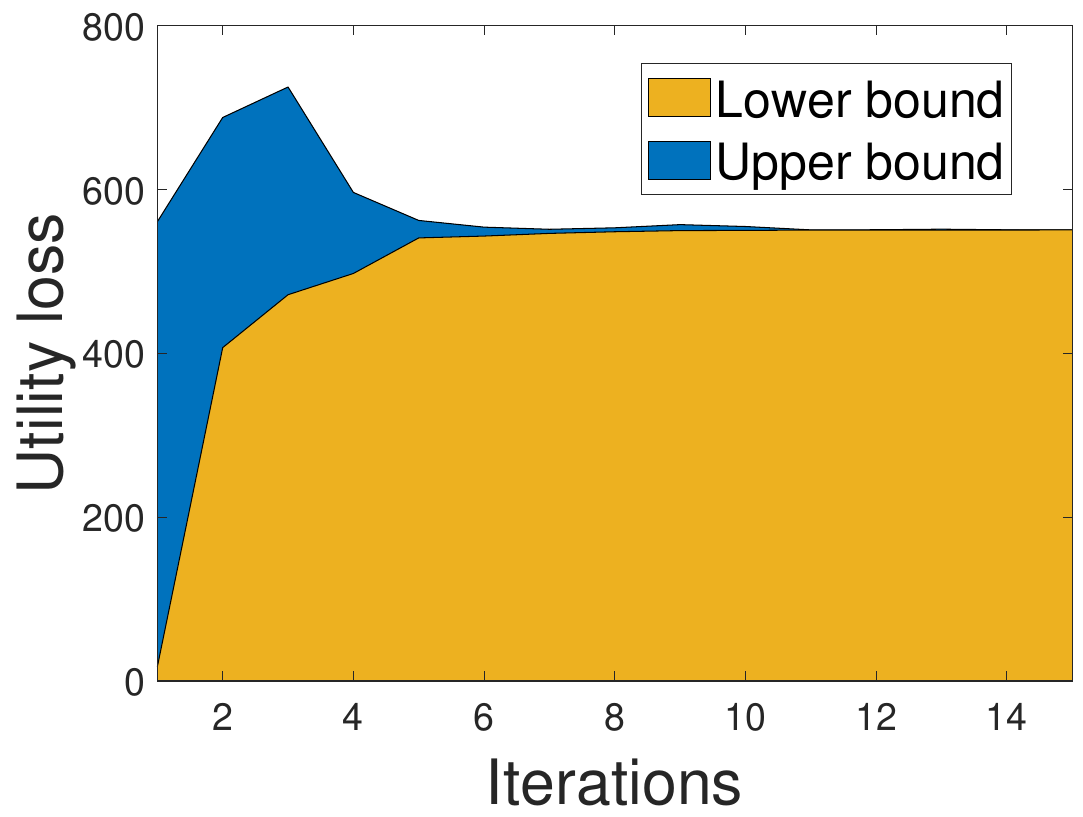}}
\label{}
\end{minipage}
\hspace{0.05in}
\begin{minipage}{0.23\textwidth}
\centering
  \subfigure[Computation time]{
\includegraphics[width=1.00\textwidth, height = 0.12\textheight]{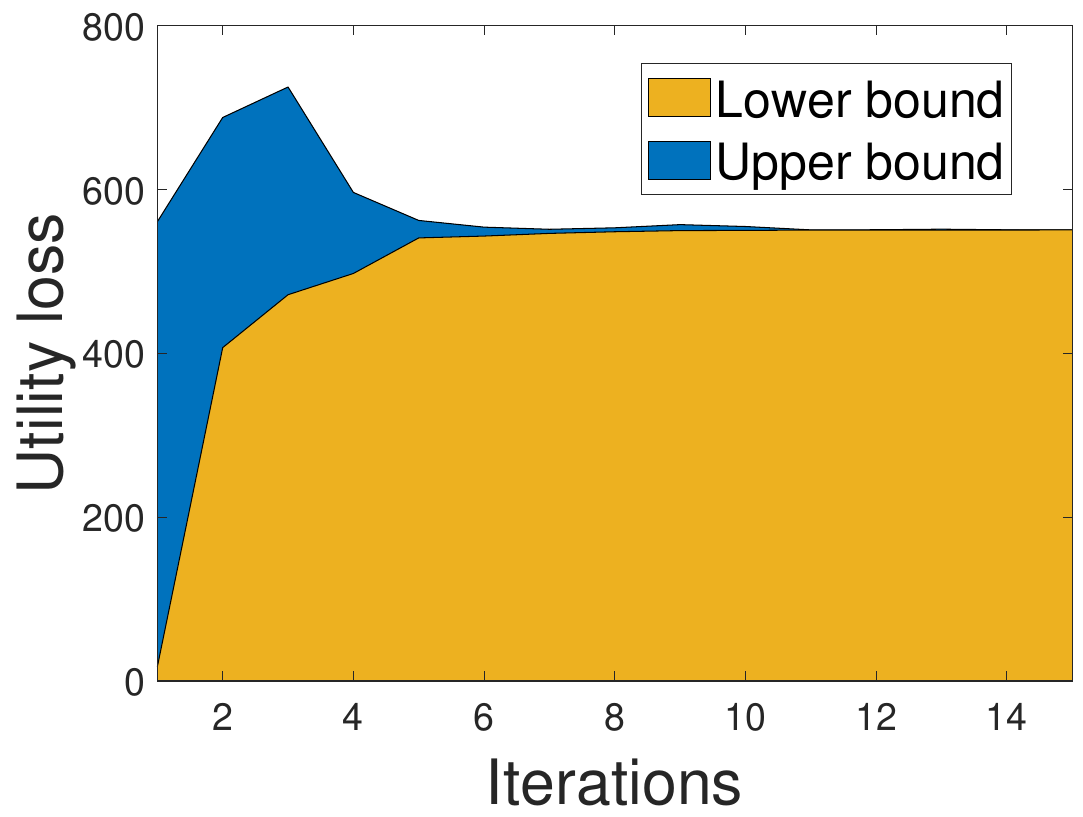}}
\label{}
\end{minipage}
\vspace{-0.10in}
\caption{Computation time given different parameters.}
\label{fig:exp:convergence_parameters}
\vspace{-0.15in}
\end{figure}}

\vspace{-0.03in}

\subsection{Cost Measurement}
\label{subsec:exp:cost}

\vspace{-0.03in}

\subsubsection{Comparison with the benchmarks} 
\begin{table}[t]
\vspace{-0.00in}
\caption{Cost (kilometers) of different methods. Mean$\pm$1.96$\times$ standard deviation. }
\vspace{-0.1in}
\label{Tb:exp:cost}
\centering
\small 
\begin{tabular}{ c|c|c|c|c}
\hline
\hline
\rowcolor{Gray}
\multicolumn{1}{ c  }{}
&\multicolumn{4}{ c }{Problem size}
\\
\cline{2-5}
\multicolumn{1}{ c|  }{Methods}
&\multicolumn{1}{ |c| }{$K = 100$}
 & \multicolumn{1}{ |c| }{$K = 200$}&\multicolumn{1}{ |c }{$K = 300$}&\multicolumn{1}{ |c }{$K = 400$}
 \\ 
\hline
\hline
\multicolumn{1}{ c|  }{LR-Geo} & 0.36$\pm$0.04 & 0.36$\pm$0.07 & 0.35$\pm$0.03 & 0.37$\pm$0.05 \\ 
\multicolumn{1}{ c|  }{LR-Geo-F} & 0.38$\pm$0.03 & 0.36$\pm$0.05 & 0.38$\pm$0.06 & 0.40$\pm$0.05 \\ 
\multicolumn{1}{ c|  }{ConstOPT} & 0.35$\pm$0.03 & 0.34$\pm$0.04 & ------ & ------ \\ 
\multicolumn{1}{ c|  }{LP} & 0.33$\pm$0.02 & 0.33$\pm$0.03 & ------ & ------\\ 
\multicolumn{1}{ c|  }{Laplace} & 0.81$\pm$0.02 & 0.78$\pm$0.02 & 0.80$\pm$0.06 &0.79$\pm$0.01\\  
\multicolumn{1}{ c|  }{ExpMech} & 0.67$\pm$0.04 & 0.64$\pm$0.07 & 0.68$\pm$0.12 & 0.70$\pm$0.05\\ 
\hline
\multicolumn{1}{ c|  }{Lower bound} & 0.29$\pm$0.08 & 0.30$\pm$0.06 & 0.31$\pm$0.08 & 0.30$\pm$0.05 \\ 
\hline
\end{tabular}
\vspace{-0.12in}
\end{table}

Table \ref{Tb:exp:cost} compares the expected costs incurred by various algorithms for $K = 100, 200, 300, 400$. It is observed that LR-Geo significantly reduces the expected cost compared to Laplace and ExpMech. Specifically, LR-Geo's expected cost is, on average, 54.70\% and 46.64\% lower than that of Laplace and ExpMech, respectively. This efficiency is attributed to Laplace and ExpMech's reliance on Laplace/Exponential distributions for selecting obfuscated locations, which fails to accurately reflect the mobility constraints of vehicles within the road network, thereby elevating the cost. Furthermore, LR-Geo's cost performance is nearly on par with ConstOPT's for $K=100,200, 300$, yet it surpasses LP in cost at $K=100,200, 300$. Although LP is designed to achieve the global minimum cost by evaluating all potential locations within the target area, this approach is negated by its extensive computational requirements. As indicated in Table \ref{Tb:exp:computation}, LP struggles to compute obfuscation matrices within the 1800-second limit, highlighting a critical trade-off between cost efficiency and computational feasibility. {\revisiondone Finally, the cost of LR-Geo-F is 5.56\% higher than that of LR-Geo because LR-Geo-F imposes Geo-Ind constraints on all pairs of LR locations, resulting in a smaller feasible region for the obfuscation matrix and, consequently, higher utility loss.  \looseness = -1}

\subsubsection{Comparison with the theoretical bounds} To assess how close LR-Geo can achieve the optimal, we calculate a lower bound for the expected cost by solving the relaxed version of LR-Geo in Equ.  (\ref{eq:OMGLRobjguarantee})--(\ref{eq:OMGLRconstr}), with the findings presented in Table \ref{Tb:exp:computation}. Here, we introduce the \emph{approximation ratio}, defined as the quotient of the expected cost derived from LR-Geo over the calculated lower bound. A smaller approximation ratio indicates a closer proximity of LR-Geo's solution to the optimal. The results in the table indicate that, on average, \textbf{the approximation ratio for the expected cost of LR-Geo stands at 1.24, 1.2, 1.13, and 1.23 for $K = 100, 200, 300, 400$, respectively}. \looseness = -1

\begin{figure}[t]
\centering
\begin{minipage}{0.490\textwidth}
\centering
  \subfigure[Different $\Gamma$]{
\includegraphics[width=0.480\textwidth, height = 0.12\textheight]{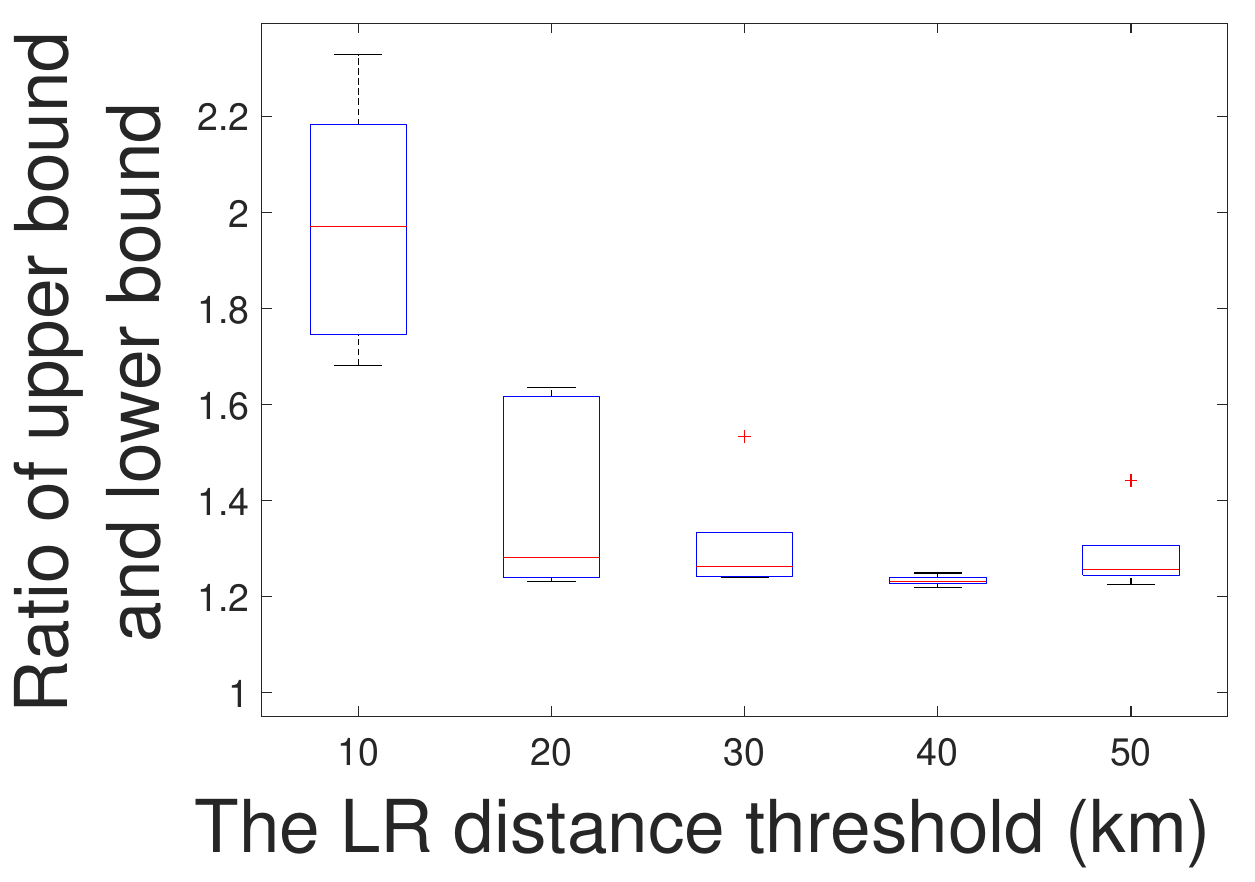}}
\vspace{-0.00in}
\centering
  \subfigure[Difference cell size]{
\includegraphics[width=0.480\textwidth, height = 0.12\textheight]{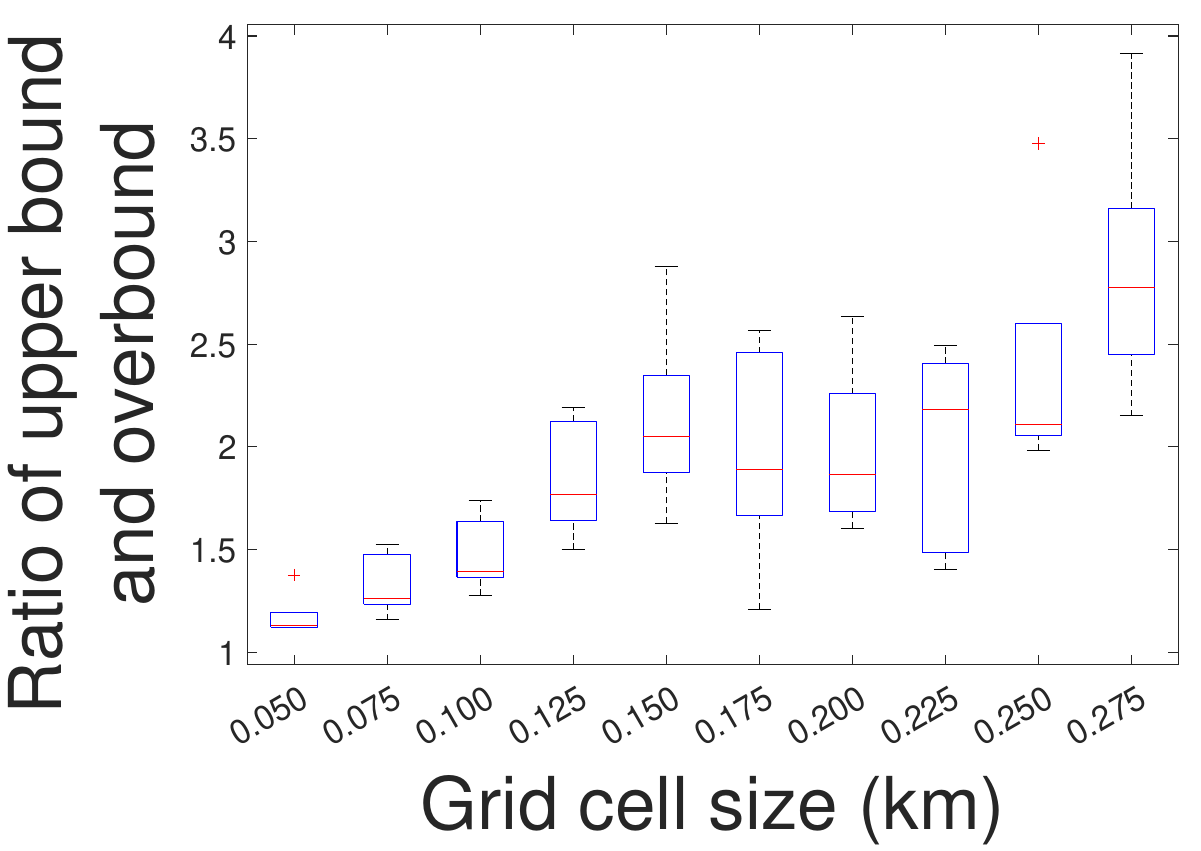}}
\vspace{-0.10in}
\caption{Ratio of upper bound and lower bound.}
\label{fig:exp:bounds}
\end{minipage}
\hspace{0.10in}
\vspace{-0.20in}
\end{figure}
 
It's important to recognize that LR-Geo does not attain the optimal solution since it operates with a constrained set of locations (LR locations) rather than the entire location set. Furthermore, LR-Geo does not utilize exact cost coefficients; instead, it estimates these coefficients using a cost reference table. 
Thus, it is interesting to test how the LR-Geo's approximation ratio is impacted by 
\begin{itemize}
\item[(i)] the selection of the LR locations,  determined by the parameter $\Gamma$, i.e., the LR distance threshold, and 
\item[(ii)] the accuracy of the cost coefficient estimation, determined by the \textbf{cell size} of the grid map of the cost reference table. 
\end{itemize}

Fig. \ref{fig:exp:bounds}(a) shows the variation in the approximation ratio of LR-Geo as $\Gamma$ increases from 10km to 50km. As defined in Equ. (\ref{eq:Gamma}), $\Gamma$ influences the size of the LR location set $\mathcal{N}_{m}$, with a higher $\Gamma$ resulting in a larger $\mathcal{N}_{m}$. The figure indicates that the approximation ratio experiences a more pronounced decrease (averaging 34.72\%) as $\Gamma$ is increased from 10km to 20km. However, the decrease becomes marginal (only 0.79\%) when $\Gamma$ is further expanded from 20km to 50km. This observation suggests that enhancing $\Gamma$ contributes to the optimality of the obfuscation matrices, yet beyond a certain threshold (20km in this instance), additional increases in $\Gamma$ yield negligible improvements.

Fig. \ref{fig:exp:bounds}(b) shows the approximation ratio of LR-Geo as the cell size increases from 0.05km to 0.275km. As expected, the approximation ratio escalates with the increase in cell size, indicating that finer granularity in the location's representation within the cost reference table allows LR-Geo to more closely approximate the optimal solution. Specifically, the approximation ratio remains relatively stable and low for cell sizes up to 0.20km. Beyond this point, particularly when the cell size surpasses 0.20km, the ratio sees a marked increase. This trend underscores the importance of maintaining a cell size at or below 0.20km to optimize cost efficiency.

\DEL{

It's important to recognize that LR-Geo does not attain the optimal solution since it operates with a constrained set of locations (LR locations) rather than the entire location set. Furthermore, LR-Geo does not utilize exact cost coefficients; instead, it estimates these coefficients using a cost reference table. 
Thus, it is interesting to test how the LR-Geo's approximation ratio is impacted by 
\begin{itemize}
\item[(i)] the selection of the LR locations,  determined by the parameter $\Gamma$, i.e., the LR distance threshold, and 
\item[(ii)] the accuracy of the cost coefficient estimation, determined by the \textbf{cell size} of the grid map of the cost reference table. 
\end{itemize}

Fig. \ref{fig:exp:bounds}(a) shows the variation in the approximation ratio of LR-Geo as $\Gamma$ increases from 10km to 50km. As defined in Equ. (\ref{eq:Gamma}), $\Gamma$ influences the size of the LR location set $\mathcal{N}_{m}$, with a higher $\Gamma$ resulting in a larger $\mathcal{N}_{m}$. The figure indicates that the approximation ratio experiences a more pronounced decrease (averaging 4.14\%) as $\Gamma$ is increased from 10km to 20km. However, the decrease becomes marginal (only 2.05\%) when $\Gamma$ is further expanded from 20km to 50km. This observation suggests that enhancing $\Gamma$ contributes to the optimality of the obfuscation matrices, yet beyond a certain threshold (20km in this instance), additional increases in $\Gamma$ yield negligible improvements.

Fig. \ref{fig:exp:bounds}(b) shows the approximation ratio of LR-Geo as the cell size increases from 0.05km to 0.25km. As expected, the approximation ratio escalates with the increase in cell size, indicating that finer granularity in the location's representation within the cost reference table allows LR-Geo to more closely approximate the optimal solution. Specifically, the approximation ratio remains relatively stable and low for cell sizes up to 0.15km. Beyond this point, particularly when the cell size surpasses 0.175km, the ratio sees a marked increase. This trend underscores the importance of maintaining a cell size at or below 0.15km to optimize cost efficiency.}

\begin{figure}[t]
\centering
\DEL{
\begin{minipage}{0.490\textwidth}
\centering
  \subfigure[Different $\Gamma$]{
\includegraphics[width=0.480\textwidth, height = 0.12\textheight]{./fig/exp/approxratio_Gamma}}
\vspace{-0.00in}
\centering
  \subfigure[Difference cell size]{
\includegraphics[width=0.480\textwidth, height = 0.12\textheight]{./fig/exp/approxratio_cellsize}}
\vspace{-0.2in}
\caption{Ratio of upper bound and lower bound.}
\label{fig:exp:bounds}
\end{minipage}
\hspace{0.00in}}
\begin{minipage}{0.230\textwidth}
\centering
  \subfigure{
\includegraphics[width=1.00\textwidth, height = 0.12\textheight]{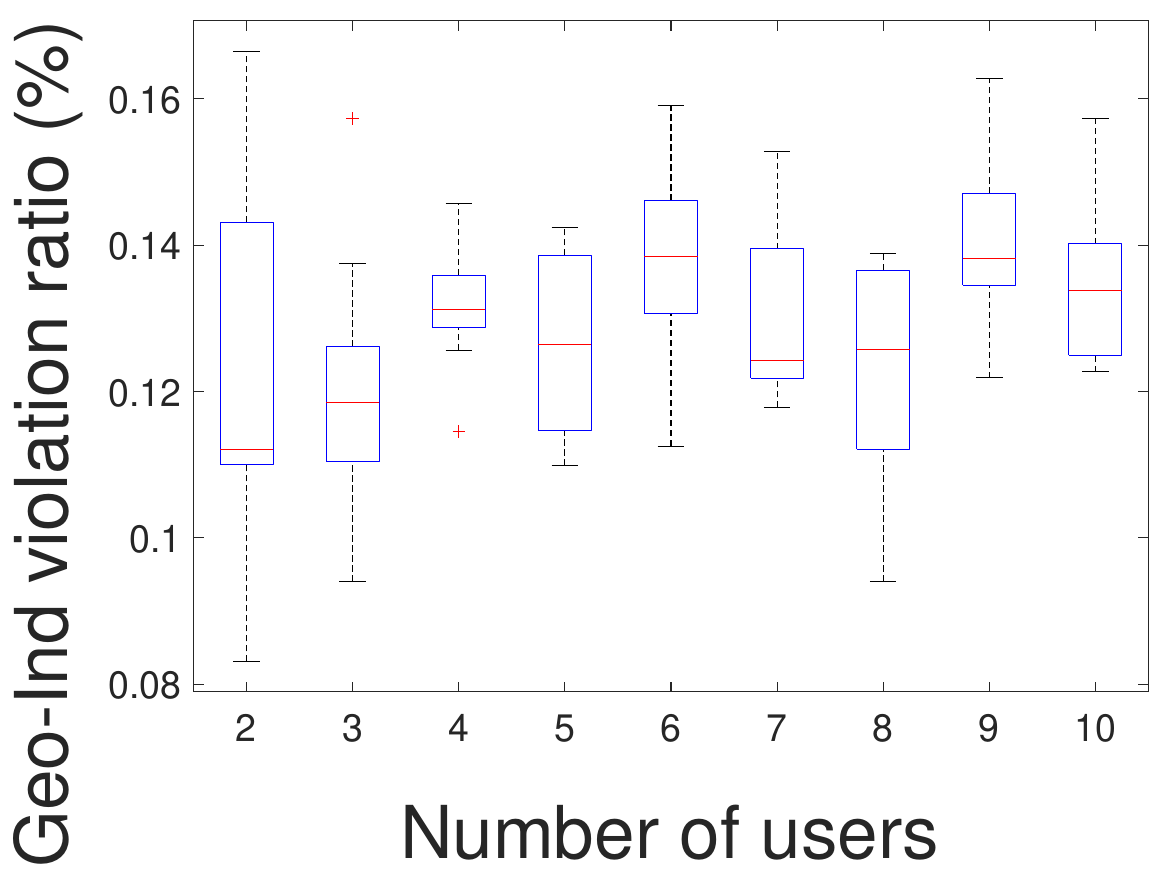}}
\vspace{-0.15in}
\caption{Geo-Ind violation ratio.}
\label{fig:exp:GVR}
\end{minipage}
\hspace{0.05in}
\begin{minipage}{0.230\textwidth}
\centering
  \subfigure{
\includegraphics[width=1.00\textwidth, height = 0.12\textheight]{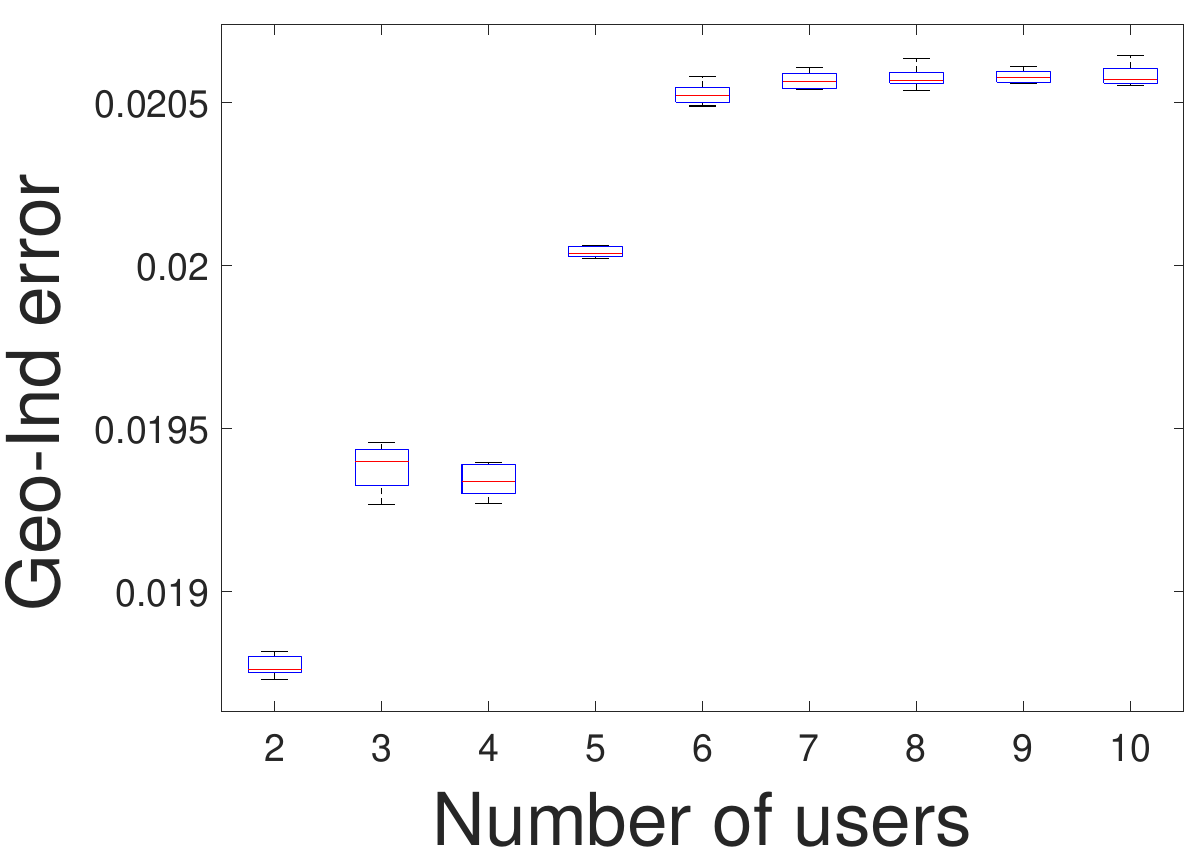}}
\vspace{-0.15in}
\caption{Geo-Ind violation error.}
\label{fig:exp:violation_error}
\label{}
\end{minipage}
\begin{minipage}{0.460\textwidth}
\centering
  \subfigure{
\includegraphics[width=0.50\textwidth, height = 0.12\textheight]{./fig/exp/violation_error}}
\vspace{-0.15in}
\caption{Threat using cost coefficient matrices.}
\label{fig:exp:threat}
\label{}
\end{minipage}
\vspace{-0.15in}
\end{figure}

\DEL{
\begin{figure}[t]
\centering
\begin{minipage}{0.23\textwidth}
\centering
  \subfigure{
\includegraphics[width=1.00\textwidth, height = 0.13\textheight]{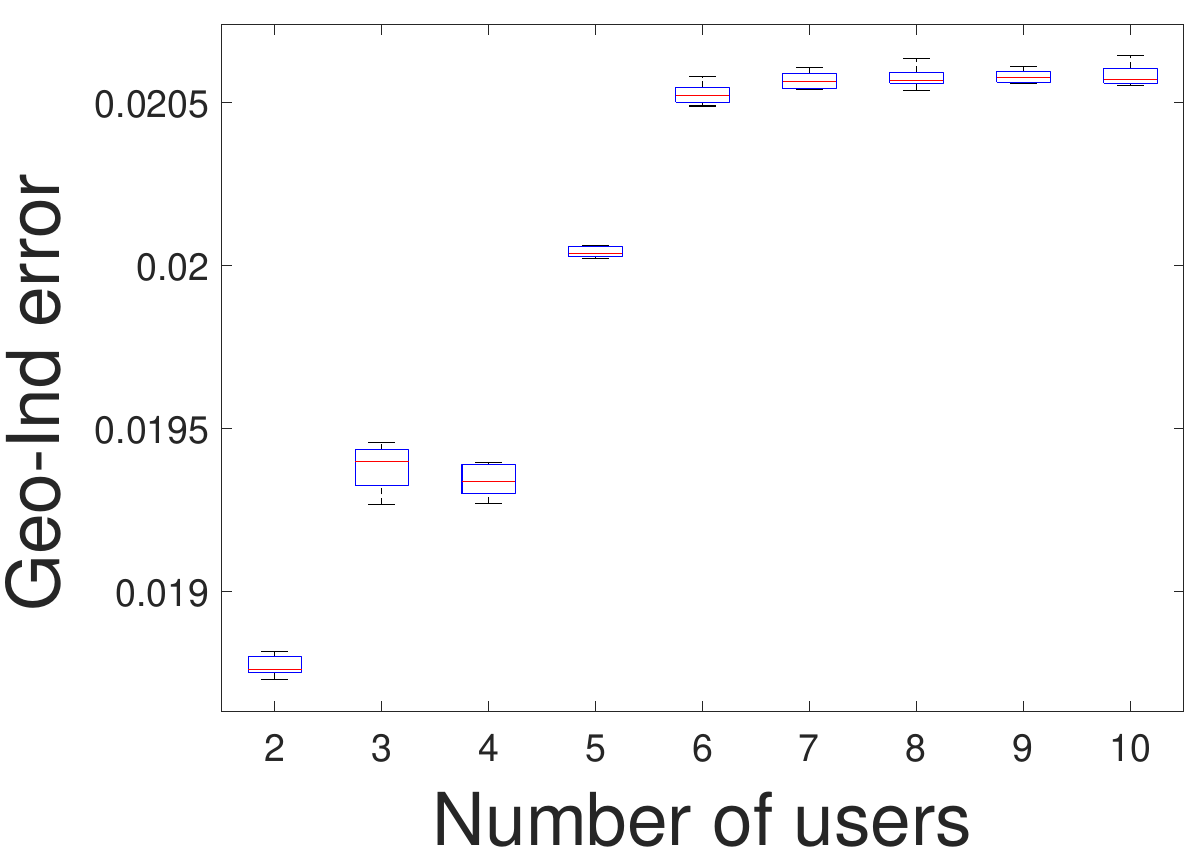}}
\end{minipage}
\vspace{-0.10in}
\caption{Geo-Ind violation error.}
\label{fig:exp:violation_error}
\vspace{-0.10in}
\end{figure}}

\vspace{-0.10in}
\subsection{Privacy Measure}
\label{subsec:exp:privacy}
\vspace{-0.00in}
In LR-Geo, the computation of obfuscation matrices for each user is performed independently. While the obfuscation probabilities that conform to the constraints of the exponential distribution (in Equ. (\ref{eq:expo})) meet the Geo-Ind privacy criterion, as substantiated by \textbf{Theorem \ref{thm:privacyguarantee}}, the remaining obfuscation probabilities do not guarantee Geo-Ind privacy. In this part, we examine the GV ratio as defined in Equ. (\ref{eq:GVratio}). Fig. \ref{fig:exp:GVR} shows the GV ratios for varying numbers of users. The figure reveals that the GV ratio remains exceptionally low, with a maximum of only 0.13\%, demonstrating that, in practice, the Geo-Ind constraints are exceedingly likely to be met across the various obfuscation matrices tailored for different users. \looseness = -1


{\revisiondone Another question is the extent to which obfuscation probabilities violate the Geo-Ind constraint during these violations. If a pair of obfuscation probabilities, $z_{i,k}$ and $z_{j,k}$, fails to satisfy the Geo-Ind constraint—i.e., $z_{i,k} > e^{\epsilon d_{v_i, v_j}}z_{j,k}$—we define their 'Geo-Ind error (GVE)' as:
\vspace{-0.15in}
\begin{equation}
\mathrm{GVE}(z_{i,k},z_{j,k}) = z_{i,k} - e^{\epsilon d_{v_i, v_j}}z_{j,k}. 
\vspace{-0.05in}
\end{equation}
A smaller $\mathrm{GVE}(z_{i,k},z_{j,k})$ indicates that the Geo-Ind constraint between $z_{i,k}$ and $z_{j,k}$ is violated to a lesser extent.

Using this definition, we measured the Geo-Ind error for all Geo-Ind violations as the number of users increased from 2 to 10. Figure \ref{fig:exp:violation_error} illustrates the distribution of these errors. The results show that the Geo-Ind error reaches up to 0.021 across all experiments, and this error tends to increase as the number of users grows. This is because Geo-Ind violations occur between users, and as the user count rises, the number of potential Geo-Ind violations increases, resulting in a greater likelihood of larger Geo-Ind errors. }

Finally, we investigate the potential risk associated with the upload of cost matrices, a concern discussed in \textbf{Section \ref{subsec:discussionthreats}}. We simulate a scenario where a user uploads 100 cost matrices. We analyze, for each cost coefficient, the number of rows in the cost reference table that can be mapped to that coefficient. Intuitively, a greater number of rows mapped to a specific uploaded coefficient suggests a broader range of potential real and obfuscated location pairs, thereby diminishing the risk of LR location set disclosure (noting that the real location is within the LR location set). Fig. \ref{fig:exp:threat} displays the number of rows mapped to the uploaded cost coefficients for various grid cell sizes. As anticipated, the quantity of rows corresponding to a given coefficient increases with the increase of the cell size, indicating an increase in ambiguity and a reduced risk of location inference. The figure also underscores the difficulty of deducing the real location from the uploaded cost coefficient, as, on average, each coefficient is matched by 102.98 rows, providing a significant degree of location privacy.

\begin{table*}[t]
\vspace{-0.00in}
\caption{Comparison of major related works ("PN" means "privacy notion". "$\gamma$" means "whether $\gamma$ is considered", "EIE" means "expected inference error", In the column "Location set size", "---" means “There is no limit to the size”).}
\vspace{-0.1in}
\label{Tb:reference}
\centering
\footnotesize 
\begin{tabular}{ c|c|c|c|c|c|c|c}
\hline
\hline
\rowcolor{Gray}
\multicolumn{1}{ c  }{}
&\multicolumn{7}{ c }{Features}
\\
\cline{2-8}
\multicolumn{1}{ c|  }{Obfuscation methods}
&\multicolumn{1}{ |c| }{Lap.}
 & \multicolumn{1}{ |c }{Exp.}&\multicolumn{1}{ |c }{LP}&\multicolumn{1}{ |c }{$\gamma$}&\multicolumn{1}{ |c }{$K$} &\multicolumn{1}{ |c }{Privacy notion}  &\multicolumn{1}{ |c }{Utility loss definition} 
 \\ 
\hline
\hline
\multicolumn{1}{ c|  }{CCS 2012 \cite{Shokri-CCS2012}} &  & 
 & \checkmark &  & 30 & EIE  & Expected \& max distance between real and reported locations  \\ 
\multicolumn{1}{ c|  }{CCS 2013 \cite{Andres-CCS2013}} & \checkmark & 
 &  &\checkmark& --- & Geo-Ind& Expected distance between real and reported locations \\ 
\multicolumn{1}{ c|  }{CCS 2014 \cite{Bordenabe-CCS2014}} &  & 
& \checkmark &  & 50 & Geo-Ind & Expected distance between real and reported locations\\ 
\multicolumn{1}{ c|  }{PETS 2015 (Privacy game) \cite{Shokri-PoPETs2015}} &  & & \checkmark &  & 300 & EIE \& DP & Expected utility cost (depends on applications)\\ 
\multicolumn{1}{ c|  }{PETS 2015 (Elastic metric) \cite{Chatzikokolakis-PoPETs2015}} &  & \checkmark & &  & --- & Geo-Ind (generalized) & Expected distance between real and reported locations \\ 
\multicolumn{1}{ c|  }{ICDM 2016 \cite{Wang-CIDM2016}} &  &  & \checkmark &  & 57 & DP & Expected residual standard error\\ 
\multicolumn{1}{ c|  }{WWW 2017 \cite{Wang-WWW2017}} &  & 
 & \checkmark &  & 16 & Geo-Ind & Expected travel distance\\ 
\multicolumn{1}{ c|  }{NDSS 2017 \cite{Yu-NDSS2017}} &  & 
 & \checkmark &  & 50 & EIE \& Geo-Ind  & Expected distance between real and reported locations \\ 
 \multicolumn{1}{ c|  }{{\revisiondone CCS 2017 \cite{Oya-CCS2017}}} &  & 
 & \checkmark &  & {\revisiondone 25} & EIE \& Conditional entropy & Average \& worst-case quality loss (depends on applications) \\ 
 \multicolumn{1}{ c|  }{{\revisiondone PETS 2017 \cite{chatzikokolakis2017efficient}}} & \checkmark & 
 & \checkmark &  & --- & Geo-Ind & Expected quality loss (depends on applications)\\ 
\multicolumn{1}{ c|  }{TMC 2020 \cite{Qiu-TMC2020}} &  & 
& \checkmark &  & 300--400 & Geo-Ind & Expected difference between real and estimated travel distance\\ 
\multicolumn{1}{ c|  }{PETS 2020 \cite{Mendes-PETS2020}} &  \checkmark & 
& &  & --- & Geo-Ind & Expected distance between real and reported locations \\ 
\multicolumn{1}{ c|  }{SIGSPATIAL 2022 \cite{Qiu-SIGSPATIAL2022}} &  & & &  & 100 & Geo-Ind & Expected difference between real and estimated travel distance \\ 
\multicolumn{1}{ c|  }{UAI 2022 \cite{ImolaUAI2022}} &  & \checkmark & \checkmark
 &  & 400 & Metric DP & Expected utility cost (depends on applications) \\ 
\multicolumn{1}{ c|  }{EDBT 2023 \cite{Pappachan-EDBT2023}} &  &
 & \checkmark &  & 70 & Geo-Ind & Expected distance between real and reported locations\\ 
 \multicolumn{1}{ c|  }{EDBT 2024 \cite{Qiu-EDBT2024}} &  &
 & \checkmark &  & 300--400 & Geo-Ind & Expected distance between real and reported locations \\ 
\multicolumn{1}{ c|  }{Our work} &  & {\bf \checkmark} & {\bf \checkmark} &\checkmark & {\bf 1,600} &  Geo-Ind & Expected distance between real and reported locations \\ 
\hline
\end{tabular}
\vspace{-0.05in}
\end{table*}

\vspace{-0.05in}
\section{Related Works}
\label{sec:related}
\vspace{-0.00in}
The study of location privacy began nearly two decades ago with Gruteser and Grunwald's pioneering work \cite{Gruteser-MobiSys2003}, where they introduced the concept of \emph{location $k$-anonymity}. This idea has since evolved to include \emph{$l$-diversity}, which ensures a user's location is indistinguishable from $l-1$ other locations \cite{Yu-NDSS2017}. However, the $l$-diversity model simplifies the threat landscape by assuming all alternative locations are equally probable as the user's actual location from an attacker's perspective. This assumption renders it susceptible to a range of sophisticated inference attacks \cite{Andres-CCS2013, Yu-NDSS2017, Qiu-TMC2020}.


In recent years, Andr{\'e}s et al. \cite{Andres-CCS2013} introduced a more applicable privacy criterion, \emph{Geo-Ind}, grounded in the established concept of \emph{differential privacy (DP)}. Following this work, a large body of location obfuscation strategies have been proposed, e.g., \cite{Mendes-PETS2020, Simon-EuroSP2019, Wang-WWW2017, Shokri-TOPS2017, Yu-NDSS2017, Xiao-CCS2015, Bordenabe-CCS2014, Qiu-CIKM2020, Qiu-TMC2020, Qiu-SIGSPATIAL2022, Al-Dhubhani-PETS2017, Wang-CIDM2016}. Andr{\'e}s et al., in their seminal work, not only proposed the Geo-Ind concept but also devised a method for achieving it by perturbing the actual location using a polar Laplacian distribution. Furthermore, as geo-obfuscation naturally introduces errors in the reported locations, thereby impacting the quality of LBS, a critical challenge addressed by several studies involves balancing the trade-off between service quality and privacy. For instance, within the constraints of Geo-Ind, Bordenabe et al. \cite{Bordenabe-CCS2014} first developed an optimization framework for geo-obfuscation aimed at minimizing individual user costs. Chatzikokolakis et al. \cite{Chatzikokolakis-PoPETs2015} introduced "privacy mass" for points of interest, determining the Geo-Ind privacy budget $\epsilon$ for a location based on the local characteristics of each area. Wang et al. \cite{Wang-WWW2017} addressed the collective cost incurred by users, proposing a privacy-preserving target assignment algorithm to reduce the total travel expense. \looseness = -1
{\revisiondone Chatzikokolakis et al. \cite{chatzikokolakis2017efficient} introduced a Bayesian remapping procedure to enhance the utility of geo-obfuscation, which can be applied to both infinite and finite location domains.}

The majority of existing works in geo-obfuscation employ an LP framework, which generally necessitates $O(|\mathcal{V}|^2)$ decision variables and $O(|\mathcal{V}||\mathcal{E}|)$ linear constraints \cite{Bordenabe-CCS2014, Qiu2024ijcai}, making the LP approach computationally intensive and challenging to implement on a large-scale LBS. Table \ref{Tb:reference} compares the related geo-obfuscation methods in different categories, including Laplacian noise (``Lap.''), the exponential mechanism (``Exp.''), and LP-based methods (``LP''). As the table indicates, the computational complexity of LP restricts most geo-obfuscation studies to handling at most 100 discrete locations. However, recent advancements \cite{Qiu-CIKM2020, Qiu-TMC2020} have expanded the capability of processing secret datasets to approximately 300 records by leveraging Dantzig-Wolfe decomposition and column generation techniques. These studies primarily target LP models with Geo-Ind constraints applied across all pairs of secret records, facilitating the initialization process for column generation but are less applicable to broader geo-obfuscation challenges that only necessitate constraints for adjacent locations. Other innovative approaches, such as \cite{ImolaUAI2022}, combine LP with the exponential mechanism to improve scalability, though this may lead to compromises in solution optimality. 
Given the time-sensitive natures of many LBS applications, existing geo-obfuscation methodologies are constrained to either low spatial resolution over large areas (for instance, \cite{Yu-NDSS2017} focuses on city-scale regions, discretizing the location field into a grid where each cell measures 766m by 766m), or to high resolution within smaller areas (as in \cite{Qiu-TMC2020}, which examines a small town with location points sampled every 500 square meters).

Compared to those existing works, LR-Geo introduced in this paper substantially lowers computational costs while maintaining a degree of optimality. This advancement facilitates the application of geo-obfuscation in large-scale LBS applications, enabling more accurate representations of locations.

\vspace{-0.05in}
\section{Conclusions}
\label{sec:conclude}
\vspace{-0.00in}

We proposed to reduce the computation cost of the geo-obfuscation calculation by shrinking its range to a set of more relevant locations. Considering that the reduced geo-obfuscation range can possibly disclose the user's real location, we designed a remote computing strategy to migrate the geo-obfuscation calculation to the server without disclosing the location set covered by geo-obfuscation. The experimental results have demonstrated the superiority of our method in terms of privacy, service quality, and time efficiency, with the comparison of the selected benchmarks.  \looseness = -1

We envision several promising directions to continue this research. Firstly, this paper focuses on locations that are evenly distributed. In the next phase, we will consider cases where the locations are not necessarily evenly distributed. Since varying location densities lead to different distance matrices, which may inadvertently reveal information about a user's actual location, we will conduct a formal analysis of the potential information leakage caused by these distance matrices.  
Secondly, our current work considers a homogeneous mobility model, where a single cost reference table graph is sufficient to describe users' traveling costs. In reality, the users might be heterogeneous, e.g., a mixture of pedestrians and vehicles, and even a single user's mobility can possibly switch between different models. Then, how to model the mobility features of heterogeneous users using multiple cost reference table graphs is another problem to address. 
Finally, 
leveraging \emph{reinforcement learning (RL)} could accelerate BD convergence by treating cut selection in Stage 2 as a parameterized stochastic policy. A trained RL model can identify an optimal sequence of cuts, 
eliminating the need for re-training with each new problem instance.

\vspace{-0.10in}
\begin{acks}
This research is partially supported by U.S. NSF grants CNS2136948 and CNS-2313866. The authors used ChatGPT 4o to revise the text in Sections \ref{sec:introduction}--\ref{sec:related} to correct typos, grammatical errors, and awkward phrasing. Work of Anna Squicciarini was
supported by (while serving at) the National Science Foundation. Any opinion, findings, and conclusions or recommendations expressed in this material are those of the author(s) and do not necessarily reflect the views of the National Science Foundation. 
\end{acks}


\newpage 
\appendix

\section*{Appendix}

\section{Math Notations}
\label{sec:notations}

\vspace{-0.00in}
\begin{table}[h]
\caption{Main notations and their descriptions}
\vspace{-0.00in}
\label{Tb:Notationmodel}
\centering
\normalsize 
\small 
\begin{tabular}{l l}
\hline
\hline
\rowcolor{Gray}
Symbol                  & Description \\
\hline
$\mathcal{V}$           & The discrete location set 
 $\mathcal{V} = \left\{v_1, ..., v_k\right\}$ \\ 
$\mathcal{G}$           & The Geo-Ind graph $\mathcal{G} = \left(\mathcal{V}, \mathcal{E}\right)$, where $\mathcal{V}$ and $\mathcal{E}$ are the \\ 
                        &  location set and the edge set of $\mathcal{G}$\\
$d_{v_i, v_j}$               & The Haversine distance between $v_i$ and $v_j$ \\
$\mathbf{Z}$            & The obfuscation matrix $\mathbf{Z}$ \\
$z_{i,k}$               & The probability of selecting $v_k$ as the obfuscated location \\
                        &  given the real location $v_i$ \\ 
$\mathbf{z}_i$            & The obfuscation vector of $v_i$, i.e.,  $\mathbf{z}_i = [z_{i,1}, ..., z_{i,K}]$ \\
$\mathcal{C}\left(v, r\right)$           & The circle centered at $v$ with radius $r$ \\ 
$tc\left(v, u\right)$           & The travel cost from location $v$ to location $u$ \\ 
$c_{v_i, v_k}$   & The cost coefficient of $z_{i,k}$ in OMG \\ 
$\mathcal{N}_{m}$           & The LR location set of $v_i$\\ 
$\mathcal{O}_{m}$           & The obfuscated location set of $\mathcal{N}_{m}$\\ 
$\hat{\mathcal{V}}$           & The discrete location set covered by cost reference table \\ 
$\tilde{\mathcal{G}}$           & The cost reference table graph $\tilde{\mathcal{G}} = \left(\hat{\mathcal{V}}, \hat{\mathcal{E}}\right)$, where $\hat{\mathcal{V}}$ and  \\
& $\hat{\mathcal{E}}$ are the location set and the edge set of $\tilde{\mathcal{G}}$\\ 
$\beta_{i,k}$           & The expectation of the cost estimation error taken over   \\
& all possible destination locations \\ 
$\epsilon$           & The privacy budget of Geo-Ind\\ 
$\gamma$           & The neighbor threshold \\
$\Gamma$ & The LR distance threshold \\ 
\hline
\end{tabular}
\normalsize
\vspace{-0.00in}
\end{table}

\subsection{Detailed Notations in Benders' Decomposition}
\label{sec:Bendersnotations}

\begin{itemize}
\item The coefficient matrices $\left[\mathbf{A}_{\mathcal{N}_m}^{\mathrm{GeoI}}, \mathbf{B}_{\mathcal{N}_m}^{\mathrm{GeoI}}\right]$ includes the Geo-Ind constraints between the obfuscation vectors of the locations in $\mathcal{N}_m$: 
\begin{eqnarray}
\label{eq:}
\scriptsize   
\nonumber && \left[\mathbf{A}_{\mathcal{N}_m}^{\mathrm{GeoI}}, \mathbf{B}_{\mathcal{N}_m}^{\mathrm{GeoI}}\right] \\ \nonumber 
&=& \scriptsize  
\left[\begin{array}{lcccr} 
\ddots & \cdots     & \cdots  & \cdots                              & \iddots \\ 
\cdots & 1   & \cdots  & -e^{{\epsilon d_{v_i,v_j}}}    & \cdots \\ 
\cdots & -e^{{\epsilon d_{v_i,v_j}}}    & \cdots  & 1    & \cdots \\
\iddots & \cdots     & \cdots  & \cdots                              & \ddots \\ \end{array}\right] 
\hspace{-0.06in}
\begin{array}{l}
\left\}
\begin{array}{l}
\forall v_i, v_j \in \mathcal{N}_{m} \\
\mbox{s.t.}~d_{v_i, v_j} \leq \gamma
\end{array}
\right.
\\
\end{array}
\end{eqnarray}
\item $\left[\mathbf{A}_{\mathcal{N}_m}^{\mathrm{unit}}, \mathbf{B}_{\mathcal{N}_m}^{\mathrm{unit}}\right]$ includes $|\mathcal{N}_{m}|$ rows, where each row corresponds to the unit measure constraint of the obfuscation vector $\mathbf{z}_i$ of location $v_i \in \mathcal{N}_{m}$. 
\item $\mathbf{b}_{\mathcal{N}_m}^{\mathrm{GeoI}}$ is an all-zeros vector, which corresponds to the right-hand side coefficients of the constraint matrix $\left[\mathbf{A}_{\mathcal{N}_m}^{\mathrm{GeoI}}, \mathbf{B}_{\mathcal{N}_m}^{\mathrm{GeoI}}\right]$ in the LP formulation.
\item $\mathbf{b}_{\mathcal{N}_m}^{\mathrm{unit}}$ is an all-ones vector, which corresponds to the right-hand side coefficients of the constraint matrix $\left[\mathbf{A}_{\mathcal{N}_m}^{\mathrm{unit}}, \mathbf{B}_{\mathcal{N}_m}^{\mathrm{unit}}\right]$ in the LP formulation.  
\end{itemize}

\DEL{
\begin{eqnarray}
\sum_{m=1}^M\mathcal{L}\left(\mathbf{Z}_{\mathcal{N}_m}\right) &=& \sum_{m=1}^M\sum_{v_i\in \mathcal{N}_m} \mathcal{L}\left(\mathbf{z}^a_i\right) + \sum_{m=1}^M\sum_{v_i\in \mathcal{N}_m} \mathcal{L}\left(\mathbf{z}^b_i\right) \\
&=& \sum_{m=1}^M\sum_{v_i\in \mathcal{N}_m} \sum_{v_k \in I(v_i)} c_{v_i, v_k}y_k e^{-\frac{\epsilon d_{v_i, v_k}}{2}} \\ 
&+& \sum_{m=1}^M\sum_{v_i\in \mathcal{N}_m} \mathcal{L}\left(\mathbf{z}^b_i\right)\\
&=& \sum_{k=1}^K \underbrace{\sum_{m=1}^M\sum_{v_i\in \mathcal{N}_m} \mathbf{1}_{v_k \in I(v_i)} c_{v_i, v_k} e^{-\frac{\epsilon d_{v_i, v_k}}{2}}}_{\mbox{Constant value $\alpha_k$}} y_k\\ 
&+& \sum_{m=1}^M\sum_{v_i\in \mathcal{N}_m} \mathcal{L}\left(\mathbf{z}^b_i\right)
\end{eqnarray}}

\vspace{-0.00in}
\section{Omitted Proofs}
\label{sec:proofs}
\subsection{Proof of Theorem \ref{thm:SPD}}
\begin{proof}
We let $\{v_{i}, v_{l_1}, v_{l_2}, ..., v_{l_{n-1}}, v_{l_n}, v_{j}\}$ represent the sequence of locations in the shortest path between $v_i$ and $v_j$. Therefore, $$D_{v_i, v_j} = d_{v_i, v_{l_1}} + \sum_{m=1}^{n-1} d_{v_{l_m}, v_{l_{m+1}}} + d_{v_{l_n}, v_j}.$$
Since each pair of adjacent locations is geo-indistinguishable, for each $v_k \in \mathcal{V}$, we have 
\begin{eqnarray}
&& \frac{z_{i,k}}{z_{l_1,k}} \leq e^{\epsilon d_{v_i, v_{l_1}}}, \\
&& \frac{z_{l_m,k}}{z_{l_{m+1}},k} \leq e^{\epsilon d_{v_{l_m}, v_{l_{m+1}}}}~(m = 1, ..., n-1), \\
&& \frac{z_{l_n,k}}{z_{j,k}} \leq e^{\epsilon d_{v_{l_n}, v_{j}}},
\end{eqnarray}
from which we can derive that 
\begin{eqnarray} 
\frac{z_{i,k}}{z_{j,k}} &=& \frac{z_{i,k}}{z_{l_1,k}} \prod_{m=1}^{n-1} \frac{z_{l_m,k}}{z_{l_{m+1}},k} \frac{z_{l_n,k}}{z_{j,k}} \\
&\leq& e^{\epsilon d_{v_i, v_{l_1}}} \prod_{m=1}^{n-1} e^{\epsilon d_{v_{l_m}, v_{l_{m+1}}}} e^{\epsilon d_{v_{l_n}, v_{j}}} \\  
&=& e^{\epsilon \left(d_{v_i, v_{l_1}} + \sum_{m=1}^{n-1} d_{v_{l_m}, v_{l_{m+1}}} + d_{v_{l_n}, v_j}\right)}  \\
&=& e^{\epsilon D_{v_i, v_j}}. 
\end{eqnarray}
The proof is completed. 
\end{proof}

\subsection{Proof of Proposition \ref{prop:cost reference tablerange}} 
\begin{proof}
First, since the Haversine distance between $v_m$ and $v_j$ should be no larger than their path distance in the Geo-Ind graph, i.e., 
\begin{equation}
\label{eq:d_vivj}
d_{v_m, v_j} \leq D_{v_m, v_j}.
\end{equation}
According to the definition of LR location set in Equ. (\ref{eq:Gamma}), $\forall v_j\in \mathcal{N}_{m}$
\begin{equation}
\label{eq:D_vivj}
D_{v_m, v_j} \leq \Gamma.
\end{equation}
Based on Equ. (\ref{eq:d_vivj}) and Equ. (\ref{eq:D_vivj}), we obtain that 
\begin{equation}
\label{eq:D_vivj1}
d_{v_m, v_j} \leq \Gamma, \forall v_j\in \mathcal{N}_{m}. 
\end{equation}
According to Equ. (\ref{eq:O}), we have 
\begin{equation}
\label{eq:D_vivj2}
d_{v_m, v_j} \leq r_{\mathrm{obf}}, \forall v_j \in \mathcal{O}_{m}.  
\end{equation}
According to Equ. (\ref{eq:D_vivj1}) and Equ. (\ref{eq:D_vivj2}), we have 
\begin{equation}
d_{v_m, v_j} \leq \max\{\Gamma, r_{\mathrm{obf}}\}, ~ \forall v_j \in \mathcal{N}_{m} \cup \mathcal{O}_{m}. 
\end{equation}
$d_{v_m, v_a} \leq \Gamma$ because $v_a$ is selected within the LR location set $\mathcal{N}_m$. Then, according to the triangle inequality, 
\begin{eqnarray}
d_{v_a, v_j} &\leq&  d_{v_m, v_a} + d_{v_m, v_j} \\
&\leq&  \Gamma + \max\{\Gamma, r_{\mathrm{obf}}\} \\
&=& \max\{2\Gamma, \Gamma + r_{\mathrm{obf}}\}, 
\end{eqnarray}
for each $v_j \in \mathcal{N}_{m} \cup \mathcal{O}_{m}$, 
indicating that $\mathcal{C}(v_a, \max\{2\Gamma, \Gamma+r_{\mathrm{obf}}\})$ covers both $\mathcal{N}_{m}$ and $\mathcal{O}_{m}$. 
\end{proof}

\subsection{Proof of Theorem \ref{thm:privacyguarantee}}
\begin{proof}
We prove it by considering the following three cases: 
\newline \textbf{Case 1:}  $v_k$ is within the obfuscation range of both $v_i$ and $v_j$, i.e., $v_k \in \mathcal{O}_i \cap \mathcal{O}_j$. Then, $z^{(n)}_{i,k}$ and $z^{(m)}_{j,k}$ satisfy the constraint Equ. (\ref{eq:LR-OLMGexpoconstr}): 
\begin{eqnarray}
&& z^{(n)}_{i,k} = y_k e^{-\frac{\epsilon d_{v_i, v_k}}{2}}, \\
&& z^{(m)}_{j,k} = y_k e^{-\epsilon \frac{d_{v_j, v_k}}{2}}, ~\forall v_k
\end{eqnarray}
implying that 
\begin{eqnarray}
\nonumber 
&&z^{(n)}_{i,k} - z^{(m)}_{j,k} e^{\epsilon d_{v_i, v_j}} \\ \nonumber 
&=& y_k e^{-\frac{\epsilon d_{v_i, v_k}}{2}} - y_k e^{-\frac{\epsilon d_{v_j, v_k}}{2}} e^{\epsilon d_{v_i, v_j}}\\ \nonumber 
&=& y_k e^{-\frac{\epsilon d_{v_i, v_k}}{2}} - y_k e^{-\frac{\epsilon (d_{v_j, v_k}-d_{v_i, v_j})}{2}} e^{\frac{\epsilon d_{v_i, v_j}}{2}}\\ \nonumber 
&\leq& y_k e^{-\frac{\epsilon d_{v_i, v_k}}{2}} - y_k e^{-\frac{\epsilon d_{v_i, v_k}}{2}} e^{\frac{\epsilon d_{v_i, v_j}}{2}} ~\mbox{(triangle inequality)}\\ \nonumber 
&=& y_k e^{-\frac{\epsilon d_{v_i, v_k}}{2}}(1 - e^{\frac{\epsilon d_{v_i, v_j}}{2}})\\
&\leq& 0
\end{eqnarray}

\noindent \textbf{Case 2:}  $v_k$ is outside of the obfuscation range of either $v_i$ or $v_j$. Without loss of generality, we consider the case $v_k \in \mathcal{O}_i$ and $v_k \notin \mathcal{O}_j$ (meaning $r_{\mathrm{obf}} < d_{v_j, v_k}$), indicating that $z^{(n)}_{i,k} = y_k e^{-\frac{\epsilon d_{v_i, v_k}}{2}}$ and $z^{(m)}_{j,k} = y_k e^{-\frac{\epsilon r_{\mathrm{obf}}}{2}}$. Therefore, 
\begin{eqnarray}
\nonumber
&&z^{(n)}_{i,k} - z^{(m)}_{j,k} e^{\epsilon d_{v_i, v_j}} \\ \nonumber
&=& y_k e^{-\frac{\epsilon d_{v_i, v_k}}{2}} - y_k e^{-\frac{\epsilon r_{\mathrm{obf}}}{2}} e^{\epsilon d_{v_i, v_j}} \\ \nonumber
&=& y_k e^{-\frac{\epsilon d_{v_i, v_k}}{2}} - y_k e^{-\frac{\epsilon (r_{\mathrm{obf}}-d_{v_i, v_j})}{2}} e^{\frac{\epsilon d_{v_i, v_j}}{2}}\\ \nonumber
&<& y_k e^{-\frac{\epsilon d_{v_i, v_k}}{2}} - y_k e^{-\frac{\epsilon (d_{v_j, v_k}-d_{v_i, v_j})}{2}} e^{\frac{\epsilon d_{v_i, v_j}}{2}}  ~\mbox{(since $r_{\mathrm{obf}} < d_{v_j, v_k}$)}\\ \nonumber
&\leq& y_k e^{-\frac{\epsilon d_{v_i, v_k}}{2}} - y_k e^{-\frac{\epsilon d_{v_i, v_k}}{2}} e^{\frac{\epsilon d_{v_i, v_j}}{2}} ~\mbox{(triangle inequality)} \\ \nonumber
&=& y_k e^{-\frac{\epsilon d_{v_i, v_k}}{2}}(1 - e^{\frac{\epsilon d_{v_i, v_j}}{2}})\\
&\leq& 0
\end{eqnarray}

\noindent \textbf{Case 3:}  $v_k$ is outside of the obfuscation range of both $v_i$ and $v_j$, i.e., $v_k \notin \mathcal{O}_i \cup \mathcal{O}_j$. In this case, $z^{(n)}_{i,k} = z^{(m)}_{j,k} = y_k e^{-\frac{\epsilon r_{\mathrm{obf}}}{2}}$, and it is trivial to prove that $z^{(n)}_{i,k}  - e^{{\epsilon d_{v_i, v_j}}}  z^{(m)}_{j,k} \leq 0$, since $e^{{\epsilon d_{v_i, v_j}}} \geq 1$. 

The proof is completed. 
\end{proof}

{\revisiondone 
\subsection{Proof of Proposition \ref{prop:privacyguarantee}}
\label{subsec:prop:privacyguarantee}
{\rd \begin{proof}
(1) First, for each pair of LR locations $v_i$ and $v_j$, their $(\epsilon, \gamma)$-Geo-Ind constraints are 
\begin{eqnarray}
\label{eq:Geo-Ind-count}
z_{i,k}  - e^{{\epsilon d_{v_i, v_j}}}  z_{j,k} \leq 0, ~ \forall v_k \in \mathcal{V}, \\
z_{j,k}  - e^{{\epsilon d_{v_i, v_j}}}  z_{i,k} \leq 0, ~ \forall v_k \in \mathcal{V}, 
\vspace{-0.00in}
\end{eqnarray}
including $2|\mathcal{V}|$ linear constraints.

For any two Users $n$ and $m$, there are totally $(|\mathcal{N}_n|+|\mathcal{N}_m|)$ LR locations, including $\left(\begin{array}{c}|\mathcal{N}_n|+|\mathcal{N}_m| \\ 2 \end{array}\right)$ location pairs. Hence, the total number of Geo-Ind constraints for User $n$ and User $m$ is $$2|\mathcal{V}| \times \left(\begin{array}{c}|\mathcal{N}_n|+|\mathcal{N}_m| \\ 2 \end{array}\right) = (|\mathcal{N}_n|+|\mathcal{N}_m|)(|\mathcal{N}_n|+|\mathcal{N}_m|-1)|\mathcal{V}|.$$

Note that within $\mathcal{N}_n$ (or $\mathcal{N}_m$), the Geo-Ind constraints are satisfied between any peer of relevant locations due to the linear constraints Eq. (). Therefore, the total number of Geo-Ind constraints satisfied within $\mathcal{N}_n$ and $\mathcal{N}_m$ is 
\begin{eqnarray}
&& 2|\mathcal{V}| \times \left(\begin{array}{c}|\mathcal{N}_n| \\ 2 \end{array}\right) + 2|\mathcal{V}| \times \left(\begin{array}{c}|\mathcal{N}_m| \\ 2 \end{array}\right) \\
&=& (|\mathcal{N}_n|^2+|\mathcal{N}_m|^2-|\mathcal{N}_n|-|\mathcal{N}_m|)|\mathcal{V}|
\end{eqnarray}

Now, we consider the Geo-Ind constraints across $\mathcal{N}_n$ and $\mathcal{N}_m$. For each pair of locations $v_i \in \mathcal{N}_n$ and $v_j \in \mathcal{N}_m$, the set of obfuscated locations following the exponential distributions for both $v_i$ and $v_j$ is $\mathcal{A}_i \cap \mathcal{A}_j$. For any $v_k \in \mathcal{A}_i \cap \mathcal{A}_j$, the constraints in Equ. (\ref{eq:Geo-Ind-count}) are guaranteed (according to Theorem \ref{thm:privacyguarantee}), including totally $2|\mathcal{A}_i \cap \mathcal{A}_j|$ linear constraints. The total number of constraints following exponential distributions for all pairs $(v_i, v_j) \in \mathcal{N}_n \times \mathcal{N}_m$ is $\sum_{(v_i, v_j) \in \mathcal{N}_n \times \mathcal{N}_m} 2|\mathcal{A}_i \cap \mathcal{A}_j|.$

Therefore, we can conclude that for each pair User $n$ and User $m$, the Geo-Ind constraint violation ratio is upper bounded by 
\begin{equation}
\small 1- \frac{2\sum_{(v_i, v_j) \in \mathcal{N}_n \times \mathcal{N}_m} |\mathcal{A}_i \cap \mathcal{A}_j|+(|\mathcal{N}_n|^2+|\mathcal{N}_m|^2-|\mathcal{N}_n|-|\mathcal{N}_m|)|\mathcal{V}|}{(|\mathcal{N}_n|+|\mathcal{N}_m|)(|\mathcal{N}_n|+|\mathcal{N}_m|-1)|\mathcal{V}|}.
\end{equation}
The proof is completed. 
\vspace{0.03in}
\newline (2) The total number of Geo-Ind constraints for all the users is $$2|\mathcal{V}| \times\left(\begin{array}{c}\sum_{n=1}^M|\mathcal{N}_n|\\ 2 \end{array}\right) = \sum_{n=1}^M|\mathcal{N}_n|\left(\sum_{n=1}^M|\mathcal{N}_n|-1\right)|\mathcal{V}|.$$

We can obtain the total number of Geo-Ind constraints within each of $\mathcal{N}_1, ..., \mathcal{N}_M$ is 
$$2|\mathcal{V}| \times \sum_{n=1}^M\left(\begin{array}{c}|\mathcal{N}_n| \\ 2 \end{array}\right) = |\mathcal{V}| \sum_{n=1}^M\left(|\mathcal{N}_n|^2 - |\mathcal{N}_n|\right)$$
We can also obtain that the total number of obfuscated locations achieve Geo-Ind constraints across $\mathcal{N}_1, ..., \mathcal{N}_M$ (since they follow exponential distribution) is 
$$2\sum_{n=1}^M \sum_{m=n+1}^M \sum_{(v_i, v_j) \in \mathcal{N}_n \times \mathcal{N}_m} |\mathcal{A}_i \cap \mathcal{A}_j|$$

Finally, we can conclude that Geo-Ind constraint violation ratio for all the users is upper bounded by 
\begin{equation}
\small 1-\frac{2\sum_{n=1}^M \sum_{m=n+1}^M \sum_{(v_i, v_j) \in \mathcal{N}_n \times \mathcal{N}_m} |\mathcal{A}_i \cap \mathcal{A}_j|+ \sum_{n=1}^M\left(|\mathcal{N}_n|^2 - |\mathcal{N}_n|\right)|\mathcal{V}|}{\sum_{n=1}^M|\mathcal{N}_n|\left(\sum_{n=1}^M|\mathcal{N}_n|-1\right)|\mathcal{V}|}.
\end{equation}
The proof is completed. 
\end{proof}}}

\subsection{Proof of Theorem \ref{thm:upperbounds}}
\begin{proof} 
Before proving Theorem \ref{thm:upperbounds}, we first introduce the following lemma: 
\begin{lemma}
\label{lem:upperbounds}
The actual cost $c_{v_i, v_k}$ between location $v_i$ and $v_k$ is upper bounded by the estimated cost $\hat{c}_{v_i,v_k}$. The detailed proof of this lemma can be found in Section \ref{subsec:prooflemmaupperbounds}.
\end{lemma}

Let $\hat{\mathbf{Z}}_{\mathcal{N}_m} = \left\{\hat{z}^{(m)}_{i,k}\right\}_{(v_i,v_k)\in \mathcal{N}_m \times \mathcal{O}_m}$ denote the optimal solution of the CLR-Geo problem in Equ. (\ref{eq:OMGLRobj})--(\ref{eq:LR-OLMGexpoconstr}) using the estimated cost matrix $\hat{\mathbf{C}}_{\mathcal{N}_m, \mathcal{O}_m}$  ($m = 1, ..., M$). Then, for each user $m$, the minimum expected cost calculated by the CLR-Geo problem is given by 
\begin{eqnarray}
&& \mathcal{L}\left(\hat{\mathbf{Z}}_{\mathcal{N}_m}\right) \\ &=& \sum_{v_i\in \mathcal{N}_m} \sum_{v_k\in \mathcal{O}_m} \hat{c}_{v_i,v_k}\hat{z}^{(m)}_{i,k} \\
&\geq& \sum_{v_i\in \mathcal{N}_m} \sum_{v_k\in \mathcal{O}_m} c_{v_i,v_k}\hat{z}^{(m)}_{i,k}~\mbox{(Lemma \ref{lem:upperbounds})} \\
&\geq & \underbrace{\sum_{v_i\in \mathcal{N}_m} \sum_{v_k\in \mathcal{O}_m} c_{v_i,v_k}z^{(m)*}_{i,k}}_{\mbox{the minimum expected cost}}
\end{eqnarray}
where $\mathbf{Z}_{\mathcal{N}_m}^* = \left\{z^{(m)*}_{i,k}\right\}_{(v_i,v_k)\in \mathcal{N}_m \times \mathcal{O}_m}$ denote user $m$'s optimal obfuscation matrix that achieves the minimum cost. The proof is completed. 
\end{proof}

\subsection{Proof of Lemma \ref{lem:upperbounds}}
\label{subsec:prooflemmaupperbounds}
According to $c_{v_i, v_k}$'s  definition (Equ. (\ref{eq:c})),
\normalsize
\small 
\begin{eqnarray}
\nonumber c_{v_i, v_k} &=& p_i \sum_{j=1}^Q q_j \left|d_{v_i, v_j} - d_{v_k,v_j}\right| \\ \nonumber  
&=& p_i \sum_{v_j \in \mathcal{Q}'} q_j \left(d_{v_i, v_j} - d_{v_k,v_j}\right) + p_i \sum_{v_j \in \mathcal{Q}''} q_j \left(d_{v_k,v_j} - d_{v_i, v_j}\right) \\ \nonumber
&\leq& p_i \sum_{v_j\in \mathcal{Q}'} q_j \left(\underbrace{\left(d_{\hat{v}_i,v_j} + d_{v_i, \hat{v}_i}\right)}_{\geq d_{v_i, v_j}\mbox{\footnotesize (triangle inequal.)}} - \underbrace{\left(d_{\hat{v}_k,v_j} - d_{v_k,\hat{v}_k}\right)}_{\leq d_{v_j,v_k} \mbox{\footnotesize (triangle inequal.)}} \right) \\ \nonumber
&+& p_i \sum_{v_j\in \mathcal{Q}''} q_j \left(\underbrace{\left(d_{\hat{v}_k,v_j} + d_{v_k,\hat{v}_k}\right)}_{\geq d_{v_j,v_k} \mbox{\footnotesize (triangle inequal.)}} 
 - \underbrace{\left(d_{\hat{v}_i,v_j} - d_{v_i, \hat{v}_i}\right)}_{\leq d_{v_i, v_j}\mbox{\footnotesize (triangle inequal.)}}\right)
\\
\nonumber &=& p_i \sum_{j=1}^Q q_j \left|d_{\hat{v}_i,v_j} - d_{\hat{v}_k,v_j} \right| + p_i \sum_{j=1}^Q q_j \left(d_{v_i, \hat{v}_i} + d_{v_k,\hat{v}_k} \right) \\
\nonumber &=& p_i \beta_{i,k} - p_i \left(d_{v_i, \hat{v}_i} + d_{v_k,\hat{v}_k} \right) \\
&=& \hat{c}_{v_i,v_k}. 
\end{eqnarray}
\normalsize
\DEL{
\begin{eqnarray}
\nonumber c_{v_i, v_k} &=& p_i \sum_{j=1}^Q q_j \left|d_{v_i, v_j} - d_{v_k,v_j}\right| \geq p_i \sum_{j=1}^Q q_j \left(d_{v_i, v_j} - d_{v_k,v_j}\right) \\ \nonumber
&\geq& p_i \sum_{j=1}^Q q_j \left(\underbrace{\left(d_{\hat{v}_i,v_j} - d_{v_i, \hat{v}_i}\right)}_{\leq d_{v_i, v_j}\mbox{\footnotesize (triangle inequal.)}} - \underbrace{\left(d_{\hat{v}_k,v_j} + d_{v_k,\hat{v}_k}\right)}_{\geq d_{v_j,v_k} \mbox{\footnotesize (triangle inequal.)}} \right) \\
\nonumber &=& p_i \sum_{j=1}^Q q_j \left(d_{\hat{v}_i,v_j} - d_{\hat{v}_k,v_j} \right) - p_i \sum_{j=1}^Q q_j \left(d_{v_i, \hat{v}_i} + d_{v_k,\hat{v}_k} \right) \\
\nonumber &=& p_i \beta_{i,k} - p_i \left(d_{v_i, \hat{v}_i} + d_{v_k,\hat{v}_k} \right)  = \hat{c}_{v_i,v_k}. 
\end{eqnarray}}
\normalsize
\vspace{-0.00in}

\vspace{0.00in}
\subsection{Proof of Theorem \ref{thm:lowerbounds}}
\begin{proof}
Before proving Theorem \ref{thm:lowerbounds}, we first introduce the following lemma:
\begin{lemma}
\label{lem:lowerbounds}
The actual cost $c_{v_i, v_k}$ between location $v_i$ and $v_k$ is lower bounded by the estimated cost $\tilde{c}_{v_i,v_k}$. The detailed proof of this lemma can be found in Section \ref{subsec:prooflemmaoverbounds}.
\end{lemma}
Let $\tilde{\mathbf{Z}}_{\mathcal{N}_m} = \left\{\tilde{z}^{(m)}_{i,k}\right\}_{(v_i,v_k)\in \mathcal{N}_m \times \mathcal{O}_m}$ denote the optimal solution of the relaxed LR-Geo problem in Equ. (\ref{eq:OMGLRobjguarantee})--(\ref{eq:OMGLRconstr}) using the estimated cost matrix $\tilde{\mathbf{C}}_{\mathcal{N}_m, \mathcal{O}_m}$  ($m = 1, ..., M$). Then, for each user $m$, the minimum expected cost calculated by the relaxed LR-Geo problem is given by 
\begin{eqnarray}
&&\mathcal{L}\left(\tilde{\mathbf{Z}}_{\mathcal{N}_m}\right) \\
&=& \sum_{v_i\in \mathcal{N}_m} \sum_{v_k\in \mathcal{O}_m} \tilde{c}_{v_i,v_k}\tilde{z}^{(m)}_{i,k} \\ \nonumber 
&\leq& \sum_{v_i\in \mathcal{N}_m} \sum_{v_k\in \mathcal{O}_m} \tilde{c}_{v_i,v_k}z^{(m)*}_{i,k} ~\mbox{(as $\tilde{\mathbf{Z}}_{\mathcal{N}_m}$ is a relaxed solution of $\mathbf{Z}_{\mathcal{N}_m}^{*}$)} \\
&\leq & \underbrace{\sum_{v_i\in \mathcal{N}_m} \sum_{v_k\in \mathcal{O}_m} c_{v_i,v_k}z^{(m)*}_{i,k}}_{\mbox{the minimum expected cost}} ~\mbox{(Lemma \ref{lem:upperbounds})}
\end{eqnarray}
where $\mathbf{Z}_{\mathcal{N}_m}^* = \left\{z^{(m)*}_{i,k}\right\}_{(v_i,v_k)\in \mathcal{N}_m \times \mathcal{O}_m}$ denote user $m$'s optimal obfuscation matrix that achieves the minimum cost. The proof is completed. 

\vspace{0.00in}
\subsection{Proof of Lemma \ref{lem:lowerbounds}}
\label{subsec:prooflemmaoverbounds}
According to $c_{v_i, v_k}$'s  definition (Equ. (\ref{eq:c})),
\normalsize
\small 
\begin{eqnarray}
\nonumber c_{v_i, v_k} &=& p_i \sum_{j=1}^Q q_j \left|d_{v_i, v_j} - d_{v_k,v_j}\right| \geq p_i \sum_{j=1}^Q q_j \left(d_{v_i, v_j} - d_{v_k,v_j}\right) \\ \nonumber
&\geq& p_i \sum_{j=1}^Q q_j \left(\underbrace{\left(d_{\hat{v}_i,v_j} - d_{v_i, \hat{v}_i}\right)}_{\leq d_{v_i, v_j}\mbox{\footnotesize (triangle inequal.)}} - \underbrace{\left(d_{\hat{v}_k,v_j} + d_{v_k,\hat{v}_k}\right)}_{\geq d_{v_j,v_k} \mbox{\footnotesize (triangle inequal.)}} \right) \\
\nonumber &=& p_i \sum_{j=1}^Q q_j \left(d_{\hat{v}_i,v_j} - d_{\hat{v}_k,v_j} \right) - p_i \sum_{j=1}^Q q_j \left(d_{v_i, \hat{v}_i} + d_{v_k,\hat{v}_k} \right) \\
\nonumber &=& p_i \beta_{i,k} - p_i \left(d_{v_i, \hat{v}_i} + d_{v_k,\hat{v}_k} \right) \\
&=& \tilde{c}_{v_i,v_k}. 
\end{eqnarray}
\normalsize
\end{proof}
\DEL{
\begin{eqnarray}
\nonumber c_{v_i, v_k} &=& p_i \sum_{j=1}^Q q_j \left|d_{v_i, v_j} - d_{v_k,v_j}\right| \geq p_i \sum_{j=1}^Q q_j \left(d_{v_i, v_j} - d_{v_k,v_j}\right) \\ \nonumber
&\geq& p_i \sum_{j=1}^Q q_j \left(\underbrace{\left(d_{\hat{v}_i,v_j} - d_{v_i, \hat{v}_i}\right)}_{\leq d_{v_i, v_j}\mbox{\footnotesize (triangle inequal.)}} - \underbrace{\left(d_{\hat{v}_k,v_j} + d_{v_k,\hat{v}_k}\right)}_{\geq d_{v_j,v_k} \mbox{\footnotesize (triangle inequal.)}} \right) \\
\nonumber &=& p_i \sum_{j=1}^Q q_j \left(d_{\hat{v}_i,v_j} - d_{\hat{v}_k,v_j} \right) - p_i \sum_{j=1}^Q q_j \left(d_{v_i, \hat{v}_i} + d_{v_k,\hat{v}_k} \right) \\
\nonumber &=& p_i \beta_{i,k} - p_i \left(d_{v_i, \hat{v}_i} + d_{v_k,\hat{v}_k} \right)  = \hat{c}_{v_i,v_k}. 
\end{eqnarray}}
\normalsize
\vspace{-0.00in}

\section{Detailed Description of Benders Decomposition}
\label{sec:BDdetaileddescription}
Benders' decomposition is composed of two stages, 
\begin{itemize}
\item \textbf{Stage 1:} A \textbf{Master Program (MP)} to derive $\{y_1, ..., y_K\}$, 
\item \textbf{Stage 2:} A set of \textbf{subproblems} $\mathrm{Sub}_m$ ($m = 1, ..., M$), where each $\mathrm{Sub}_m$ aims to derive $\mathbf{z}'_{\mathcal{N}_m}$. 
\end{itemize}

\vspace{0.03in}
\noindent \textbf{Stage 1: Master program}. The MP derives $y_1, ..., y_K$ and replaces each cost $\mathbf{c}'_{\mathcal{N}_m}\mathbf{z}'_{\mathcal{N}_m}$ by a single decision variable $w_m$, i.e., $w_m = \mathbf{c}'_{\mathcal{N}_m}\mathbf{z}'_{\mathcal{N}_m}$. The MP is formulated as the following LP problem
\normalsize
\begin{eqnarray}
\label{eq:MPObj}
\min && \textstyle 
\sum_{k=1}^K \alpha_k y_k + \sum_{m=1}^M w_m\\
\mathrm{s.t.}
&& \mathcal{H}: \mbox{Cut set of $y_1$, ..., $y_K$, $w_1, ..., w_M$}\\ \label{eq:MPzy0}
&& y_k \geq 0, ~k = 1,..., K. 
\end{eqnarray}
\normalsize
where each \emph{cut} in $\mathcal{H}$ is a \emph{linear inequality} of the decision variables $y_1$, ..., $y_K$, $w_1, ..., w_M$. 
According to the central LR-Geo formulated in Equ. (\ref{eq:OMGLRobj})--(\ref{eq:LR-OLMGexpoconstr}), each $w_m$ is given by 
\begin{equation}
\label{eq:w_m} 
w_m = \min\left\{\mathcal{L}'\left(\mathbf{Z}_{\mathcal{N}_m}\right)\left|\mbox{Equ. (\ref{eq:LR-OMGconstrrewrite}) for $\mathcal{N}_m$ is satisfied} \right.\right\}. 
\end{equation}
Since the MP doesn't know the optimal values of $\mathbf{Z}_{\mathcal{N}_m}$, instead of using Equ. (\ref{eq:w_m}), it ``guesses'' the value of $w_m$ based the \emph{cut set} $\mathcal{H}$. In the subsequent \textbf{Stage 2}, each $\mathrm{Sub}_m$ verifies whether the ``guessed'' value of $w_m$ is feasible and achieves the minimum data cost as defined in Equ. (\ref{eq:w_m}); if not, $\mathrm{Sub}_m$ proposes the addition of a new cut to be included in $\mathcal{H}$, thereby guiding the MP to refine $w_m$ during the next iteration.

In the following, we use $\left\{\overline{y}_{1}, ..., \overline{y}_{K}, \overline{w}_1, ..., \overline{w}_M\right\}$ to represent the optimal solution of the MP. 

\vspace{0.03in}
\noindent \textbf{Stage 2: Subproblems}. After the MP derives its optimal solution $\left\{\overline{y}_{1}, ..., \overline{y}_{K}, \overline{w}_1, ..., \overline{w}_M\right\}$ in \textbf{Stage 1}, each $\mathrm{Sub}_m$ validates whether $\overline{w}_m$ has achieved the minimum data cost, 
\begin{equation}
\label{eq:w_m1} 
w_m = \min \left\{\mathbf{c}'_{\mathcal{N}_m}\mathbf{z}'_{\mathcal{N}_m}\left|\mathbf{A}_{\mathcal{N}_m}\mathbf{z}'_{\mathcal{N}_m} \geq \mathbf{b}_{\mathcal{N}_m} - \mathbf{B}_{\mathcal{N}_m}\mathbf{z}''_{\mathcal{N}_m}\left(\overline{\mathbf{y}}\right)\right.\right\}. 
\end{equation}
of which the \emph{dual problem} can be formulated as the following LP problem: 
\vspace{-0.20in}
\normalsize
\begin{eqnarray}
\label{eq:dualobj}
\max && \left(\mathbf{b}_{\mathcal{N}_m} - \mathbf{B}_{\mathcal{N}_m}\mathbf{z}''_{\mathcal{N}_m}\left(\overline{\mathbf{y}}\right)\right)^{\top} \mathbf{u}_{\mathcal{N}_m} \\ \label{eq:dualconstr1}
\mathrm{s.t.} && \mathbf{A}_{\mathcal{N}_m}^{\top} \mathbf{u}_{\mathcal{N}_m}\leq \mathbf{c}'_{\mathcal{N}_m},~ \mathbf{u}_{\mathcal{N}_m} \geq \mathbf{0}. 
\end{eqnarray}
\normalsize
There are three cases of the dual problem: 

\vspace{0.03in}
\noindent \textbf{Case 1}: The optimal objective value is \textbf{unbounded}: By \emph{weak duality} \cite{Linear&Nonlinear}, $\overline{\mathbf{y}}$ does not satisfy 
$\mathbf{A}_{\mathcal{N}_m}\mathbf{z}'_{\mathcal{N}_m} \geq \mathbf{b}_{\mathcal{N}_m} - \mathbf{B}_{\mathcal{N}_m}\mathbf{z}''_{\mathcal{N}_m}\left(\overline{\mathbf{y}}\right)$ 
for any $\mathbf{z}'_{\mathcal{N}_m} \geq \mathbf{0}$. 
Since the dual problem is unbounded, there exists an \emph{extreme ray} $\tilde{\mathbf{u}}_{\mathcal{N}_m}$ subject to 
$\mathbf{A}_{\mathcal{N}_m}^{\top} \tilde{\mathbf{u}}_{\mathcal{N}_m} \leq \mathbf{0}$ and  \newline $\left(\mathbf{b}_{\mathcal{N}_m} - \mathbf{B}_{\mathcal{N}_m}\mathbf{z}''_{\mathcal{N}_m}\left(\overline{\mathbf{y}}\right)\right)^{\top} \tilde{\mathbf{u}}_{\mathcal{N}_m} > 0$. 
To ensure that $\tilde{\mathbf{u}}_{\mathcal{N}_m}$ won't be an extreme ray in the next iteration, $\mathrm{Sub}_m$ suggests a \emph{new cut} $h$ (\emph{feasibility cut}) to the MP: 
\vspace{-0.05in}
$$h:~ \left(\mathbf{b}_{\mathcal{N}_m} - \mathbf{B}_{\mathcal{N}_m}\mathbf{z}''_{\mathcal{N}_m}\left(\mathbf{y}\right)\right)^{\top} \tilde{\mathbf{u}}_{\mathcal{N}_m} \leq \mathbf{0}.$$

\vspace{0.00in}
\noindent \textbf{Case 2}: The optimal objective value is \textbf{bounded} with the solution $\overline{\mathbf{u}}_{\mathcal{N}_m}$:  By \emph{weak duality}, the optimal value of the dual problem is equal to the optimal value of $w_l$ constrained on the choice of $\overline{\mathbf{y}}$. In this case, $\mathrm{Sub}_m$ checks whether 
$\overline{w}_m < \left(\mathbf{b}_{\mathcal{N}_m} - \mathbf{B}_{\mathcal{N}_m}\mathbf{z}''_{\mathcal{N}_m}\left(\overline{\mathbf{y}}\right)\right)^{\top} \mathbf{u}_{\mathcal{N}_m}$. 
If yes, then
$\overline{w}_m < \min \left\{\mathbf{c}'_{\mathcal{N}_m}\mathbf{z}'_{\mathcal{N}_m}\left|\mathbf{A}_{\mathcal{N}_m}\mathbf{z}'_{\mathcal{N}_m} \geq \mathbf{b}_{\mathcal{N}_m} - \mathbf{B}_{\mathcal{N}_m}\mathbf{z}''_{\mathcal{N}_m}\left(\overline{\mathbf{y}}\right)\right.\right\}$, 
meaning that $\overline{w}_m$ derived by the MP is lower than the minimum cost. Therefore, $\mathrm{Sub}_m$ suggests a \emph{new cut} 
$$h:~ w_m \geq \left(\mathbf{b}_{\mathcal{N}_m} - \mathbf{B}_{\mathcal{N}_m}\mathbf{z}''_{\mathcal{N}_m}\left(\mathbf{y}\right)\right)^{\top} \overline{\mathbf{u}}_{\mathcal{N}_m}$$ 
to the MP to improve $w_m$ in the next iteration. 

\vspace{0.00in}
\noindent \textbf{Case 3}: There is \textbf{no feasible solution}: By \emph{weak duality}, the primal problem either has no feasible/unbounded solution. The algorithm terminates.

After adding the new cuts (from all the subproblems) to the cut set $\mathcal{H}$, the BD moves to the next iteration by recalculating the MP and obtaining updated $\left\{\overline{y}_{1}, ..., \overline{y}_{K}, \overline{w}_1, ..., \overline{w}_M\right\}$. As \textbf{Stage 1} and \textbf{Stage 2} are repeated over iterations, the MP collects more cuts from the subproblems, 
converging the solution $\left\{\overline{y}_{1}, ..., \overline{y}_{K}, \overline{w}_1, ..., \overline{w}_M\right\}$ to the optimal.

\DEL{
\begin{figure}[t]
\centering
\begin{minipage}{0.23\textwidth}
\centering
  \subfigure[Computation time]{
\includegraphics[width=1.00\textwidth, height = 0.13\textheight]{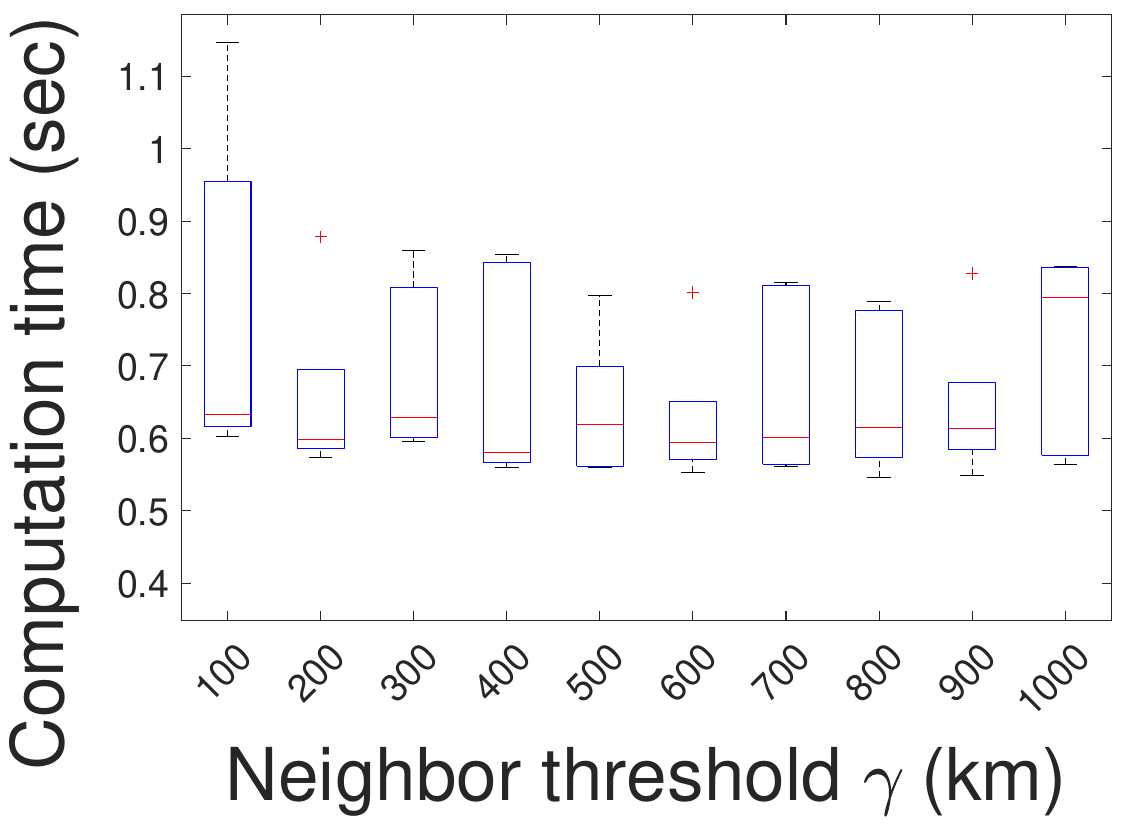}}
\label{}
\end{minipage}
\hspace{0.05in}
\begin{minipage}{0.23\textwidth}
\centering
  \subfigure[Cost]{
\includegraphics[width=1.00\textwidth, height = 0.13\textheight]{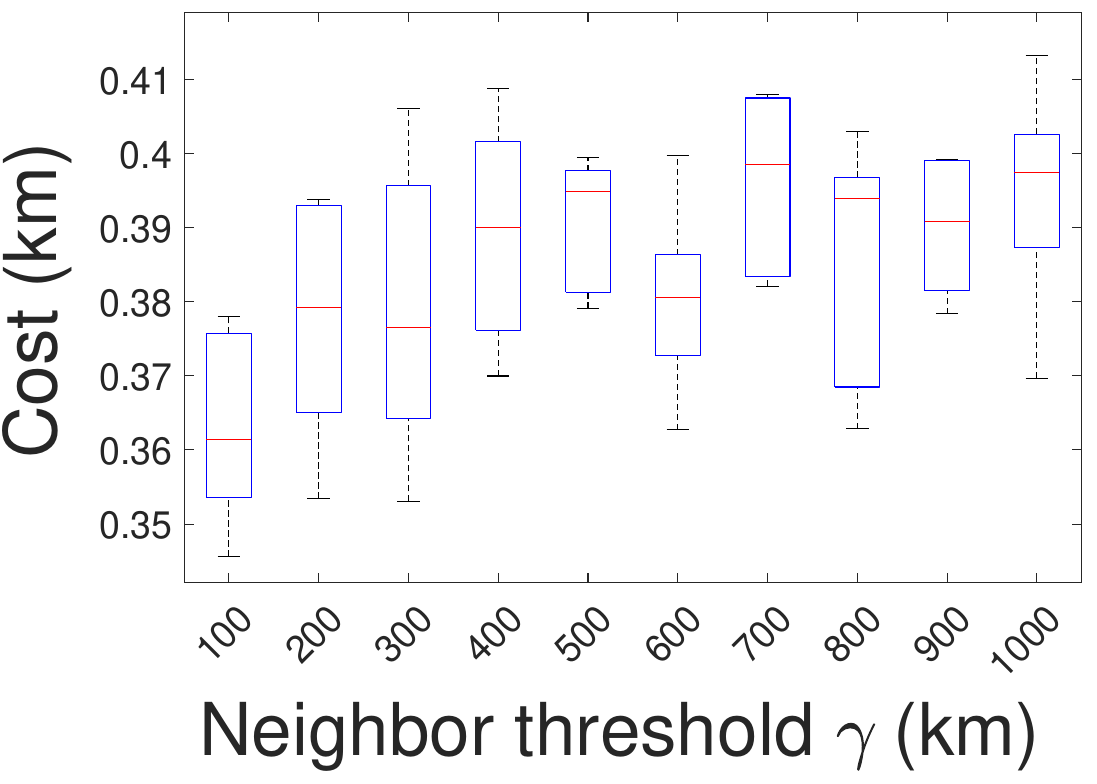}}
\label{}
\end{minipage}
\vspace{-0.10in}
\caption{Performance of LR-Geo with different neighbor threshold $\gamma$.}
\label{fig:exp:neighbor}
\vspace{-0.10in}
\end{figure}}

\end{document}